\documentclass[12pt,  a4paper]{article}
\usepackage[top=1in, bottom=1.5in,left=1in,right=1in]{geometry}
\usepackage{amsmath,amssymb,amsfonts, amsthm}
\usepackage{upgreek}
\usepackage{physics}
\usepackage{float}
\usepackage{flafter}
\usepackage{bm}
\usepackage{dsfont}
\usepackage{tabularx}
\usepackage[final, hidelinks]{hyperref}
\usepackage{diagbox}
\hypersetup{allcolors=black}
\usepackage{booktabs} 
\usepackage{enumitem} 
\usepackage{placeins} 
\usepackage{afterpage} 
\usepackage{needspace}
\usepackage{caption} 
\usepackage{siunitx}

\usepackage{longtable}

\usepackage[usenames,dvipsnames]{xcolor}
\usepackage{tikz} \usetikzlibrary{calc, arrows.meta, intersections, patterns, positioning, shapes.misc, fadings, through,decorations.pathreplacing}
\definecolor{ColorOne}{named}{MidnightBlue}
\definecolor{ColorTwo}{named}{Dandelion}
\definecolor{ColorThree}{named}{Plum}

\newtheorem{lemma}{Lemma}[section]
\newtheorem{proposition}{Proposition}[section]
\newtheorem{assumption}{Assumption}[section]

\theoremstyle{remark}

\usepackage[ruled]{algorithm2e}

\DeclareMathOperator*{\argmax}{argmax}

\renewcommand{\epsilon}{\varepsilon}

\usepackage[round]{natbib}
\usepackage[utf8]{inputenc}
\usepackage[toc]{appendix}
\title{Estimation of Rational Inattention--Additive Random Utility Models}

\author{Ruolei Zhu\thanks{Department of Economics, University College London. E-mail: \href{mailto:ruolei.zhu.19@ucl.ac.uk}{ruolei.zhu.19@ucl.ac.uk}.}}
\date{October 12, 2026\\[0.4em]
{\small\emph{\href{https://drive.google.com/file/d/1qUFOCWYNU9siIjDHhSKut9S7MYlnr6a2/view?usp=sharing}{Click here for the latest version}}}}

\begin{document}
\maketitle
\thispagestyle{empty}

\begin{abstract}  
This paper develops likelihood-based estimation for parametric rational inattention--additive random utility models. We establish consistency and asymptotic normality of the exact maximum likelihood estimator and provide conditions under which the simulated estimator shares its asymptotic distribution. To implement estimation, we develop Newton--BA, a KKT-guided hybrid of Newton--Raphson and Blahut--Arimoto (BA). The algorithm accommodates both interior and boundary solutions and substantially improves computational efficiency relative to direct BA while preserving numerical accuracy. Monte Carlo simulations show that parameter recovery improves with sample size and that confidence-interval coverage is generally close to the nominal $95$ percent level. Finally, we compare our model with no- and perfect-information logit benchmarks in an empirically calibrated ModeCanada experiment. Our model achieves similar in-sample likelihood fit to perfect-information logit but more accurately recovers key preference parameters, choice probabilities, and responses to travel-time changes.

\end{abstract}

\clearpage
\setcounter{page}{1}

\section{Introduction}
\label{section:introduction}

Rational-inattention models describe decision-making under limited information-processing capacity \citep{sims2003implications}. In discrete-choice settings, optimal acquisition of costly information generates stochastic choice probabilities \citep{caplin2015revealed,matvejka2015rational,fosgerau2020discrete}. Bringing these models to data, however, presents two related computational challenges. The conditional choice probabilities (CCPs) entering the likelihood require numerical integration over unobserved payoff components. In addition, the underlying rational-inattention problem must be solved repeatedly across individuals at each likelihood evaluation, making estimation computationally demanding. Building on the nonparametric rational inattention--additive random utility models (RI--ARUMs) of \citet{kristensenzhu2026RI}, this paper studies simulated maximum likelihood estimation under a parametric payoff specification. We establish the estimator's large-sample properties and provide a practical computational implementation.

In our model, payoffs for each decision maker (DM) $i$ consist of two components, $V_{0,i}$ and $V_{1,i}$. The DM observes $V_{0,i}$ before information acquisition but remains uncertain about $V_{1,i}$ and holds a subjective prior over its possible realizations. We parameterize the deterministic parts of both components as linear functions of alternative-specific and individual-specific covariates, collected in $X_i$ and $W_i$, respectively, with coefficient vector $\theta:=(\alpha^{\prime},\beta^{\prime},\gamma^{\prime},\delta^{\prime})^{\prime}$. The component $V_{0,i}$ also includes a taste shifter $\varepsilon_{0,i}$, while $V_{1,i}$ includes a payoff shock $\varepsilon_{1,i}$. This specification retains the information structure of the nonparametric RI--ARUMs while providing a flexible parametric framework for empirical analysis.

For estimation, consider a random sample $\{(Y_i,X_i,W_i)\}_{i=1}^N$. The Econometrician observes choices and covariates, but not $(\varepsilon_{0,i},\varepsilon_{1,i})$. Conditional on the realized payoff vectors $(v_{0,i},v_{1,i})$, the DM's rational-inattention problem determines the optimal choice probabilities $P_j^*(v_{1,i}\mid v_{0,i})$. These probabilities have a logit-like exponential form. Unlike standard multinomial logit probabilities, however, their baseline weights are the DM's unconditional choice probabilities, determined endogenously by the information-acquisition problem.

The likelihood uses the CCPs $\Phi_j(x_i,w_i;\theta):=\Pr_\theta(Y_i=j\mid X_i=x_i,W_i=w_i)$, obtained by integrating $P_j^*(v_{1,i}\mid v_{0,i})$ over the unobserved payoff components $(\varepsilon_{0,i},\varepsilon_{1,i})$. For estimation, we specify the DM's subjective prior and the population distributions of these components, treating them as known. The exact maximum likelihood estimator (MLE), denoted by $\widehat\theta_N$, maximizes the sample likelihood constructed from these CCPs.

Existing empirical applications obtain tractable choice probabilities for estimation in different ways. \citet{bertoli2020rational} and \citet{brown2020endogenous} derive closed-form expressions under specific distributional assumptions. \citet{joo2023rational} and \citet{natan2025choice} parameterize the unconditional choice probabilities entering the choice rule through consideration shifters, while \citet{habib2023rational} uses aggregate market shares to measure these probabilities. Our approach instead uses numerical methods to evaluate the likelihood without requiring closed-form probabilities and computes the unconditional choice probabilities directly from the DM's maximization problem.

To implement estimation, we use a nested procedure similar to the nested fixed-point (NFXP) algorithm of \citet{Rust1987}. Nested maximum likelihood estimation is also used in the multivariate RI model of \citet{chen2025multivariate}. At each candidate parameter vector, our inner layer solves the rational-inattention problem rather than a dynamic-programming fixed point. The resulting choice probabilities are integrated to evaluate the likelihood, which the outer layer maximizes over $\theta$.

Within this nested procedure, we address the first computational challenge by approximating the prior expectation in the DM's problem and the likelihood integrations with finite Monte Carlo averages. Maximizing the resulting simulated likelihood yields the estimator implemented in this paper, following the simulation-based estimation literature \citep[see, e.g.,][]{Lee1992,GourierouxMonfort1996}.

To establish the statistical properties of this estimator, our asymptotic analysis uses the exact MLE as a theoretical benchmark. We establish continuity and identification under model-specific conditions and combine these results with standard MLE regularity assumptions to obtain consistency and asymptotic normality. We then give high-level conditions controlling approximation error in the simulated likelihood and its derivatives. These conditions account for both the simulated inner problem and the likelihood integrations and ensure that the simulated MLE is consistent and has the same asymptotic distribution as the exact MLE.

To address the second computational challenge, we develop an efficient solver for the rational-inattention problem. A key computational consideration is that this problem may have boundary solutions in which some alternatives have zero unconditional choice probability. We therefore develop Newton--BA, which combines efficient Newton--Raphson iterations in the interior with the boundary-capable Blahut--Arimoto (BA) algorithm. The algorithm starts with Newton--Raphson and switches once to BA near the boundary. The Karush--Kuhn--Tucker (KKT) optimality conditions guide the choice of the BA starting point. Numerical comparisons show that Newton--BA substantially improves computational efficiency relative to direct BA while preserving numerical accuracy.

Monte Carlo simulations show that parameter recovery improves with sample size. Coverage of the information-based confidence intervals is generally close to the nominal level, although modest undercoverage remains for some parameters.

Finally, we compare the RI--ARUM with no- and perfect-information logit benchmarks in an empirically calibrated ModeCanada experiment. Our model and perfect-information logit achieve similar in-sample likelihood fit. However, our model more closely recovers preferences over uncertain travel time than perfect-information logit. It also more accurately recovers choice probabilities and their responses to reduced travel time than either benchmark.

The remainder of the paper is organized as follows. Section~\ref{section:parametric_model} presents the parametric RI--ARUM and the DM's information problem. Section~\ref{section:Estimation} defines the exact and simulated maximum likelihood estimators. Section~\ref{sec:asymptotic-theory} establishes their large-sample properties. Section~\ref{sec:computation} develops the numerical implementation and compares inner solvers. Section~\ref{sec:simulation-results} evaluates finite-sample performance through Monte Carlo simulations. Section~\ref{sec:modecanada-experiment} presents the calibrated ModeCanada experiment and compares the RI--ARUM with no- and perfect-information logit benchmarks. Section~\ref{section:conclusion} concludes.

\section{The Parametric RI--ARUM}
\label{section:parametric_model}

This section specializes the nonparametric RI--ARUM framework developed by \citet{kristensenzhu2026RI} to the parametric specification used for estimation. In that framework, the payoff vectors over the inside alternatives satisfy $V_0=u_0(X_0,W_0)+\varepsilon_0$ and $V_1=u_1(X_1,W_1)+\varepsilon_1$. We retain this payoff decomposition and specify $u_0(\cdot)$ and $u_1(\cdot)$ as linear functions of the covariates. We first introduce the resulting payoff specification and state space, and then formulate the DM's rational-inattention problem.

Let $\mathcal J=\{0,\ldots,J\}$ denote the finite choice set, with outside alternative $j=0$ and inside alternatives $j=1,\ldots,J$. For each DM $i$, let $V_{0,i}:=(V_{0,i1},\ldots,V_{0,iJ})^{\prime}$ and $V_{1,i}:=(V_{1,i1},\ldots,V_{1,iJ})^{\prime}$ denote the random payoff vectors over the inside alternatives. The coefficient vectors $\alpha$, $\beta$, $\gamma$, and $\delta$ have dimensions matching their respective covariates. We collect them in $\theta:=(\alpha^{\prime},\beta^{\prime},\gamma^{\prime},\delta^{\prime})^{\prime}\in\Theta\subseteq\mathbb R^{d_\theta}$. For each inside alternative $j=1,\ldots,J$, the payoff components are
\[
V_{0,ij}
=
\alpha^{\prime}X_{0,ij}
+
\beta^{\prime}W_{0,i}
+
\varepsilon_{0,ij},
\qquad
V_{1,ij}
=
\gamma^{\prime}X_{1,ij}
+
\delta^{\prime}W_{1,i}
+
\varepsilon_{1,ij}.
\]
The outside alternative is treated separately from the inside payoff vectors. Following Lemma~2.3 of \citet{kristensenzhu2026RI}, we normalize $V_{0,i0}=0$ and $V_{1,i0}=0$.

For each $t\in\{0,1\}$, both $X_{t,ij}$ and $W_{t,i}$ may be scalar- or vector-valued, with the scalar case included by setting the corresponding dimension equal to one. For each inside alternative $j=1,\ldots,J$, let $X_{t,ij}:=(X_{t,ij,1},\ldots,X_{t,ij,d_{x_{t,j}}})^{\prime}\in\mathcal X_{t,j}\subseteq\mathbb R^{d_{x_{t,j}}}$ denote the alternative-specific covariate of DM $i$, with realization $x_{t,ij}$. Because the same coefficient vector applies across the inside alternatives, $d_{x_{t,j}}$ is constant across $j$ for each $t$. Define $\mathcal X_t:=\prod_{j=1}^J\mathcal X_{t,j}$ and $d_{x_t}:=\sum_{j=1}^Jd_{x_{t,j}}$, and collect the alternative-specific covariates in $X_{t,i}:=(X_{t,i1}^{\prime},\ldots,X_{t,iJ}^{\prime})^{\prime}\in\mathcal X_t\subseteq\mathbb R^{d_{x_t}}$, with realization $x_{t,i}$. Let $W_{t,i}:=(W_{t,i,1},\ldots,W_{t,i,d_{w_t}})^{\prime}\in\mathcal W_t\subseteq\mathbb R^{d_{w_t}}$ denote the individual-specific covariate common across the inside alternatives, with realization $w_{t,i}$. Finally, define $\varepsilon_{t,i}:=(\varepsilon_{t,i1},\ldots,\varepsilon_{t,iJ})^{\prime}\in\mathcal E_t\subseteq\mathbb R^J$ as the unobserved payoff vector for component $t$, with realization $e_{t,i}:=(e_{t,i1},\ldots,e_{t,iJ})^{\prime}$.

The timing is as follows. At $t=0$, the DM observes $V_{0,i}=v_{0,i}$ but remains uncertain about $V_{1,i}$. She holds a subjective prior over the possible realizations of $V_{1,i}$, chooses a costly signal rule, observes a signal $S_i=s_i$, updates her beliefs, and chooses the alternative that maximizes posterior expected payoff. At $t=1$, the state underlying $V_{1,i}$ is realized and determines the time-$1$ payoffs. Accordingly, $\varepsilon_{0,i}$ is known to the DM at $t=0$ and is interpreted as a taste shifter, whereas $\varepsilon_{1,i}$ remains uncertain at $t=0$ and is interpreted as a payoff shock.

To represent the uncertainty about $V_{1,i}$, let $\Omega_1:=\mathcal X_1\times\mathcal W_1\times\mathcal E_1$ denote the state space. For each DM $i$, let $\omega_{1,i}:=(x_{1,i},w_{1,i},e_{1,i})\in\Omega_1$ denote a possible state of the world, where $x_{1,i}\in\mathcal X_1$, $w_{1,i}\in\mathcal W_1$, and $e_{1,i}\in\mathcal E_1$ are possible realizations of $X_{1,i}$, $W_{1,i}$, and $\varepsilon_{1,i}$, respectively. Throughout this section, $\theta$ is held fixed and its dependence is suppressed from the payoff notation. Let $v_{1,i}(\omega_{1,i})$ denote the payoff vector induced by $\omega_{1,i}$, with components $v_{1,ij}(\omega_{1,i})=\gamma^{\prime}x_{1,ij}+\delta^{\prime}w_{1,i}+e_{1,ij}$ for $j=1,\ldots,J$. Let $\pi_1$ denote the common subjective prior over $\Omega_1$, assumed invariant to $v_{0,i}$, so that $\pi_1(\cdot\mid v_{0,i})=\pi_1(\cdot)$. Although DMs share a common prior, their realized states may differ.

Let $\mathcal S$ denote the signal space of $S_i$, and let $\Delta(\mathcal S)$ denote the set of probability measures on $\mathcal S$. Given $v_{0,i}$, a signal rule maps each possible payoff vector $v_{1,i}$ into a conditional signal distribution $\kappa(\cdot\mid v_{0,i},v_{1,i})\in\Delta(\mathcal S)$. Under a fixed tie-breaking rule,\footnote{If several alternatives yield the same maximum posterior expected payoff, the tie-breaking rule selects one of them according to a predetermined ordering. This ensures that the choice rule is single-valued.} the choice induced by signal realization $s_i$ is
\[
a(s_i;v_{0,i})
\in
\arg\max_{\ell\in\mathcal J}
\mathbb E\!\left[
v_{0,i\ell}+V_{1,i\ell}
\mid
S_i=s_i,v_{0,i}
\right].
\]

For each $j\in\mathcal J$, the indicator $\mathbf 1\{a(s_i;v_{0,i})=j\}$ indicates whether the choice rule selects alternative $j$, and the set $\mathcal S_j(v_{0,i}):=\{s_i\in\mathcal S:a(s_i;v_{0,i})=j\}$ collects the signal realizations that lead to alternative $j$. The signal rule and choice rule induce the CCP
\[
\begin{aligned}
P_j(v_{1,i}\mid v_{0,i})
&=
\int_{\mathcal S}
\mathbf 1\{a(s_i;v_{0,i})=j\}
\,\kappa(ds_i\mid v_{0,i},v_{1,i}) \\
&=
\int_{\mathcal S_j(v_{0,i})}
\kappa(ds_i\mid v_{0,i},v_{1,i}) \\
&=
\kappa\!\left(
\mathcal S_j(v_{0,i})
\mid
v_{0,i},v_{1,i}
\right),
\qquad j\in\mathcal J.
\end{aligned}
\]
Thus, $P_j(v_{1,i}\mid v_{0,i})$ is the probability of choosing alternative $j$, conditional on $(v_{0,i},v_{1,i})$, after integrating over the signals that may be received in that state.

Lemma~1 of \citet{matvejka2015rational} shows that signals leading to the same alternative can be pooled without changing expected payoff or increasing the Shannon information cost. After pooling, each alternative $j\in\mathcal J$ can be represented by a single signal. Relabelling each pooled signal by the corresponding alternative allows us to take $\mathcal S_j(v_{0,i})=\{j\}$ for each $j\in\mathcal J$ and hence $\mathcal S=\mathcal J$. The induced CCP is therefore\footnote{After pooling and relabelling, the signal realization $s_i=j$ directly recommends alternative $j$, and the DM follows this recommendation. Hence, $\Pr(\text{choose }j\mid v_{0,i},v_{1,i})=\Pr(S_i=j\mid v_{0,i},v_{1,i})$, which implies $P_j(v_{1,i}\mid v_{0,i})=\kappa(\{j\}\mid v_{0,i},v_{1,i})$.}
\[
P_j(v_{1,i}\mid v_{0,i})
=
\kappa\!\left(
\{j\}
\mid
v_{0,i},v_{1,i}
\right),
\qquad j\in\mathcal J.
\]
Thus, in the pooled representation, the signal rule is represented by its induced CCP rule, and both expected payoff and the Shannon information cost can be rewritten in terms of the induced CCPs. Corollary~1 of \citet{matvejka2015rational} then establishes that maximizing over feasible CCP rules is equivalent to the original signal-based problem. Accordingly, the DM problem below is written directly in terms of CCPs.

Let $P(v_{1,i}\mid v_{0,i}):=\bigl(P_1(v_{1,i}\mid v_{0,i}),\ldots,P_J(v_{1,i}\mid v_{0,i})\bigr)^{\prime}$ denote the vector of CCPs for the inside alternatives. The corresponding integrated CCPs, also called unconditional choice probabilities, are collected in $P^0(v_{0,i}):=\bigl(P_1^0(v_{0,i}),\ldots,P_J^0(v_{0,i})\bigr)^{\prime}$.\footnote{The integrated CCPs average $P(v_{1,i}\mid v_{0,i})$ over possible realizations of $V_{1,i}$ using the DM's prior. They are called ``unconditional'' because they no longer condition on a particular realization of $V_{1,i}$, although $v_{0,i}$ remains given.} For each inside alternative $j=1,\ldots,J$, the Bayes-consistency condition is
\[
P_j^0(v_{0,i})
=
\int_{\Omega_1}
P_j\!\left(
v_{1,i}(\omega_{1,i})
\mid v_{0,i}
\right)
\,d\pi_1(\omega_{1,i}).
\]
The outside probabilities are recovered residually as $P_0(v_{1,i}\mid v_{0,i})=1-\sum_{j=1}^{J}P_j(v_{1,i}\mid v_{0,i})$ and $P_0^0(v_{0,i})=1-\sum_{j=1}^{J}P_j^0(v_{0,i})$. With these residual definitions, the Bayes-consistency condition extends to the outside alternative and hence to all $j\in\mathcal J$.

The information-cost parameter $\lambda>0$ measures the cost of one unit of Shannon information. We treat $\lambda$ as a known constant. For each fixed $v_{0,i}$, the equivalent choice-probability formulation of the DM's problem is
\[
\max_{P(\cdot\mid v_{0,i})\in\mathcal P(v_{0,i})}
Q(P;v_{0,i}),
\]
where $\mathcal P(v_{0,i})$ is the set of feasible conditional probability rules satisfying $P_j(v_{1,i}(\omega_{1,i};\theta)\mid v_{0,i})\geq 0$ for $j=1,\ldots,J$, $\sum_{j=1}^{J}P_j(v_{1,i}(\omega_{1,i};\theta)\mid v_{0,i})\leq 1$ for $\pi_1$-almost every $\omega_{1,i}\in\Omega_1$, and the consistency condition above.\footnote{Although the DM's choice problem is defined over choice probabilities for all alternatives $j\in\mathcal J$, we represent feasible rules by their inside components, with the outside probabilities recovered residually as defined above.} The objective is
\[
Q(P;v_{0,i})
:=
\int_{\Omega_1}
\sum_{j\in\mathcal J}
P_j\!\left(
v_{1,i}(\omega_{1,i};\theta)
\mid v_{0,i}
\right)
\left[
v_{0,ij}
+
v_{1,ij}(\omega_{1,i};\theta)
\right]
\,d\pi_1(\omega_{1,i})
-
\lambda\mathcal I(P;v_{0,i}),
\]
where the Shannon information cost is\footnote{In the Shannon information cost, we use the convention $0\ln 0=0$.}
\[
\mathcal I(P;v_{0,i})
:=
-\sum_{j\in\mathcal J}
P_j^0(v_{0,i})\ln P_j^0(v_{0,i})
+
\int_{\Omega_1}
\sum_{j\in\mathcal J}
P_j\!\left(
v_{1,i}(\omega_{1,i};\theta)
\mid v_{0,i}
\right)
\ln
P_j\!\left(
v_{1,i}(\omega_{1,i};\theta)
\mid v_{0,i}
\right)
\,d\pi_1(\omega_{1,i}).
\]

Applying Theorem~1 of \citet{matvejka2015rational}, for each fixed $v_{0,i}$ and realized $v_{1,i}$, the solution to the DM's problem yields the following optimal CCPs for the inside alternatives:
\begin{equation}
P_j^*(v_{1,i}\mid v_{0,i})
=
\frac{
P_j^{0*}(v_{0,i})
\exp\!\left((v_{0,ij}+v_{1,ij})/\lambda\right)
}{
\sum_{\ell\in\mathcal J}
P_\ell^{0*}(v_{0,i})
\exp\!\left((v_{0,i\ell}+v_{1,i\ell})/\lambda\right)
},
\qquad j\in\{1,\ldots,J\}.
\label{eq:logitccp}
\end{equation}
Equation~\eqref{eq:logitccp} also applies when the optimal integrated CCP vector lies on the boundary of the simplex. In particular, if $P_j^{0*}(v_{0,i})=0$ for some alternative $j$, then the numerator in Equation~\eqref{eq:logitccp} is zero and therefore $P_j^*(v_{1,i}\mid v_{0,i})=0$ for every realization of $v_{1,i}$. Thus, unlike a standard full-support logit rule, the rational-inattention model allows alternatives to receive zero choice probability.

Equation~\eqref{eq:logitccp} gives the CCPs conditional on realized payoff vectors. Hence, for DM $i$'s choice outcome $Y_i\in\mathcal J$,
\[
\Pr(Y_i=j\mid V_{0,i}=v_{0,i},V_{1,i}=v_{1,i})
=
P_j^*(v_{1,i}\mid v_{0,i}),
\qquad j\in\{1,\ldots,J\}.
\]
The outside CCP is recovered residually. Once $\omega_{1,i}$ is realized, Equation~\eqref{eq:logitccp} is evaluated at the corresponding payoff vector $v_{1,i}(\omega_{1,i})$.

Equation~\eqref{eq:logitccp} expresses the model-implied CCPs in terms of the unknown integrated CCP vector $P^{0*}(v_{0,i})$. Bayes consistency, in turn, requires this vector to equal the prior average of $P^*(\cdot\mid v_{0,i})$. The two are therefore jointly determined, so the formula does not provide an explicit solution state by state. We therefore use an indirect-function approach: first solve the reduced problem for $P^{0*}(v_{0,i})$, then recover the model-implied CCPs from Equation~\eqref{eq:logitccp}.

Combining the optimal CCP characterization in Equation~\eqref{eq:logitccp} with Bayes consistency yields the reduced problem $\max_{P^0(v_{0,i})\in\Delta(\mathcal J)}Q^0(P^0;v_{0,i})$, where $\Delta(\mathcal J):=\{p\in\mathbb R^J:p_j\geq 0,\ j\in\{1,\ldots,J\},\ \text{and}\ \sum_{j=1}^Jp_j\leq 1\}$ and
\begin{equation}
Q^0(P^0;v_{0,i})
:=
\lambda
\int_{\Omega_1}
\ln\!\left(
\sum_{j\in\mathcal J}
P_j^0(v_{0,i})
\exp\!\left((v_{0,ij}+v_{1,ij}(\omega_{1,i}))/\lambda\right)
\right)
\,d\pi_1(\omega_{1,i}).
\label{eq:parametric_Q0}
\end{equation}
In Equation~\eqref{eq:parametric_Q0}, the outside integrated probability is defined residually by $P_0^0(v_{0,i})=1-\sum_{j=1}^J P_j^0(v_{0,i})$.

Since $Q^0(\cdot;v_{0,i})$ is concave in $P^0(v_{0,i})$ and the feasible set is convex, the KKT conditions are necessary and sufficient for optimality. Following Proposition~1 of \citet{caplin2019rational}, the optimal integrated CCPs satisfy
\[
\int_{\Omega_1}
\frac{
\exp\!\left((v_{0,ij}+v_{1,ij}(\omega_{1,i}))/\lambda\right)
}{
\sum_{\ell\in\mathcal J}
P_\ell^{0*}(v_{0,i})
\exp\!\left((v_{0,i\ell}+v_{1,i\ell}(\omega_{1,i}))/\lambda\right)
}
\,d\pi_1(\omega_{1,i})
\leq 1,
\qquad j\in\mathcal J,
\]
with equality whenever $P_j^{0*}(v_{0,i})>0$. In the denominator, the outside integrated probability is defined residually.

However, $Q^0(\cdot;v_{0,i})$ is not strictly concave in general.\footnote{Strict concavity can fail because $P^0(v_{0,i})$ enters $Q^0$ only through the weighted sum inside the logarithm. For example, if two alternatives $a$ and $b$ satisfy $v_{0,ia}+v_{1,ia}(\omega_{1,i})=v_{0,ib}+v_{1,ib}(\omega_{1,i})$ for $\pi_1$-almost every $\omega_{1,i}$, then reallocating probability mass between $a$ and $b$ leaves $Q^0$ unchanged.} Hence, although the KKT conditions characterize optimality, they do not guarantee a unique maximizer. We therefore impose the following sufficient condition based on affine independence of the transformed payoff functions.

\begin{assumption}
\label{ass:affineindependence}
For every $\theta\in\Theta$ and almost every realization of $(X_{0,i},W_{0,i},\varepsilon_{0,i})$, define $r_{ij}(\omega_{1,i};\theta):=\exp((v_{0,ij}(\theta)+v_{1,ij}(\omega_{1,i};\theta))/\lambda)$ for each $j\in\mathcal J$, where $v_{0,i0}(\theta)=v_{1,i0}(\omega_{1,i};\theta)=0$. The functions $r_{ij}(\cdot;\theta)$, $j\in\mathcal J$, belong to $L^1(\Omega_1,\pi_1)$ and are affinely independent.\footnote{Affine independence means that $\sum_{j\in\mathcal J}a_jr_{ij}(\omega_{1,i};\theta)=0$ $\pi_1$-almost surely and $\sum_{j\in\mathcal J}a_j=0$ imply $a_j=0$ for every $j\in\mathcal J$. Here, $L^1(\Omega_1,\pi_1)$ is the space of $\pi_1$-integrable functions, with functions equal $\pi_1$-almost everywhere treated as identical.}
\end{assumption}

\begin{lemma}
\label{lem:uniqueccp}
Under Assumption~\ref{ass:affineindependence}, the reduced problem in Equation~\eqref{eq:parametric_Q0} admits a unique solution $P^{0*}(v_{0,i}(\theta))$ for every $\theta\in\Theta$ and almost every realization of $(X_{0,i},W_{0,i},\varepsilon_{0,i})$.
\end{lemma}

Lemma~\ref{lem:uniqueccp} ensures that the integrated CCP vector $P^{0*}(v_{0,i}(\theta))$ is single-valued. Given this vector, Equation~\eqref{eq:logitccp} uniquely determines the model-implied CCPs $P_j^*(v_{1,i}(\theta)\mid v_{0,i}(\theta))$ for every $j\in\mathcal J$. Hence, the model-implied CCPs are well defined.

\section{Estimation}
\label{section:Estimation}

This section defines the exact and simulated maximum likelihood estimators. We first construct the exact CCPs and likelihood, then use simulation draws to approximate the inner problem and the likelihood integrations.

\subsection{Exact maximum likelihood}
\label{subsec:exact-mle}

Let $(Y,X,W)$ denote the population random vector, where $Y\in\mathcal J$, $X:=(X_0^{\prime},X_1^{\prime})^{\prime}$, and $W:=(W_0^{\prime},W_1^{\prime})^{\prime}$. The Econometrician observes realizations of $(Y,X,W)$, whereas the unobserved components $(\varepsilon_0,\varepsilon_1)$ remain latent. Under the assumptions below, the structural model determines the CCPs entering the exact likelihood.

Throughout this section, a distribution described as known is taken as given by the Econometrician for estimation. The following three assumptions specify the independence, distributional, and prior-belief restrictions used to construct the likelihood:

\begin{assumption}
\label{ass:independence}
Under the population distribution, $(X,W)$, $\varepsilon_0$, and $\varepsilon_1$ are mutually independent.
\end{assumption}

The independence of $(X,W)$ and $(\varepsilon_0,\varepsilon_1)$ follows \citet[Assumption~3.1]{kristensenzhu2026RI}. To simplify likelihood construction, we additionally assume $\varepsilon_0\perp\varepsilon_1$, so their joint distribution factorizes as $F_{\varepsilon_0}\otimes F_{\varepsilon_1}$.

\begin{assumption}
\label{ass:epsdist}
The population distributions of $\varepsilon_0$ and $\varepsilon_1$, denoted by $F_{\varepsilon_0}$ and $F_{\varepsilon_1}$, respectively, are known. The distribution $F_{\varepsilon_0}$ has full support on $\mathbb R^J$.
\end{assumption}

Requiring these distributions to be known is an additional restriction relative to the identification framework of \citet{kristensenzhu2026RI}. We retain the full-support condition considered in their analysis. Throughout, $F$ denotes a data-generating distribution, whereas $\pi$ denotes the DM's subjective prior.

\begin{assumption}
\label{ass:knownprior}
The subjective prior $\pi_1$ over $\Omega_1$ is known.
\end{assumption}

The prior may be misspecified and need not coincide with the population distribution of $(X_1,W_1,\varepsilon_1)$.

Fix a candidate parameter vector $\theta$ and covariate realizations $(x,w)$, where $x=(x_0,x_1)$ and $w=(w_0,w_1)$. For given realizations $(e_0,e_1)$, let $v_0(x_0,w_0,e_0;\theta)$ and $v_1(x_1,w_1,e_1;\theta)$ denote the corresponding payoff vectors. Equation~\eqref{eq:logitccp} gives the model-implied CCPs conditional on these payoffs. Integrating these probabilities with respect to $F_{\varepsilon_0}$ and $F_{\varepsilon_1}$ yields the CCPs $\Phi_j(x,w;\theta)$ entering the likelihood, for $j=1,\ldots,J$.

The integration formula below follows from the law of total probability and Assumption~\ref{ass:independence}.\footnote{Let $\varepsilon=(\varepsilon_0,\varepsilon_1)$ and $e=(e_0,e_1)$. The law of total probability gives $\Pr_\theta(Y=j\mid X=x,W=w)=\int_{\mathcal E_0\times\mathcal E_1}\Pr_\theta(Y=j\mid X=x,W=w,\varepsilon=e)\,dF_{\varepsilon\mid X,W}(e\mid x,w)$. Under Assumption~\ref{ass:independence}, the conditional measure factorizes as $dF_{\varepsilon\mid X,W}(e\mid x,w)=dF_{\varepsilon_0}(e_0)\,dF_{\varepsilon_1}(e_1)$.} Assumption~\ref{ass:epsdist} ensures that $F_{\varepsilon_0}$ and $F_{\varepsilon_1}$ are known to the Econometrician, whereas Assumption~\ref{ass:knownprior} ensures that the subjective prior entering the DM's inner problem is known. Therefore, for each candidate parameter vector $\theta$ and inside alternative $j=1,\ldots,J$, the CCP is
\[
\begin{aligned}
\Phi_j(x,w;\theta)
&:=
\Pr_\theta(Y=j\mid X=x,W=w) \\
&=
\int_{\mathcal E_0\times\mathcal E_1}
P_j^*\!\left(
v_1(x_1,w_1,e_1;\theta)
\mid v_0(x_0,w_0,e_0;\theta)
\right)
\,dF_{\varepsilon_0}(e_0)
\,dF_{\varepsilon_1}(e_1).
\end{aligned}
\]
Let $\Phi(x,w;\theta):=(\Phi_1(x,w;\theta),\ldots,\Phi_J(x,w;\theta))$ denote the inside CCP vector. Since the outside model-implied CCP is defined residually, linearity of integration gives $\Phi_0(x,w;\theta)=1-\sum_{j=1}^{J}\Phi_j(x,w;\theta)$.

By Lemma~\ref{lem:uniqueccp}, the affine-independence condition ensures that $P^{0*}\bigl(v_0(x_0,w_0,e_0;\theta)\bigr)$ is single-valued. Equation~\eqref{eq:logitccp} then uniquely determines the inside model-implied CCPs $P_j^*\bigl(v_1(x_1,w_1,e_1;\theta)\mid\allowbreak v_0(x_0,w_0,e_0;\theta)\bigr)$ for $j=1,\ldots,J$, while the outside model-implied CCP is recovered residually. Hence the CCPs $\Phi_j(x,w;\theta)$ obtained by integrating the model-implied CCPs over $(\varepsilon_0,\varepsilon_1)$ are well defined. We maintain the affine-independence condition throughout the likelihood analysis.

Let $\{(Y_i,X_i,W_i)\}_{i=1}^{N}$ be a random sample from the distribution of $(Y,X,W)$. For each $i=1,\ldots,N$, the log-likelihood contribution of DM $i$ at a candidate parameter vector $\theta$ is
\[
\ell_i(\theta)
=
\sum_{j\in\mathcal J}
\mathbf 1\{Y_i=j\}
\ln\Phi_j(X_i,W_i;\theta).
\]
Let $L_N(\theta):=N^{-1}\sum_{i=1}^{N}\ell_i(\theta)$ denote the exact sample log-likelihood criterion. The maximum likelihood estimator satisfies
\[
\widehat\theta_N
\in
\arg\max_{\theta\in\Theta}
L_N(\theta).
\]
\subsection{Simulated maximum likelihood}
\label{subsec:simulated-mle}

Evaluating the exact CCPs $\Phi_j$ requires integrating over the unobserved payoff components and solving the DM's inner problem. For each DM, we use $R_0$ draws from $F_{\varepsilon_0}$ and $R_1$ draws from $F_{\varepsilon_1}$ for the likelihood integrations. We use $S$ prior-state draws from $\pi_1$ to approximate the expectation in the inner objective. Write $v_{t,i}(e_{t,i}^{(r_t)};\theta):=v_t(x_{t,i},w_{t,i},e_{t,i}^{(r_t)};\theta)$ for $t\in\{0,1\}$ and $v_{0,i}^{(r_0)}(\theta):=v_{0,i}(e_{0,i}^{(r_0)};\theta)$. Let $\omega_{1,i}^{(s)}$ denote the prior-state draws and $v_{1,i}(\omega_{1,i}^{(s)};\theta)$ their payoff vectors. For each DM and likelihood draw $r_0$, the simulated inner objective is
\[
\widehat Q_i^0\!\left(P^0;v_{0,i}^{(r_0)}(\theta)\right)
:=
\lambda\,\frac{1}{S}
\sum_{s=1}^{S}
\ln\left(
\sum_{j\in\mathcal J}
P_j^0
\exp\left(
\left(
v_{0,ij}^{(r_0)}(\theta)
+
v_{1,ij}(\omega_{1,i}^{(s)};\theta)
\right)/\lambda
\right)
\right).
\]
Let $\widehat P_i^{0*}(v_{0,i}^{(r_0)}(\theta))\in\argmax_{P^0\in\Delta(\mathcal J)}\widehat Q_i^0(P^0;v_{0,i}^{(r_0)}(\theta))$ be a solution.

Lemma~\ref{lem:uniqueccp} establishes uniqueness for the population inner problem. For the simulated problem, uniqueness follows whenever the simulated transformed payoff vectors are affinely independent, so that $\widehat Q_i^0$ is strictly concave on the simplex. For each fixed inner problem satisfying Assumption~\ref{ass:affineindependence}, independent prior-state draws satisfy this sample condition with probability approaching one as $S\to\infty$.\footnote{Suppress the fixed indices in the transformed payoffs of Assumption~\ref{ass:affineindependence}. The outside alternative has transformed payoff $r_0(\omega_1)=1$. Define $d_s=(r_1(\omega_1^{(s)})-1,\ldots,r_J(\omega_1^{(s)})-1)^{\prime}$. The condition is full column rank of the matrix with rows $d_s^{\prime}$. The sample average of $d_sd_s^{\prime}/(1+\|d_s\|^2)$ converges to a positive-definite matrix under population affine independence. Hence its rank is $J$ with probability approaching one. This is a pointwise result for a fixed inner problem, not a uniform guarantee over all candidate parameters.}

For each likelihood-draw pair $(r_0,r_1)$, substituting the simulated solution into Equation~\eqref{eq:logitccp} gives the simulated model-implied CCP for $j\in\mathcal J$:
\begingroup
\footnotesize
\begin{equation}
\widehat P_{ij}^*\left(
v_{1,i}(e_{1,i}^{(r_1)};\theta)
\mid v_{0,i}(e_{0,i}^{(r_0)};\theta)
\right)
=
\frac{
\widehat P_{ij}^{0*}\left(
v_{0,i}(e_{0,i}^{(r_0)};\theta)
\right)
\exp\left(
\left(
v_{0,ij}(e_{0,i}^{(r_0)};\theta)
+
v_{1,ij}(e_{1,i}^{(r_1)};\theta)
\right)/\lambda
\right)
}{
\displaystyle
\sum_{\ell\in\mathcal J}
\widehat P_{i\ell}^{0*}\left(
v_{0,i}(e_{0,i}^{(r_0)};\theta)
\right)
\exp\left(
\left(
v_{0,i\ell}(e_{0,i}^{(r_0)};\theta)
+
v_{1,i\ell}(e_{1,i}^{(r_1)};\theta)
\right)/\lambda
\right)
}.
\label{eq:simulated-conditional-ccp}
\end{equation}
\endgroup
Holding the observed covariates fixed, we then average these probabilities over the likelihood-draw pairs to approximate $\Phi_j(x_i,w_i;\theta)$:
\[
\widehat\Phi_{ij}(x_i,w_i;\theta)
=
\frac{1}{R_0R_1}
\sum_{r_0=1}^{R_0}\sum_{r_1=1}^{R_1}
\widehat P_{ij}^*\!\left(
v_{1,i}(e_{1,i}^{(r_1)};\theta)
\mid v_{0,i}(e_{0,i}^{(r_0)};\theta)
\right),
\qquad j\in\mathcal J.
\]
For the observed choice $y_i$, define the simulated log-likelihood contribution and the simulated sample log-likelihood criterion by
\[
\widehat\ell_i(\theta)
:=
\sum_{j\in\mathcal J}\mathbf 1\{y_i=j\}
\ln\widehat\Phi_{ij}(x_i,w_i;\theta),
\qquad
\widehat L_N(\theta)
:=
\frac{1}{N}\sum_{i=1}^{N}\widehat\ell_i(\theta).
\]
The simulated maximum likelihood estimator satisfies
\[
\widehat\theta_N^{\mathrm{sim}}
\in
\arg\max_{\theta\in\Theta}\widehat L_N(\theta).
\]
\section{Asymptotic Theory}
\label{sec:asymptotic-theory}

We first establish consistency and asymptotic normality for the exact MLE based on $L_N(\theta)$. We then give conditions on $\widehat L_N(\theta)-L_N(\theta)$ and its derivatives under which the simulated MLE shares these properties.

\subsection{Exact maximum likelihood}

Let $\theta_0\in\Theta$ denote the true parameter vector, and define the population log-likelihood criterion by $L(\theta):=\mathbb E_{\theta_0}[\ell_i(\theta)]$, where the expectation is taken under the true population distribution of $(Y,X,W)$. We establish the large-sample properties of $\widehat\theta_N$ in two steps. For consistency, we verify compactness of $\Theta$, continuity of $L(\theta)$, uniform convergence of $L_N(\theta)$ to $L(\theta)$, and unique maximization of $L(\theta)$ at $\theta_0$, as required by \citet[Theorem~2.1]{newey1994large}. The MLE-specific conditions in \citet[Theorem~2.5]{newey1994large} guide the corresponding likelihood arguments. For asymptotic normality, we establish the mean-zero score and information-matrix equality associated with the MLE specialization in \citet[Theorem~3.3]{newey1994large}, and use these results to verify the conditions of the general asymptotic-normality theorem in \citet[Theorem~3.1]{newey1994large}.

We first assume that $\Theta$ is compact, establish that $\ell_i(\theta)$ is continuous in $\theta$ almost surely, and assume that $\sup_{\theta\in\Theta}|\ell_i(\theta)|$ is integrable. Together with i.i.d.\ sampling, these conditions allow us to apply the uniform law of large numbers in \citet[Lemma~2.4]{newey1994large}, which establishes continuity of $L(\theta)$ and uniform convergence of $L_N(\theta)$ to $L(\theta)$.

\begin{assumption}
\label{ass:theta}
The parameter space $\Theta$ is a compact subset of $\mathbb R^{d_\theta}$.
\end{assumption}

\begin{assumption}
\label{ass:priorpayoffintegrability}
The prior-state payoffs satisfy
$\int_{\Omega_1}\sup_{\theta\in\Theta}\|v_1(\omega_1;\theta)\|_{\infty}\,d\pi_1(\omega_1)<\infty$.
\end{assumption}

\begin{lemma}
\label{lem:criterioncontinuity}
Under Assumptions~\ref{ass:affineindependence}, \ref{ass:independence}, \ref{ass:epsdist}, \ref{ass:theta}, and~\ref{ass:priorpayoffintegrability}, for each $\theta\in\Theta$, $\ell_i(\cdot)$ is continuous at $\theta$ almost surely.
\end{lemma}

We first establish continuity of the inner solution. The continuous payoff specification and Assumptions~\ref{ass:theta} and~\ref{ass:priorpayoffintegrability} allow dominated convergence to establish joint continuity of the inner objective. Assumption~\ref{ass:affineindependence} gives uniqueness of the optimizer for almost every realization at each fixed $\theta$. Berge's maximum theorem on the compact simplex then yields continuity of the optimizer at that parameter value.

We next turn to the likelihood. Equation~\eqref{eq:logitccp} transfers continuity of the optimizer to the model-implied choice probabilities. Since these probabilities are bounded by one and the integration distribution is fixed under Assumption~\ref{ass:independence}, dominated convergence gives continuity of $\Phi_j$. Strict positivity under Assumptions~\ref{ass:independence} and~\ref{ass:epsdist} then allows us to take logarithms, establishing continuity of $\ell_i$ at each fixed $\theta$ almost surely. Appendix~\ref{appendix:criterion} provides the full proof.

\begin{assumption}
\label{ass:criterionintegrability}
The log-likelihood contribution satisfies $\mathbb E[\sup_{\theta\in\Theta}|\ell_i(\theta)|]<\infty$.
\end{assumption}

Assumption~\ref{ass:criterionintegrability} is precisely primitive condition~(iv) of the MLE consistency result in \citet[Theorem~2.5]{newey1994large}. Moreover, setting $d_i:=\sup_{\theta\in\Theta}|\ell_i(\theta)|$, it provides the integrable dominating function required by the uniform law of large numbers in \citet[Lemma~2.4]{newey1994large}.

\begin{assumption}
\label{ass:sampling}
The observations $\{(Y_i,X_i,W_i)\}_{i=1}^{N}$ are independently and identically distributed.
\end{assumption}

The i.i.d.\ sampling condition in Assumption~\ref{ass:sampling}, compactness of $\Theta$ in Assumption~\ref{ass:theta}, almost-sure continuity from Lemma~\ref{lem:criterioncontinuity}, and the integrable domination condition in Assumption~\ref{ass:criterionintegrability} satisfy the conditions of the uniform law of large numbers in \citet[Lemma~2.4]{newey1994large}. Applying that lemma gives continuity of $L(\theta)$ on $\Theta$ and $\sup_{\theta\in\Theta}|L_N(\theta)-L(\theta)|\xrightarrow{p}0$.

To complete the consistency argument, it remains to establish that $L(\theta)$ is uniquely maximized at $\theta_0$. For the population criterion, the information inequality in \citet[Lemma~2.2]{newey1994large} delivers this result once point identification and pointwise integrability of the log-likelihood have been established. We first verify point identification of $\theta$.

Under the support and regularity conditions in \citet{kristensenzhu2026RI}, the CCP function $\Phi(\cdot)$ identifies the payoff functions $u_0(\cdot)$ and $u_1(\cdot)$ up to location and scale. In the parametric specification introduced in Section~\ref{section:parametric_model}, these functions satisfy $u_{0,j}(X_{0,j},W_0;\theta)=\alpha^{\prime}X_{0,j}+\beta^{\prime}W_0$ and $u_{1,j}(X_{1,j},W_1;\theta)=\gamma^{\prime}X_{1,j}+\delta^{\prime}W_1$ for $j=1,\ldots,J$. The following rank condition identifies the coefficients from these payoff functions.

\begin{assumption}
\label{ass:parametricrank}
For each $t\in\{0,1\}$, the matrix $\mathbb E\!\left[\sum_{j=1}^{J}\binom{X_{t,j}}{W_t}\binom{X_{t,j}}{W_t}^{\prime}\right]$ exists and is positive definite.
\end{assumption}

Assumption~\ref{ass:parametricrank} rules out exact collinearity separately in the two payoff components. It ensures that $u_0(\cdot)$ uniquely determines $(\alpha,\beta)$ and that $u_1(\cdot)$ uniquely determines $(\gamma,\delta)$. With the maintained location and scale normalizations, the payoff-function identification results and this rank condition give
\[
\Phi(\cdot;\theta)=\Phi(\cdot;\theta_0)
\Longrightarrow
\bigl(u_0(\cdot;\theta),u_1(\cdot;\theta)\bigr)
=
\bigl(u_0(\cdot;\theta_0),u_1(\cdot;\theta_0)\bigr)
\Longrightarrow
\theta=\theta_0.
\]
The first implication follows from the payoff-function identification results, while the second follows from the parametric specification and Assumption~\ref{ass:parametricrank}. Hence, $\theta_0$ is point identified from the CCP function $\Phi(\cdot)$.

Pointwise integrability follows from Assumption~\ref{ass:criterionintegrability}, which implies $\mathbb E[|\ell_i(\theta)|]<\infty$ for every $\theta\in\Theta$. Together with point identification, this yields the following lemma.

\begin{lemma}
\label{lem:popcriterion}
Suppose the support and regularity conditions in \citet{kristensenzhu2026RI} hold and the location and scale normalizations are imposed. Under Assumptions~\ref{ass:affineindependence}, \ref{ass:parametricrank}, and~\ref{ass:criterionintegrability}, $L(\theta)$ is uniquely maximized at $\theta_0$ over $\Theta$.
\end{lemma}

By the information inequality in \citet[Lemma~2.2]{newey1994large}, $L(\theta)\leq L(\theta_0)$, with equality only if the likelihood probabilities under $\theta$ and $\theta_0$ coincide almost surely. The point-identification argument above then implies $\theta=\theta_0$, proving Lemma~\ref{lem:popcriterion}.

Continuity of $L(\theta)$ and uniform convergence, together with compactness of $\Theta$ under Assumption~\ref{ass:theta} and the unique-maximizer result in Lemma~\ref{lem:popcriterion}, satisfy the conditions of \citet[Theorem~2.1]{newey1994large}. Therefore, $\widehat\theta_N\xrightarrow{p}\theta_0$.

Having established consistency, we impose local smoothness, moment, and nonsingularity conditions for asymptotic normality. Their formulation follows the MLE specialization in \citet[Theorem~3.3]{newey1994large}, and we use them to verify the conditions of the general result in \citet[Theorem~3.1]{newey1994large}.

\begin{assumption}
\label{ass:AN}
There exists a bounded convex neighborhood $\mathcal N_0$ of $\theta_0$ whose closure satisfies $\overline{\mathcal N}_0\subset\operatorname{int}(\Theta)$. On this neighborhood, the following conditions hold:
\begin{enumerate}
\item[(i)] For every $j\in\mathcal J$, the map $\theta\mapsto\Phi_j(X,W;\theta)$ is twice continuously differentiable on $\mathcal N_0$ almost surely.

\item[(ii)] Let $s_i(\theta):=\nabla_\theta\ell_i(\theta)$. The information matrix $\mathcal I_\theta(\theta_0):=\mathbb E[s_i(\theta_0)s_i(\theta_0)^{\prime}]$ is finite and positive definite.

\item[(iii)] The log-likelihood Hessian satisfies $\mathbb E[\sup_{\theta\in\mathcal N_0}\|\nabla_{\theta\theta^{\prime}}^{2}\ell_i(\theta)\|]<\infty$.
\end{enumerate}
\end{assumption}

Because $\mathcal J$ is finite and Assumption~\ref{ass:AN}(i) permits differentiation, the adding-up identity $\sum_{j\in\mathcal J}\Phi_j(X,W;\theta)=1$ implies $\mathbb E[s_i(\theta_0)\mid X_i,W_i]=0$. Assumption~\ref{ass:AN}(ii) ensures that the score is square-integrable, so $\mathbb E[s_i(\theta_0)]=0$. Differentiating the adding-up identity a second time and using Assumption~\ref{ass:AN}(ii)--(iii) yields the information-matrix equality $\mathbb E[-\nabla_{\theta\theta^{\prime}}^{2}\ell_i(\theta_0)]=\mathcal I_\theta(\theta_0)$. Hence, $\mathbb E[\nabla_{\theta\theta^{\prime}}^{2}\ell_i(\theta_0)]=-\mathcal I_\theta(\theta_0)$ is nonsingular because $\mathcal I_\theta(\theta_0)$ is positive definite. The mean-zero score and information-matrix equality are derived in Appendix~\ref{appendix:likelihoodidentities}.

Assumptions~\ref{ass:sampling} and~\ref{ass:AN}(ii), together with the mean-zero score, give the central limit theorem for $N^{-1/2}\sum_{i=1}^Ns_i(\theta_0)$. Assumptions~\ref{ass:sampling} and~\ref{ass:AN}(i), together with the integrable Hessian envelope in Assumption~\ref{ass:AN}(iii), ensure that $\nabla_{\theta\theta^{\prime}}^2L_N(\theta)$ converges locally uniformly to $\mathbb E[\nabla_{\theta\theta^{\prime}}^2\ell_i(\theta)]$. Finally, consistency and $\theta_0\in\operatorname{int}(\Theta)$ imply that $\widehat\theta_N$ is an interior solution with probability approaching one, so $\nabla_\theta L_N(\widehat\theta_N)=0$ with probability approaching one. These results verify the conditions of \citet[Theorem~3.1]{newey1994large}. Hence, $\widehat\theta_N$ is asymptotically normal.

Combining this result with consistency gives the following proposition:

\begin{proposition}
\label{prop:AN}
Suppose the conditions for consistency established above and Assumption~\ref{ass:AN} hold. Then
\[
\widehat\theta_N\xrightarrow{p}\theta_0,
\qquad
\sqrt{N}\bigl(\widehat\theta_N-\theta_0\bigr)
\xrightarrow{d}
\mathcal N\!\left(0,\mathcal I_\theta(\theta_0)^{-1}\right).
\]
\end{proposition}

\subsection{Simulated maximum likelihood}

Proposition~\ref{prop:AN} concerns the exact MLE. For the simulated MLE, we must also control approximation error from the $R_0$, $R_1$, and $S$ draws.

For conventional simulated likelihood using simulation draws that are independent across observations, \citet{lee1995asymptotic} shows that divergence of the number of draws as $N\to\infty$ is sufficient for consistency, without imposing a particular growth rate relative to $N$. Under the corresponding regularity conditions, increasing the number of draws faster than $\sqrt N$ ensures the same centered root-$N$ limiting distribution and asymptotic efficiency as the exact MLE. If the number of draws increases only at the $\sqrt N$ rate, the limiting distribution generally contains a simulation-induced bias. With the inner problem solved exactly, the rate argument of \citet{lee1995asymptotic} applies to the two likelihood integrations approximated by the $R_0$ and $R_1$ draws. Under the corresponding regularity conditions, sufficient conditions are $R_0\to\infty$ and $R_1\to\infty$ for consistency, and $\sqrt N/R_0\to0$ and $\sqrt N/R_1\to0$ for first-order equivalence to the exact MLE.

These rate conditions do not, however, account for the $S$-draw approximation, which affects the likelihood through the solution of the simulated inner problem. This nested approximation can be studied within the more general framework for approximate extremum estimators developed by \citet{kristensen2017higher}. Applying this framework requires smoothness and stochastic-expansion conditions that are technically demanding to verify here. In particular, the set of alternatives with positive optimal integrated CCPs can change with $\theta$, complicating the analysis of the inner solution's derivatives.

We therefore impose sufficient high-level conditions directly on the simulated criterion and its first two derivatives. Suppose that $\widehat L_N(\theta)$ is twice continuously differentiable on $\mathcal N_0$ with probability approaching one and that
\begin{equation}
\begin{aligned}
\sup_{\theta\in\Theta}
\left|
\widehat L_N(\theta)-L_N(\theta)
\right|
&=o_p(1),\\
\sup_{\theta\in\mathcal N_0}
\left\|
\nabla_\theta\widehat L_N(\theta)
-
\nabla_\theta L_N(\theta)
\right\|
&=o_p(N^{-1/2}),\\
\sup_{\theta\in\mathcal N_0}
\left\|
\nabla_{\theta\theta^{\prime}}^2\widehat L_N(\theta)
-
\nabla_{\theta\theta^{\prime}}^2L_N(\theta)
\right\|
&=o_p(1).
\end{aligned}
\label{eq:simulation-approximation-conditions}
\end{equation}
The first condition makes the criterion approximation uniformly negligible and preserves consistency. The second makes simulation error in the average score negligible at the $N^{-1/2}$ scale. The third ensures that the simulated and exact Hessians converge to the same limit.

Under these conditions and those of Proposition~\ref{prop:AN}, $\widehat\theta_N^{\mathrm{sim}}\xrightarrow{p}\theta_0$ and $\sqrt N(\widehat\theta_N^{\mathrm{sim}}-\widehat\theta_N)=o_p(1)$. Hence, the simulated estimator has the asymptotic distribution stated in Proposition~\ref{prop:AN}. Appendix~\ref{appendix:simulated-mle-equivalence} provides the derivation.

\section{Computation}
\label{sec:computation}

This section implements the simulated estimator defined in Section~\ref{subsec:simulated-mle}.
Conditional on the observed data $\{(y_i,x_i,w_i)\}_{i=1}^{N}$, which are held fixed throughout this section, the numerical procedure has an inner rational-inattention layer and an outer optimization layer. The inner layer solves the DM's rational-inattention problem for the optimal integrated CCP vector $P^{0*}(v_0)$. Given the integrated CCPs, we evaluate the model-implied CCPs $P^*(v_1\mid v_0)$ and construct the simulated CCPs $\widehat\Phi_{ij}(x_i,w_i;\theta)$, which approximate the exact CCPs $\Phi_j(x_i,w_i;\theta)$.\footnote{Throughout the computation section, the exact model objects $Q^0$, $P^{0*}$, $P_j^*$, and $\Phi_j$ are written without an individual subscript. Their simulated counterparts, $\widehat Q_i^0$, $\widehat P_i^{0*}$, $\widehat P_{ij}^*$, and $\widehat\Phi_{ij}$, carry the subscript $i$ because each DM uses individual-specific simulation draws. The subscript $j$ denotes the alternative.} The outer layer searches over the structural parameter vector $\theta$ and maximizes the simulated log-likelihood $\widehat L_N(\theta)$.

Because the rational-inattention problem must be solved repeatedly inside the simulated likelihood, computation of the inner layer is central to the implementation. We first discuss the algorithms used to solve the reduced inner problem defined by Equation~\eqref{eq:parametric_Q0}, namely the maximization of $Q^0(P^0;v_0)$ over $P^0(v_0)\in\Delta(\mathcal J)$, and then describe the full nested algorithm used for simulated maximum likelihood estimation.

\subsection{Inner problem solver}
\label{subsec:inner-problem-solver}

Fix a candidate parameter value $\theta:=(\alpha^{\prime},\beta^{\prime},\gamma^{\prime},\delta^{\prime})^{\prime}$, and let $r_0=1,\ldots,R_0$ index the individual-specific likelihood draws $e_{0,i}^{(r_0)}$ from $F_{\varepsilon_0}$. For each DM $i$ and draw $r_0$, the payoff vector $v_{0,i}^{(r_0)}(\theta)$ defined in Section~\ref{subsec:simulated-mle} is $v_0(x_{0,i},w_{0,i},e_{0,i}^{(r_0)};\theta)$. Its inside components are $v_{0,ij}^{(r_0)}(\theta)=\alpha^{\prime}x_{0,ij}+\beta^{\prime}w_{0,i}+e_{0,ij}^{(r_0)}$ for $j=1,\ldots,J$, and its outside component is normalized to zero. Because the covariates and $e_{0,i}^{(r_0)}$ are fixed within this evaluation, we suppress them as explicit arguments on the left-hand side and use the superscript $r_0$ to index the likelihood draw.

The expectation under the common prior $\pi_1$ in each inner objective is approximated using an individual-specific average of $S$ draws. For each DM $i$, draw $S$ independent prior states $\{\omega_{1,i}^{(s)}\}_{s=1}^S$, where $\omega_{1,i}^{(s)}=(x_{1,i}^{(s)},w_{1,i}^{(s)},e_{1,i}^{(s)})$ is drawn from $\pi_1$. For each draw, define $v_{1,ij}(\omega_{1,i}^{(s)};\theta)=\gamma^{\prime}x_{1,ij}^{(s)}+\delta^{\prime}w_{1,i}^{(s)}+e_{1,ij}^{(s)}$ for $j=1,\ldots,J$, with the outside component normalized to zero.

For each DM and likelihood draw $r_0$, we solve the simulated inner problem defined in Section~\ref{subsec:simulated-mle} to obtain $\widehat P_i^{0*}(v_{0,i}^{(r_0)}(\theta))$. All inner solvers use the same two numerical checks.

First, the \textit{simplex-feasibility check} verifies the probability constraints $P^0\in\Delta(\mathcal J)$. Equivalently, the inside probabilities satisfy $P_j^0\geq0$ for $j=1,\ldots,J$ and $\sum_{j=1}^J P_j^0\leq1$, with the outside probability given by $P_0^0=1-\sum_{j=1}^J P_j^0$. Numerically, the candidate vector must be finite and nonnegative, and its adding-up error must satisfy $\left|\sum_{j\in\mathcal J}P_j^0-1\right|\leq\tau_{\mathrm{simplex}}$. The prespecified tolerance $\tau_{\mathrm{simplex}}>0$ applies only to the adding-up condition.

Second, candidates passing the simplex-feasibility check proceed to the \textit{KKT check}. At the simulated optimum $\widehat P_i^{0*}(v_{0,i}^{(r_0)}(\theta))$, the KKT condition for each $j\in\mathcal J$ is
\[
\frac{1}{S}
\sum_{s=1}^{S}
\frac{
\exp\!\left(
\left(
v_{0,ij}^{(r_0)}(\theta)
+
v_{1,ij}(\omega_{1,i}^{(s)};\theta)
\right)/\lambda
\right)
}{
\sum_{\ell\in\mathcal J}
\widehat P_{i\ell}^{0*}\!\left(
v_{0,i}^{(r_0)}(\theta)
\right)
\exp\!\left(
\left(
v_{0,i\ell}^{(r_0)}(\theta)
+
v_{1,i\ell}(\omega_{1,i}^{(s)};\theta)
\right)/\lambda
\right)
}
\leq 1.
\]
Let $\widehat m_{ij}^{(r_0)}(P^0;\theta)$ denote the simulated KKT ratio obtained by evaluating the left-hand side at a candidate probability vector $P^0$. At the simulated optimum, equality is required when $\widehat P_{ij}^{0*}>0$, whereas the inequality may be strict when $\widehat P_{ij}^{0*}=0$. These cases are summarized by the complementary-slackness condition $\widehat P_{ij}^{0*}\bigl[\widehat m_{ij}^{(r_0)}(\widehat P_i^{0*};\theta)-1\bigr]=0$. The KKT conditions therefore accommodate both full-support and boundary solutions.

Numerically, the KKT check uses a prespecified tolerance $\tau_{\mathrm{KKT}}>0$: each ratio must be at most $1+\tau_{\mathrm{KKT}}$, and each absolute complementarity residual must be at most $\tau_{\mathrm{KKT}}$. Each inner solver applies these checks to its accepted iterates without renormalizing them. Before returning a candidate probability vector, we normalize it once and verify that it passes both checks.

We next turn to computational efficiency. We develop an interior-first strategy to reduce the computational cost of repeated inner solves and compare it with a direct boundary-capable strategy. The two strategies differ in when the boundary-capable solver is used. For a fixed simulated inner problem, we suppress the DM and likelihood-draw indices and let $P^{0,(k)}$ denote the candidate probability vector at iteration $k$.

\paragraph{Interior-first strategy.}

Starting from a full-support probability vector $P^{0,(0)}$, the interior solver is applied to the interior KKT equalities. If it converges to an interior solution, the resulting probability vector is retained. If an iterate approaches the boundary before convergence, as indicated by $\min_{j\in\mathcal J}P_j^{0,(k)}\leq\tau_{\mathrm{switch}}$ for a prespecified threshold $\tau_{\mathrm{switch}}>0$, the algorithm makes a one-way transition to a boundary-capable solver, initialized at a feasible probability vector obtained from the interior iterations. In either case, the final candidate must pass the simplex-feasibility and KKT checks described above.

\paragraph{Direct boundary-capable strategy.}

Starting from $P^{0,(0)}$, this strategy applies a boundary-capable solver throughout and uses it for both interior and boundary solutions.

Both strategies require a boundary-capable solver. We first compare four methods under the direct strategy. Based on their performance, we choose a method for the second stage of our interior-first strategy. Sequential quadratic programming (SQP), active-set, and interior-point are implemented using \texttt{fmincon} in \textsc{Matlab}, while Blahut--Arimoto (BA) is implemented using a custom update. Table~\ref{tab:boundary-capable-solvers} summarizes their numerical approaches and implementations.

\begin{table}[!htbp]
\centering
\caption{Boundary-capable inner solvers}
\label{tab:boundary-capable-solvers}
\footnotesize
\setlength{\tabcolsep}{6pt}
\renewcommand{\arraystretch}{1.08}
\begin{tabular}{@{}lll@{}}
\toprule
Method & Numerical approach & \textsc{Matlab} implementation \\
\midrule
SQP & Sequential quadratic programming & \texttt{fmincon} with \texttt{sqp} \\
Active-set & Active-set optimization & \texttt{fmincon} with \texttt{active-set} \\
Interior-point & Interior-point optimization & \texttt{fmincon} with \texttt{interior-point} \\
Blahut--Arimoto & Multiplicative fixed-point iteration & Custom update \\
\bottomrule
\end{tabular}
\end{table}

Appendix~\ref{app:boundary-solvers} provides the algorithmic details. The main distinction is how the methods handle the lower probability bounds. SQP and active-set can treat zero probabilities as binding constraints. Interior-point and BA generally approach the boundary from the interior, using barrier problems and multiplicative fixed-point updates, respectively. When a boundary solution assigns zero probability to an alternative, BA reduces the corresponding probability through repeated multiplicative updates. This mechanism may increase the BA iteration count when boundary or near-boundary solutions are frequent and motivates the cross-$\lambda$ comparison.

Appendix~\ref{app:inner-solver-performance} compares the four methods over $\lambda\in\{1,1.5,2\}$ using the same numerical KKT criterion. Table~\ref{tab:boundary-solver-comparison} shows that all four methods satisfy the simulated KKT conditions in at least $99.995$ percent of cases and produce similar objective values. BA has the lowest mean and median solve times for every value of $\lambda$. Active-set is consistently slower than SQP, while interior-point is the slowest method throughout.

Table~\ref{tab:ba-lambda-performance} in Appendix~\ref{app:inner-solver-performance} shows that BA's computation time increases as numerically near-boundary solutions become more frequent at higher values of $\lambda$. Its speed advantage over SQP also narrows, although BA remains faster for all values considered. This pattern is consistent with slower convergence of BA's multiplicative updates near the boundary of the simplex. The preferred direct solver may therefore depend on the design. In settings with a substantially higher frequency of near-boundary solutions, SQP could become competitive or faster.

Based on the benchmark and robustness comparisons, we select BA as the boundary-capable solver for both strategies. The direct strategy applies BA from the initial iterate, whereas our interior-first strategy begins with Newton--Raphson and makes a one-way transition to BA when the candidate approaches the boundary. We call this hybrid algorithm Newton--BA. A probability threshold determines when to switch, while the KKT ratios guide the choice of the BA starting vector. Appendix~\ref{app:newton-ba} provides the switching-rule details.

Tables~\ref{tab:newton-ba-comparison} and~\ref{tab:newton-ba-full-mle} in Appendix~\ref{app:newton-ba} compare BA and Newton--BA for the inner problem and the complete simulated-likelihood procedure, respectively. Newton--BA has lower mean computation time in both comparisons while producing closely aligned inner solutions, likelihood values, and parameter estimates. We therefore use Newton--BA as the inner solver in our simulated maximum likelihood implementation.

\subsection{Outer problem solver}
\label{subsec:outer-problem-solver}

Given a simulated inner solution, we compute the simulated model-implied CCPs in Equation~\eqref{eq:simulated-conditional-ccp}, then form $\widehat\Phi_{ij}(x_i,w_i;\theta)$ and $\widehat L_N(\theta)$ as defined in Section~\ref{subsec:simulated-mle}. The observed covariates remain fixed, and the prior-state draws enter only through the inner solution.

In the numerical implementation, we minimize the average negative simulated log-likelihood (NLL), $-\widehat L_N(\theta)$, using the quasi-Newton algorithm implemented by \texttt{fminunc} in \textsc{Matlab}. The gradient is approximated by central finite differences with coordinate-specific step $h_k>0$ for $k=1,\ldots,d_\theta$, while the optimizer updates an approximation to the Hessian. Within each estimation run, the likelihood and prior-state draws are generated once and held fixed across evaluations of $\theta$. The simulated criterion is therefore deterministic conditional on the observed data and simulation draws.

After \texttt{fminunc} terminates, we apply two post-estimation checks. The first is the \textit{strict-validity check}. We classify a run as strictly valid only if the estimate and final NLL are finite, \texttt{fminunc} returns a positive exit flag, the NLL returned by \texttt{fminunc} and the post-optimization evaluation of $-\widehat L_N(\widehat\theta_N^{\mathrm{sim}})$ differ by no more than a prespecified tolerance $\tau_{\mathrm{obj}}>0$, and all inner solutions evaluated along the optimization path and at the final estimate pass the simplex-feasibility and KKT checks.\footnote{Under the \texttt{fminunc} exit-flag convention, positive flags denote successful solver outcomes, zero denotes termination after reaching an iteration or function-evaluation limit, and negative flags denote unsuccessful outcomes. We therefore exclude runs with nonpositive exit flags.} 

The second is the \textit{stationarity check}.\footnote{\texttt{fminunc} returns positive exit flags $1$, $2$, $3$, and $5$. Flag $1$ is based directly on a first-order optimality measure. Flags $2$, $3$, and $5$ correspond, respectively, to a change in the parameter vector below \texttt{StepTolerance}, a change in the objective value below \texttt{FunctionTolerance}, and a predicted objective decrease below \texttt{FunctionTolerance}. Because the latter criteria do not directly impose first-order stationarity, we apply the same independent stationarity check to all runs with positive exit flags.} We recompute the central-difference gradient at $\widehat\theta_N^{\mathrm{sim}}$ and require $\|\nabla_\theta[-\widehat L_N(\widehat\theta_N^{\mathrm{sim}})]\|_\infty\leq\tau_{\mathrm{stat}}$, where $\tau_{\mathrm{stat}}>0$ is a prespecified tolerance. Only estimates passing both checks are retained. These checks establish numerical validity and first-order stationarity, but not global optimality.

For inference, we additionally check the covariance matrix and the numerical scores. First, we compute the information matrix $\widehat{\mathcal I}_{\theta,N}:=-\nabla_{\theta\theta^{\prime}}^2\widehat L_N(\widehat\theta_N^{\mathrm{sim}})$ using central finite differences. The \textit{covariance-validity check} requires every likelihood evaluation used in this calculation to be valid. The numerical information matrix must be positive definite with reciprocal condition number greater than a prespecified threshold $\tau_{\mathrm{cond}}>0$, and the resulting covariance matrix must be finite and positive definite.

Second, the \textit{score--gradient identity check} requires the mean individual numerical score to agree with the numerical gradient of $\widehat L_N$ to within a prespecified tolerance $\tau_{\mathrm{score}}>0$ in every component. Both are evaluated at $\widehat\theta_N^{\mathrm{sim}}$ using the same central-difference steps. This checks consistency between the individual and aggregate derivative calculations. Inference is reported only for estimates passing both additional checks.

\subsection{Nested simulated-likelihood algorithm}
\label{subsec:nested-simulated-likelihood-algorithm}

We use Newton--BA to solve the simulated inner problem. Each inner solve applies Newton--Raphson while the iterate remains in the interior. When a Newton candidate approaches the boundary, the KKT-guided switching rule determines the BA starting value and makes a one-way transition to BA.

At each candidate value of $\theta$, the algorithm solves the inner problem for every DM and every likelihood draw of $\varepsilon_0$, combines the resulting integrated CCPs with the likelihood draws of $\varepsilon_1$, and constructs the simulated CCPs $\widehat\Phi_{ij}(x_i,w_i;\theta)$. The outer quasi-Newton optimizer then updates $\theta$. Simulation draws are held fixed across evaluations of $\theta$.

Inner solutions must pass the simplex-feasibility and KKT checks, and final estimates must pass the strict-validity and stationarity checks. Inference additionally requires the covariance-validity and score--gradient identity checks. Table~\ref{tab:computation_algorithm} summarizes the estimation procedure.

\Needspace{12\baselineskip}
\begingroup
\footnotesize
\setlength{\tabcolsep}{4pt}
\renewcommand{\arraystretch}{1.15}

\begin{longtable}{@{}>{\centering\arraybackslash}p{0.06\textwidth}
>{\raggedright\arraybackslash}p{0.89\textwidth}@{}}
\caption{Computation of the nested simulated-likelihood estimator}
\label{tab:computation_algorithm}\\

\toprule
Step & Operation \\
\midrule
\multicolumn{2}{@{}c@{}}{\textit{Observed data and simulation draws}} \\
\midrule
0 & Condition on the observed data $\{(y_i,x_i,w_i)\}_{i=1}^N$. For each DM $i$, generate the likelihood draws $\{e_{0,i}^{(r_0)}\}_{r_0=1}^{R_0}$ from $F_{\varepsilon_0}$ and $\{e_{1,i}^{(r_1)}\}_{r_1=1}^{R_1}$ from $F_{\varepsilon_1}$. Separately generate the DM-specific prior-state draws $\{\omega_{1,i}^{(s)}\}_{s=1}^{S}$ from $\pi_1$. Hold all draws fixed across evaluations of $\theta$. \\[0.35em]

\midrule
\multicolumn{2}{@{}c@{}}{\textit{Candidate parameter and inner solution}} \\
\midrule
1 & At a candidate parameter vector $\theta$, form $v_{0,i}^{(r_0)}(\theta)$ for every DM $i$ and likelihood draw $r_0$. Use the prior-state draws $\{\omega_{1,i}^{(s)}\}_{s=1}^{S}$ to approximate the prior expectation and solve for $\widehat P_i^{0*}(v_{0,i}^{(r_0)}(\theta))$ using Newton--BA. Retain the inner solution only if it passes the simplex-feasibility and KKT checks. \\[0.35em]

\midrule
\multicolumn{2}{@{}c@{}}{\textit{Likelihood evaluation and parameter update}} \\
\midrule
2 & For each DM $i$ and likelihood-draw pair $(r_0,r_1)$, hold the observed covariates fixed, form $v_{1,i}(e_{1,i}^{(r_1)};\theta)$, and evaluate the simulated model-implied CCPs $\widehat P_{ij}^*$ for each $j\in\mathcal J$ using Equation~\eqref{eq:logitccp} and the corresponding inner solution $\widehat P_i^{0*}(v_{0,i}^{(r_0)}(\theta))$. \\[0.35em]

3 & Average $\widehat P_{ij}^*$ over the $R_0R_1$ likelihood-draw pairs to obtain $\widehat\Phi_{ij}(x_i,w_i;\theta)$ for each $j\in\mathcal J$. Construct $\widehat\ell_i(\theta)$ from the observed choice outcome $y_i$ and average over DMs to obtain $\widehat L_N(\theta)$. \\[0.35em]

4 & Update $\theta$ by minimizing $-\widehat L_N(\theta)$ with the outer quasi-Newton optimizer. Repeat Steps 1--4 until the outer optimizer terminates. Retain the resulting estimate only if it passes the strict-validity and stationarity checks. \\
\bottomrule

\end{longtable}
\endgroup
\FloatBarrier

We further specify two implementation details. For initialization, we use a prespecified parameter vector $\theta^{\mathrm{start}}$ for the outer optimization. Each inner solve starts from the uniform probability vector, $P_j^{0,(0)}=1/(J+1)$ for $j\in\mathcal J$.

For numerical stability, the exponential terms entering the inner objective and KKT ratios are evaluated after subtracting the largest scaled payoff within each simulated state. Once the inner solution is obtained, the model-implied CCPs are computed from log weights for alternatives with positive integrated CCPs and set to zero for the remaining alternatives. These calculations preserve the inner objective, KKT ratios, and model-implied CCPs.\footnote{For a fixed DM and simulated state, suppress the DM and simulation-draw indices and define $z_j:=(v_{0j}+v_{1j})/\lambda$ and $z_{\max}:=\max_{j\in\mathcal J}z_j$. The identity $\lambda\ln(\sum_{j\in\mathcal J}P_j^0e^{z_j})=\lambda z_{\max}+\lambda\ln(\sum_{j\in\mathcal J}P_j^0e^{z_j-z_{\max}})$ shows that subtracting $z_{\max}$ and adding it back preserves the inner objective, while the common factor $e^{-z_{\max}}$ cancels from the KKT ratios. After obtaining $P^0$, define $a:=\max_{k:P_k^0>0}\{\ln P_k^0+z_k\}$ and set $\widetilde P_j^*:=e^{\ln P_j^0+z_j-a}$ when $P_j^0>0$ and $\widetilde P_j^*:=0$ otherwise. The model-implied CCP is then evaluated as $P_j^*(v_1\mid v_0)=\widetilde P_j^*/\sum_k\widetilde P_k^*$.}

\section{Monte Carlo Simulations}
\label{sec:simulation-results}

Having described the nested simulated-likelihood algorithm for a fixed observed dataset, we now evaluate the finite-sample performance of its finite-draw implementation through a Monte Carlo experiment. In each replication, we draw the covariates and unobserved components from the specified data-generating process (DGP), generate choice outcomes at the true parameter vector $\theta_0$, and apply the simulated maximum likelihood estimator to the resulting artificial dataset. After presenting the design, we summarize the generated choices and computational validity, evaluate parameter recovery and inference, and examine sensitivity to the number of simulation draws.

\subsection{Monte Carlo design}
\label{sec:monte-carlo-design}

We conduct $B=1{,}000$ Monte Carlo replications for each sample size $N\in\{50,100,200\}$.\footnote{Within each replication, we generate an i.i.d.\ master sample of $N=200$ DMs and use its first $50$ and $100$ observations for the smaller samples. This preserves the marginal DGP at each sample size while permitting paired comparisons across sample sizes within a replication.} The choice set is $\mathcal J=\{0,1,2\}$, where $0$ denotes the outside option and $j=1,2$ denote the two inside alternatives. The outside payoffs are normalized as $v_{0,i0}=v_{1,i0}=0$. The Monte Carlo design specializes to scalar covariates and coefficients, so $d_\theta=4$. The true parameter vector is $\theta_0:=(\alpha_0,\beta_0,\gamma_0,\delta_0)^{\prime}=(0.5,0.5,0.5,0.5)^{\prime}$, and the information-cost parameter is fixed at $\lambda=1$.

\begin{table}[!t]
\centering
\caption{Monte Carlo design}
\label{tab:mc-design}
\footnotesize
\setlength{\tabcolsep}{4pt}
\renewcommand{\arraystretch}{1.02}
\begin{tabular}{@{}
>{\raggedright\arraybackslash}p{0.28\textwidth}
>{\raggedright\arraybackslash}p{0.66\textwidth}
@{}}
\toprule
Object & Specification \\
\midrule
Choice set
& $\mathcal J=\{0,1,2\}$, with $0$ denoting the outside option and $j=1,2$ denoting the inside alternatives \\

Outside option
& $v_{0,i0}=v_{1,i0}=0$ \\

$X_{0,ij}$, $j=1,2$
& $\mathcal U[-2,2]$ \\

$W_{0,i}$
& $\mathcal U[-2,2]$, common across inside alternatives \\

$X_{1,ij}$, $j=1,2$
& $\mathcal N(0,1)$ \\

$W_{1,i}$
& $\mathcal N(0,1)$, common across inside alternatives \\

$F_{\varepsilon_0}$
& $\varepsilon_{0,ij}\sim\mathcal N(0,0.3^2)$ for $j=1,2$ \\

$F_{\varepsilon_1}$
& $\varepsilon_{1,ij}\sim\mathcal N(0,2^2)$ for $j=1,2$ \\

Subjective prior
& Correct beliefs with $\pi_1=F_{X_1,W_1}\otimes F_{\varepsilon_1}$ \\

True parameters
& $\theta_0=(\alpha_0,\beta_0,\gamma_0,\delta_0)=(0.5,0.5,0.5,0.5)$ \\

Information cost
& $\lambda=1$ \\

Monte Carlo replications
& $B=1{,}000$ \\

Sample sizes
& $N\in\{50,100,200\}$ \\

Likelihood draws
& $R_0=30$ draws from $F_{\varepsilon_0}$ and $R_1=400$ draws from $F_{\varepsilon_1}$ \\

Prior-state draws
& $S=800$ draws $\omega_{1,i}^{(s)}=(x_{1,i}^{(s)},w_{1,i}^{(s)},e_{1,i}^{(s)})$ from $\pi_1$ \\

Inner starting value
& $P_j^{0,\mathrm{start}}=1/3$ for each $j\in\mathcal J$ \\

Newton--BA settings
& $\tau_{\mathrm{switch}}=10^{-4}$, $\tau_{\mathrm{KKT}}=10^{-8}$, and $\tau_{\mathrm{simplex}}=10^{-10}$ \\

Outer starting value
& $\theta_b^{\mathrm{start}}=\theta_0+0.05z_b$, where $z_b\sim\mathcal N(0,I_4)$ \\

Numerical differentiation
& Central finite differences with relative step factor $10^{-2}$ \\

Post-estimation validation
& $\tau_{\mathrm{obj}}=10^{-5}$, $\tau_{\mathrm{stat}}=10^{-3}$, $\tau_{\mathrm{cond}}=10^{-10}$, and $\tau_{\mathrm{score}}=10^{-10}$ \\
\bottomrule
\end{tabular}
\end{table}

Within each replication, the Econometrician observes the realized dataset $\{(y_i,x_i,w_i)\}_{i=1}^{N}$. The covariate realizations $(x_i,w_i)$ and choice outcome $y_i$ are observed, whereas the realizations of $(\varepsilon_{0,i},\varepsilon_{1,i})$ are latent and are used only to generate $y_i$. Their population distributions, $F_{\varepsilon_0}$ and $F_{\varepsilon_1}$, are known and used in the simulated likelihood.

For each DM $i$ and inside alternative $j=1,2$, the alternative-specific covariates are independently drawn as $X_{0,ij}\sim\mathcal U[-2,2]$ and $X_{1,ij}\sim\mathcal N(0,1)$. The individual-specific covariates, which are common across the two inside alternatives, are drawn as $W_{0,i}\sim\mathcal U[-2,2]$ and $W_{1,i}\sim\mathcal N(0,1)$. The unobserved payoff components are independently drawn as $\varepsilon_{0,ij}\sim\mathcal N(0,0.3^2)$ and $\varepsilon_{1,ij}\sim\mathcal N(0,2^2)$ for $j=1,2$. The larger variance assigned to $\varepsilon_{1,ij}$ represents the additional uncertainty faced by the DM, who observes $\varepsilon_{0,ij}$ but not $\varepsilon_{1,ij}$ before information acquisition, while the Econometrician observes neither component.

We use a correct-belief benchmark. DMs know the payoff function determining $V_{1,i}$, including the true coefficients $\gamma_0$ and $\delta_0$, and hold the prior
\[
\pi_1=F_{X_1,W_1}\otimes F_{\varepsilon_1}.
\]
Their beliefs therefore match the distribution used to generate the uncertain payoffs.

To generate the choice outcome, we first construct the realized payoff vectors at $\theta_0$. Their inside components are $v_{0,ij}=\alpha_0x_{0,ij}+\beta_0w_{0,i}+e_{0,ij}$ and $v_{1,ij}=\gamma_0x_{1,ij}+\delta_0w_{1,i}+e_{1,ij}$ for $j=1,2$, with $v_{0,i0}=v_{1,i0}=0$. We solve the inner problem at $v_{0,i}$ using a separate DGP prior-state draw pool from $\pi_1$, obtaining $\widehat P_i^{0*}(v_{0,i})$.\footnote{The DGP and estimation prior-state draw pools are generated independently from the same subjective prior $\pi_1$. This prevents the estimator from mechanically reusing the finite-draw approximation employed to generate the data.} Combining this solution with $v_{1,i}$ in Equation~\eqref{eq:logitccp} gives the simulated model-implied CCPs $\widehat P_{ij}^*(v_{1,i}\mid v_{0,i})$ for $j\in\mathcal J$. We then draw $Y_i$ from the categorical distribution defined by these probabilities, and its realization $y_i$ enters the simulated dataset.\footnote{\label{fn:categorical-draw}For each DM $i$, we generate an independent draw $U_i^Y\sim\mathcal U[0,1]$ and apply the inverse-CDF rule using the alternative ordering $(0,1,2)$. We set $Y_i=0$ if $U_i^Y\leq\widehat P_{i0}^*$, $Y_i=1$ if $\widehat P_{i0}^*<U_i^Y\leq\widehat P_{i0}^*+\widehat P_{i1}^*$, and $Y_i=2$ otherwise.}

Using more draws generally reduces simulation approximation error but increases computational cost. One likelihood evaluation requires $NR_0$ inner solutions, each using $S$ prior-state draws, and $NR_0R_1$ evaluations of the model-implied CCP vector. We set $(R_0,R_1,S)=(30,400,800)$ to balance approximation accuracy against computational cost.\footnote{Increasing $R_0$ by one requires $N$ additional inner solutions and $NR_1$ additional CCP evaluations. Increasing $R_1$ by one requires $NR_0$ additional CCP evaluations but no new inner solutions. Increasing $S$ by one adds a prior-state contribution to every iteration of each of the $NR_0$ inner solutions and directly affects the approximation of $\widehat P_i^{0*}$. Thus, the comparatively small value of $R_0$ reflects its high marginal cost, the larger value of $S$ supports an accurate inner approximation, and the larger value of $R_1$ takes advantage of its lower marginal computational cost.} The numbers of draws per DM are the same for all three sample sizes. In each replication, new DGP and estimation draws are generated from separate, replication-specific random-number streams. During each outer optimization, the estimation draws are held fixed as $\theta$ varies.

Each inner problem begins from the uniform probability vector, $P_j^{0,\mathrm{start}}=1/3$ for $j\in\mathcal J$. Newton--BA uses $\tau_{\mathrm{switch}}=10^{-4}$, $\tau_{\mathrm{KKT}}=10^{-8}$, and $\tau_{\mathrm{simplex}}=10^{-10}$. For each replication $b=1,\ldots,B$, the outer parameter search begins from $\theta_b^{\mathrm{start}}=\theta_0+0.05z_b$, where $z_b\sim\mathcal N(0,I_4)$. The same starting value is used for all three sample sizes within replication $b$. Numerical derivatives are computed by central finite differences using the coordinate-specific step $h_k=10^{-2}\max\{1,|\theta_k|\}$ for $k=1,\ldots,4$. This step rule is used to approximate $\nabla_\theta[-\widehat L_N(\theta)]$ during the outer search and to compute a fresh Hessian $\nabla_{\theta\theta^{\prime}}^2[-\widehat L_N(\widehat\theta_N^{\mathrm{sim}})]$ for inference. Table~\ref{tab:mc-design} summarizes the Monte Carlo design and numerical settings.

\FloatBarrier

\subsection{Simulated DGP}

Before evaluating estimator performance, we summarize the generated choices and verify the numerical validity of the calculations used for choice generation. Within each replication, one realized choice $y_i$ is generated for each DM $i=1,\ldots,200$ in the master sample. For each $N$, Table~\ref{tab:mc-dgp-diagnostics} reports the average across the $B=1{,}000$ replications of the choice share among the first $N$ DMs. The final row reports the corresponding average percentage of generation-stage inner solutions $\widehat P_i^{0*}(v_{0,i})$ passing the simplex-feasibility and KKT checks defined in Section~\ref{subsec:inner-problem-solver}. Before drawing each choice, we also verify that the corresponding simulated model-implied CCP vector is finite, nonnegative, and sums to one within tolerance.

\begin{table}[H]
\centering
\caption{Simulated-DGP summary}
\label{tab:mc-dgp-diagnostics}
\footnotesize
\setlength{\tabcolsep}{10pt}
\renewcommand{\arraystretch}{1.05}
\begin{tabular}{@{}lrrr@{}}
\toprule
Statistic
& \multicolumn{1}{c}{$N=50$}
& \multicolumn{1}{c}{$N=100$}
& \multicolumn{1}{c}{$N=200$} \\
\midrule
$y_i=0$ (\%)              & 24.45  & 24.45  & 24.50  \\
$y_i=1$ (\%)              & 37.75  & 37.86  & 37.75  \\
$y_i=2$ (\%)              & 37.81  & 37.69  & 37.75  \\
Simplex-feasibility and KKT pass (\%) & 100.00 & 100.00 & 100.00 \\
\bottomrule
\end{tabular}
\end{table}

The mean choice shares vary little across sample sizes. Approximately $24.5$ percent select the outside option, while approximately $37.7$ percent select each inside alternative. All generation-stage numerical checks pass.

\subsection{Computational validity}

Before reporting parameter recovery, we summarize the numerical checks applied at three stages of estimation. At the inner-solution stage, each candidate must pass the simplex-feasibility and KKT checks defined in Section~\ref{subsec:inner-problem-solver}. At the estimation-run stage, the run must pass the strict-validity check. At the final-estimate stage, the reported estimate must also pass the stationarity check. The strict-validity and stationarity checks are defined in Section~\ref{subsec:outer-problem-solver}. Only estimates passing both checks are retained for parameter recovery. Table~\ref{tab:mc-computational-diagnostics} reports the corresponding pass rates.

\begin{table}[H]
\centering
\caption{Monte Carlo computational validity}
\label{tab:mc-computational-diagnostics}
\footnotesize
\setlength{\tabcolsep}{8pt}
\renewcommand{\arraystretch}{1.05}
\begin{tabular}{@{}lrrr@{}}
\toprule
Numerical diagnostic
& \multicolumn{1}{c}{$N=50$}
& \multicolumn{1}{c}{$N=100$}
& \multicolumn{1}{c}{$N=200$} \\
\midrule
Simplex-feasibility and KKT pass (\%)
& 100.0 & 100.0 & 100.0 \\
Strict-validity pass (\%)
& 99.4 & 98.8 & 98.6 \\
Stationarity pass (\%)
& 99.3 & 98.0 & 98.0 \\
\bottomrule
\end{tabular}
\par\smallskip
\begin{minipage}{0.92\textwidth}
\footnotesize
\emph{Notes:} The simplex-feasibility and KKT pass rate is computed over all inner solutions evaluated at the reported estimates. The remaining rates use all $B=1{,}000$ completed replications as the denominator. The stationarity rate counts estimates passing both the strict-validity and stationarity checks.
\end{minipage}
\end{table}

All $3{,}000$ outer searches return positive termination flags, and every inner solution evaluated at the reported estimates passes the simplex-feasibility and KKT checks. The strict-validity check is passed by $99.4$, $98.8$, and $98.6$ percent of the replications at $N=50$, $100$, and $200$, respectively. After additionally applying the stationarity check, $99.3$ percent of the replications at $N=50$ and $98.0$ percent at both $N=100$ and $N=200$ are retained. The parameter-recovery and inference results below are based on these retained replications.

\subsection{Parameter recovery}

Table~\ref{tab:mc-parameter-recovery} reports bias, standard deviation, and root mean squared error (RMSE) for each parameter and sample size.

Absolute bias and RMSE decrease with sample size for all four parameters. As $N$ increases from $50$ to $200$, RMSE falls by approximately $61$ percent for $\alpha$, $59$ percent for $\beta$, $55$ percent for $\gamma$, and $53$ percent for $\delta$. These declines are broadly consistent with the usual $N^{-1/2}$ benchmark.\footnote{Under this benchmark, the predicted reduction in RMSE from $N=50$ to $N=200$ is $1-\sqrt{50/200}=50$ percent.} The match is closer between $N=100$ and $N=200$, when the observed reductions of approximately $30$--$33$ percent are close to the $29$ percent predicted for a doubling of sample size.

At $N=200$, absolute biases range from $0.0121$ to $0.0161$ and are small relative to the corresponding RMSEs, so estimation error mainly reflects dispersion rather than systematic bias.\footnote{At the population level, $\operatorname{RMSE}^{2}=\operatorname{Bias}^{2}+\operatorname{Var}(\widehat\theta)$. When absolute bias is small relative to RMSE, variance accounts for most of the mean squared error.} The estimates of $\alpha$ and $\beta$ have the smallest RMSEs. Those of $\gamma$ and $\delta$ remain more dispersed, consistent with the larger variance assigned to $\varepsilon_1$ in the DGP, although their dispersion also declines substantially with sample size.

\begin{table}[!htbp]
\centering
\caption{Monte Carlo parameter recovery}
\label{tab:mc-parameter-recovery}
\scriptsize
\setlength{\tabcolsep}{5pt}
\renewcommand{\arraystretch}{1.08}
\begin{tabular}{@{}lrrrrrrrrr@{}}
\toprule
& \multicolumn{3}{c}{$N=50$}
& \multicolumn{3}{c}{$N=100$}
& \multicolumn{3}{c}{$N=200$} \\
\cmidrule(lr){2-4}
\cmidrule(lr){5-7}
\cmidrule(lr){8-10}
Parameter
& \multicolumn{1}{c}{Bias}
& \multicolumn{1}{c}{SD}
& \multicolumn{1}{c}{RMSE}
& \multicolumn{1}{c}{Bias}
& \multicolumn{1}{c}{SD}
& \multicolumn{1}{c}{RMSE}
& \multicolumn{1}{c}{Bias}
& \multicolumn{1}{c}{SD}
& \multicolumn{1}{c}{RMSE} \\
\midrule
$\alpha$
&  0.0621 & 0.1962 & 0.2057
&  0.0226 & 0.1178 & 0.1199
&  0.0121 & 0.0792 & 0.0800 \\
$\beta$
&  0.0870 & 0.2076 & 0.2250
&  0.0387 & 0.1313 & 0.1368
&  0.0161 & 0.0899 & 0.0913 \\
$\gamma$
&  0.0284 & 0.4777 & 0.4783
&  0.0181 & 0.3043 & 0.3046
&  0.0138 & 0.2140 & 0.2143 \\
$\delta$
& -0.0795 & 0.5679 & 0.5732
& -0.0439 & 0.3880 & 0.3903
& -0.0127 & 0.2706 & 0.2707 \\
\bottomrule
\end{tabular}
\par\vspace{0.7em}
\begin{minipage}{0.96\textwidth}
\footnotesize
\emph{Notes:} The true value of every parameter is $0.5$. All statistics are computed over the retained replications, which number $993$, $980$, and $980$ at $N=50$, $N=100$, and $N=200$, respectively. Bias is the mean estimate minus the true value, and SD is the standard deviation of the estimates.
\end{minipage}
\end{table}

\FloatBarrier

\subsection{Inference performance}
\label{subsec:mc-inference}

Section~\ref{sec:asymptotic-theory} derives the information-matrix equality $\mathbb E[-\nabla_{\theta\theta^{\prime}}^2\ell_i(\theta_0)]=\mathcal I_\theta(\theta_0)$. Proposition~\ref{prop:AN} shows that the exact maximum likelihood estimator is asymptotically normal, with covariance matrix approximated for large $N$ by $N^{-1}\mathcal I_\theta(\theta_0)^{-1}$. Under the simulation-approximation conditions in Equation~\eqref{eq:simulation-approximation-conditions}, the difference between the simulated and exact estimators is $o_p(N^{-1/2})$. Hence, the simulated estimator has the same asymptotic distribution as the exact estimator, and the negative Hessian of the simulated criterion evaluated at the simulated estimator converges to $\mathcal I_\theta(\theta_0)$. These conditions describe an asymptotic setting in which the number of draws per DM increases with $N$ and the combined approximation error vanishes at the required rates. Our Monte Carlo design instead uses $(R_0,R_1,S)=(30,400,800)$ at each sample size, with new draws generated in every replication. We therefore use the information-based variance formula as a large-sample benchmark and evaluate its finite-sample, finite-draw performance directly.

Using the numerical information matrix $\widehat{\mathcal I}_{\theta,N}$ defined in Section~\ref{subsec:outer-problem-solver}, we estimate the covariance matrix by $N^{-1}\widehat{\mathcal I}_{\theta,N}^{-1}$. For component $k=1,\ldots,4$, the resulting standard error is $\widehat{\mathrm{SE}}_{N,k}:=\sqrt{N^{-1}[\widehat{\mathcal I}_{\theta,N}^{-1}]_{kk}}$. The nominal $95$ percent marginal Wald confidence interval is $\widehat\theta_{N,k}^{\mathrm{sim}}\pm1.96\widehat{\mathrm{SE}}_{N,k}$.\footnote{The corresponding two-sided Wald test of $H_0:\theta_k=\theta_{0,k}$ rejects at the $5$ percent level when $\left|(\widehat\theta_{N,k}^{\mathrm{sim}}-\theta_{0,k})/\widehat{\mathrm{SE}}_{N,k}\right|>1.96$. Its Monte Carlo rejection rate therefore equals one minus its empirical coverage rate.}

Inference is reported only for replications passing all post-estimation checks defined in Section~\ref{subsec:outer-problem-solver}. Table~\ref{tab:mc-inference-performance} compares the mean estimated standard error with the SD of the estimates over the same valid-inference subset and reports the empirical coverage of the nominal $95$ percent Wald confidence intervals.

\begin{table}[H]
\centering
\caption{Monte Carlo performance of information-based inference}
\label{tab:mc-inference-performance}
\scriptsize
\setlength{\tabcolsep}{3.5pt}
\renewcommand{\arraystretch}{1.08}
\begin{tabular}{@{}lrrrrrrrrr@{}}
\toprule
& \multicolumn{3}{c}{$N=50$}
& \multicolumn{3}{c}{$N=100$}
& \multicolumn{3}{c}{$N=200$} \\
\cmidrule(lr){2-4}
\cmidrule(lr){5-7}
\cmidrule(lr){8-10}
Parameter
& \multicolumn{1}{c}{Mean SE}
& \multicolumn{1}{c}{SE/SD}
& \multicolumn{1}{c}{Coverage}
& \multicolumn{1}{c}{Mean SE}
& \multicolumn{1}{c}{SE/SD}
& \multicolumn{1}{c}{Coverage}
& \multicolumn{1}{c}{Mean SE}
& \multicolumn{1}{c}{SE/SD}
& \multicolumn{1}{c}{Coverage} \\
\midrule
$\alpha$
& 0.1738 & 0.897 & 95.35
& 0.1127 & 0.956 & 94.89
& 0.0771 & 0.974 & 94.99 \\
$\beta$
& 0.1970 & 0.963 & 96.46
& 0.1238 & 0.943 & 94.38
& 0.0830 & 0.923 & 92.95 \\
$\gamma$
& 0.4495 & 0.941 & 94.44
& 0.3077 & 1.012 & 94.99
& 0.2141 & 1.000 & 95.40 \\
$\delta$
& 0.5122 & 0.908 & 92.53
& 0.3625 & 0.934 & 92.23
& 0.2521 & 0.931 & 93.05 \\
\bottomrule
\end{tabular}
\par\vspace{0.7em}
\begin{minipage}{0.96\textwidth}
\footnotesize
\emph{Notes:} All statistics are computed over the same valid-inference subset, comprising $990$, $978$, and $979$ replications at $N=50$, $N=100$, and $N=200$, respectively. Mean SE is the mean information-based standard error, SD is the standard deviation of the estimates, and SE/SD is their ratio. $\mathrm{SE/SD}<1$ indicates that the mean SE understates the sampling dispersion, while $\mathrm{SE/SD}>1$ indicates that it overstates it. Coverage is the percentage of nominal $95$ percent marginal Wald confidence intervals containing the true parameter.
\end{minipage}
\end{table}
\FloatBarrier

Overall, coverage remains close to the nominal $95$ percent level for all four parameters, with a maximum shortfall of less than three percentage points. For $\beta$, coverage is $96.46$, $94.38$, and $92.95$ percent at $N=50$, $100$, and $200$, respectively. The mean standard error is below the SD of $\widehat\beta$ at every sample size, with the largest relative shortfall at $N=200$, where the SE/SD ratio is $0.923$. For $\delta$, coverage ranges from $92.23$ to $93.05$ percent, and the SE/SD ratios range from $0.908$ to $0.934$. These ratios indicate somewhat narrower Wald intervals than the empirical sampling dispersion would suggest. Thus, information-based inference performs well overall, with mild undercoverage for $\delta$ and, at $N=200$, for $\beta$.

To assess robustness to the choice of covariance estimator, we also compute sandwich confidence intervals. Their coverage is generally lower, particularly for $\delta$ (Table~\ref{tab:mc-sandwich} in Appendix~\ref{app:mc-sandwich}). The comparison favors information-based inference in this design.

\subsection{Sensitivity to simulation draws}
\label{subsec:mc-draw-sensitivity}

We assess sensitivity to the number of simulation draws through two complementary exercises. First, we re-estimate the model on $50$ selected Monte Carlo replications at $N=200$, separately doubling $S$, $R_0$, or $R_1$ from the baseline $(R_0,R_1,S)=(30,400,800)$. Median standard-error ratios remain close to one, and paired coverage rates are largely unchanged (Table~\ref{tab:mc-draw-diagnostics} in Appendix~\ref{app:mc-draw-diagnostics}). These results support the stability of information-based inference in the selected replications under the draw increases examined.

Second, we examine the simulated criterion $\widehat L_N(\theta_0)$ at the true parameters, without re-estimation. Across $100$ independent datasets with $N=200$, separately doubling $R_0$, $R_1$, or $S$ from the same baseline raises the mean simulated log likelihood by $0.00044$--$0.00119$ per DM, with the largest change from $S$ (Table~\ref{tab:mc-r0-likelihood-truth} in Appendix~\ref{app:mc-r0-likelihood-truth}). The simulated criterion therefore also shows limited sensitivity to the draw increases examined.

\section{Empirically Calibrated Experiment}
\label{sec:modecanada-experiment}

We complement the Monte Carlo analysis with an experiment calibrated to ModeCanada travel-mode data. The in- and out-of-vehicle times are level-of-service measures generated for each available mode and trip from origin--destination information \citep{Bhat1995}. However, the data do not provide paired measures of anticipated and realized travel times from which to calculate unexpected deviations. This limitation motivates our use of an empirically calibrated experiment. In the experiment, we treat the recorded travel-time levels as known to DMs before information acquisition and include them in $V_{0,i}$, while uncertain travel-time deviations enter $V_{1,i}$.

We retain the observed covariate realizations $(x_{0,i},w_{0,i})$, while the original choice outcomes are used only in calibration. We specify a parametric distribution $F_{X_1,W_1}$ for the uncertain travel-time components and calibrate its scale together with the unit cost of information $\lambda$. As in the Monte Carlo experiment, we use a correct-belief benchmark. After calibration, we draw realizations $(x_{1,i},w_{1,i})$ from $F_{X_1,W_1}$ and generate new choice outcomes from the resulting RI--ARUM DGP. This construction provides an empirically relevant setting with known structural parameters and DGP choice probabilities, following the broad approach of empirically calibrated numerical experiments \citep[see, e.g.,][]{ArmstrongWeidnerZeleneev2026}.

Using the generated choices and corresponding covariates, we estimate and compare three information specifications: no information $(\lambda\to\infty)$, costly partial information $(0<\lambda<\infty)$, and perfect information $(\lambda\to0)$. For this experiment, we construct $\varepsilon_{0,i}$ from i.i.d.\ standard type-I extreme-value (Gumbel) shocks and set $\varepsilon_{1,i}=\mathbf 0$. Under these restrictions, the no- and perfect-information cases reduce exactly to multinomial logit (MNL) models.\footnote{Integrating over the Gumbel shocks underlying $\varepsilon_{0,i}$ yields the standard logit form. Under no information, $V_{1,i}$ enters only through its prior mean. Under perfect information, conditional on the realized covariates, setting $\varepsilon_{1,i}=\mathbf 0$ eliminates any additional payoff shock that the Econometrician would need to integrate out. If $\varepsilon_{1,i}$ were random with a known distribution, the perfect-information CCP would generally be a mixture of logits.} These MNL models provide familiar benchmarks for assessing in-sample likelihood fit, recovery of payoff coefficients and DGP CCPs, and responses to changes in travel time.

\subsection{Experimental design}
\label{sec:modecanada-construction}

We use the public ModeCanada data distributed with the \texttt{mlogit} package \citep{Croissant2020}. The data record intercity travel-mode choices in the Montreal--Toronto corridor \citep{Bhat1995}. For each DM, we observe one trip and the chosen travel mode; trip origins and destinations may differ across DMs. The available modes are car, train, bus, and air. We begin with the $2{,}779$ DMs for whom all four modes are available, exclude the ten DMs who chose bus, and remove bus from the remaining choice sets. The retained sample therefore contains $N=2{,}769$ DMs choosing among $\mathcal J=\{0,1,2\}$, where $0$, $1$, and $2$ denote car, train, and air, respectively. Their observed shares are $0.4576$, $0.1672$, and $0.3752$.

To specify the known payoff component $V_{0,i}$, we use car as the reference alternative. For train and air ($j=1,2$), the vector $x_{0,ij}$ contains the differences relative to car in cost, in-vehicle time, out-of-vehicle time, and service frequency. The cost and travel-time differences are divided by $100$, and the frequency difference by $10$.\footnote{For example, the train cost component is $(\operatorname{cost}_{i1}-\operatorname{cost}_{i0})/100$.} The vector $w_{0,i}$ contains standardized log income, standardized log distance, and two indicators for urban location.\footnote{For example, standardized log income is $\bigl(\ln(\operatorname{income}_i)-\operatorname{mean}(\ln(\operatorname{income}))\bigr)/\operatorname{sd}(\ln(\operatorname{income}))$. The urban-location variable summarizes the large-city status of the trip's origin and destination and takes values $0$, $1$, and $2$. We include indicators for values $1$ and $2$, with $0$ as the reference category.} Following the general specification $V_{0,i}=u_{0,i}+\varepsilon_{0,i}$, we add alternative-specific constants to the deterministic component\footnote{Allowing for alternative-specific constants is a straightforward extension of the payoff specification. The estimation and inference arguments continue to apply under the corresponding identification and regularity conditions.}:
\[
u_{0,i0}=0,
\qquad
u_{0,ij}
=
c_j+\alpha^{\prime}x_{0,ij}+\beta^{\prime}w_{0,i},
\qquad j=1,2.
\]
The coefficient vectors $\alpha$ and $\beta$ retain their definitions from the general model. The additional vector $c=(c_1,c_2)^{\prime}$ contains the train and air constants, which enter separately from the covariates in $x_{0,ij}$. The car constant is normalized to zero. These constants absorb average mode-specific payoff differences not represented by the observed covariates, such as differences in comfort or convenience. The payoff parameter vector for this experiment is therefore $\vartheta:=(c^{\prime},\theta^{\prime})^{\prime}$. For the random component $\varepsilon_{0,i}$, we subtract the car shock from the i.i.d.\ standard Gumbel shocks. This gives $\varepsilon_{0,i0}=0$ and $(\varepsilon_{0,i1},\varepsilon_{0,i2})^{\prime}\sim F_{\varepsilon_0}$, the standard bivariate logistic distribution.

For the $V_{1,i}$ component, uncertainty arises from generated shocks to total travel time, defined as the sum of in-vehicle and out-of-vehicle time.\footnote{Richer data can distinguish pre-choice travel-time information from realized times. \citet{WangRakha2020} study route choice with travel-time information and GPS records, while \citet{ShelatEtAl2023} study waiting-time reliability using smart-card and vehicle-location data.} The median total travel times are $231$, $303$, and $156$ minutes for car, train, and air, respectively. We scale each mode's uncertainty in proportion to its median, with $q>0$ controlling the common scale. Specifically, let $D_i=(D_{i0},D_{i1},D_{i2})^{\prime}$ denote the random travel-time shocks and specify
\[
D_i\sim\mathcal N\!\left(\mathbf 0_3,\Sigma_D(q)\right),
\]
where $D_{i0}$, $D_{i1}$, and $D_{i2}$ have standard deviations of $231q$, $303q$, and $156q$ minutes, respectively, and $\operatorname{Corr}(D_{ij},D_{ik})=\rho=0.30$ for $j\neq k$. Let $d_i=(d_{i0},d_{i1},d_{i2})^{\prime}$ denote a realization of $D_i$. A positive $d_{ij}$ represents a longer-than-expected travel time for mode $j$. To map these shocks into the scalar $V_1$ specification, we define $X_{1,ij}:=D_{ij}/100$ and $W_{1,i}:=D_{i0}/100$, and let $x_{1,ij}$ and $w_{1,i}$ denote their corresponding realizations. Both are measured in hundreds of minutes. The deterministic component of $V_{1,i}$ is therefore
\[
u_{1,i0}=0,
\qquad
u_{1,ij}
=
\gamma x_{1,ij}+\delta w_{1,i}
=
\frac{\gamma d_{ij}+\delta d_{i0}}{100},
\qquad j=1,2.
\]
Here, $x_{1,ij}$ captures the realized travel-time shock to inside mode $j$, while $w_{1,i}$ captures the car shock entering both inside-mode payoffs. A negative $\gamma$ lowers the payoff of an inside mode relative to car when its travel time increases, while a positive $\delta$ raises its relative payoff when car's travel time increases. The Gaussian law of $D_i$ induces $F_{X_1,W_1}(\cdot;q)$. We also set $\varepsilon_{1,i}=\mathbf 0_3$, so all uncertainty in $V_{1,i}$ arises from these travel-time shocks. DMs know the payoff function determining $V_{1,i}$, including the true coefficients $\gamma_0$ and $\delta_0$, and hold correct beliefs. Their prior $\pi_1$ is therefore fully specified by the population distribution $F_{X_1,W_1}(\cdot;q)$.

Having specified the two payoff components, we next calibrate the uncertainty parameter $q$ and the unit cost of information $\lambda$, and then generate the experimental choice outcome. Figure~\ref{fig:modecanada-design} summarizes this procedure.

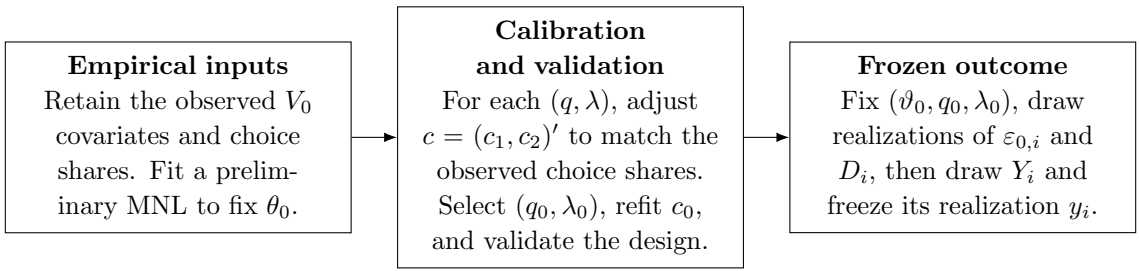
\begin{figure}[!htbp]
\centering
\begin{tikzpicture}[
    node distance=6mm,
    mainbox/.style={
        draw=black,
        line width=0.4pt,
        text width=0.265\linewidth,
        minimum height=18mm,
        align=center,
        inner sep=5pt,
        font=\footnotesize
    },
    mainflow/.style={
        -{Latex[length=1.7mm]},
        line width=0.4pt
    }
]

\node[mainbox] (inputs) {
\textbf{Empirical inputs}\\
Retain the observed $V_0$ covariates and choice shares.
Fit a preliminary MNL to fix $\theta_0$.
};

\node[mainbox, right=of inputs] (calibration) {
\textbf{Calibration and validation}\\
For each $(q,\lambda)$, adjust $c=(c_1,c_2)^{\prime}$
to match the observed choice shares.
Select $(q_0,\lambda_0)$, refit $c_0$, and validate the design.
};

\node[mainbox, right=of calibration] (generation) {
\textbf{Frozen outcome}\\
Fix $(\vartheta_0,q_0,\lambda_0)$,
draw realizations of $\varepsilon_{0,i}$ and $D_i$,
then draw $Y_i$ and freeze its realization $y_i$.
};

\draw[mainflow] (inputs) -- (calibration);
\draw[mainflow] (calibration) -- (generation);

\end{tikzpicture}
\caption{Construction of the empirically calibrated ModeCanada experiment}
\label{fig:modecanada-design}
\end{figure}

Let $y_i^{\mathrm{orig}}$ denote the choice outcome recorded for DM $i$ in the original ModeCanada data. Using the $V_0$ specification above, we fit a preliminary MNL to these choices, with the type-I extreme-value scale normalized to one. Let $\widetilde c\in\mathbb R^2$ and $\widetilde\alpha,\widetilde\beta\in\mathbb R^4$ denote the resulting column-vector estimates. We set $\alpha_0=\widetilde\alpha$ and $\beta_0=\widetilde\beta$ as the true $V_0$ covariate coefficients in the experimental DGP.

Within $\widetilde\alpha$, let $\widetilde\alpha_{\mathrm{ovt}}$ denote the coefficient on the car-relative out-of-vehicle-time difference $(\operatorname{ovt}_{ij}-\operatorname{ovt}_{i0})/100$. We use this estimate to place the generated travel-time shocks on the same payoff scale as the observed travel-time variables, setting $\gamma_0=\widetilde\alpha_{\mathrm{ovt}}$ and $\delta_0=-\widetilde\alpha_{\mathrm{ovt}}$. These equal and opposite coefficients make each inside mode's $V_1$ payoff depend on its travel-time shock relative to car, so a common shock to both modes leaves their relative payoff unchanged.

We do not retain the preliminary constants $\widetilde c$. Instead, we hold
\[
\theta_0
=
(\alpha_0^{\prime},\beta_0^{\prime},\gamma_0,\delta_0)^{\prime}
\in\mathbb R^{10}
\]
fixed and calibrate $q$ and $\lambda$ jointly. For each candidate $(q,\lambda)$, we adjust the column vector $c=(c_1,c_2)^{\prime}$ to reproduce the observed choice shares. Share matching alone does not select $(q,\lambda)$. We choose a design with nontrivial information acquisition, both numerically near-boundary and three-alternative inner solutions, and numerical stability, subject to acceptable share matching.

The search selects $(q_0,\lambda_0)$. We then refit $c_0$ using new seeds and use different seeds to validate the complete design. The resulting values are
\[
q_0=0.20,
\qquad
\lambda_0=0.75,
\qquad
c_0=(1.0264,0.3364)^{\prime}.
\]
Appendix~\ref{app:modecanada-preliminary} reports the complete preliminary MNL estimates, while Appendix~\ref{app:modecanada-calibration} describes the calibration and validation procedure.

Having fixed $(\vartheta_0,q_0,\lambda_0)$, we generate the experimental choice sample once. For each DM, we draw realizations of the random column vectors $\varepsilon_{0,i}$ and $D_i$ from their respective distributions. Conditional on the resulting $v_{0,i}$, we solve the rational-inattention problem under the correct-belief prior $\pi_1$, using $S=32{,}768$ scrambled Sobol draws to approximate the prior integration. This yields the computed inner solution $\widehat P_i^{0*}(v_{0,i})$, from which we evaluate the simulated model-implied CCPs at the realized $v_{1,i}$. We then draw
\[
Y_i\mid v_{0,i},v_{1,i}
\sim
\operatorname{Categorical}\!\left(
\widehat P_{i0}^{*}(v_{1,i}\mid v_{0,i}),
\widehat P_{i1}^{*}(v_{1,i}\mid v_{0,i}),
\widehat P_{i2}^{*}(v_{1,i}\mid v_{0,i})
\right),
\]
and let $y_i$ denote the realized choice. We implement this draw using the independent-uniform inverse-CDF rule described in footnote~\ref{fn:categorical-draw}. All draws use fixed, nonoverlapping seeds. The resulting sample $\{y_i\}_{i=1}^{N}$ is generated from the calibrated RI--ARUM DGP and has car, train, and air shares of $0.4652$, $0.1708$, and $0.3640$, respectively. We freeze this sample and use it to estimate and compare the three models. The original ModeCanada choices $\{y_i^{\mathrm{orig}}\}_{i=1}^{N}$ play no further role after calibration.

\subsection{Model comparison}
\label{sec:modecanada-models}

Using the frozen experimental choices $\{y_i\}_{i=1}^{N}$ and the corresponding covariate realizations constructed above, we estimate and compare three information specifications. Under NI $(\lambda\to\infty)$, information becomes prohibitively costly, so DMs acquire none. Because the generated travel-time shocks are centered at zero, the expected $V_{1,i}$ contribution vanishes and their choices depend only on $V_{0,i}$. Under RI $(0<\lambda<\infty)$, DMs acquire information endogenously at a cost. Under PI $(\lambda\to0)$, DMs observe the realized $V_{1,i}$ component before choosing.

Let $\vartheta=(c^{\prime},\theta^{\prime})^{\prime}\in\mathbb R^{12}$ denote the full payoff parameter vector, and let $\vartheta^{\mathrm{NI}}:=(c^{\prime},\alpha^{\prime},\beta^{\prime})^{\prime}\in\mathbb R^{10}$ denote the restricted parameter vector for NI. The deterministic payoff vectors $u_{0,i}(\vartheta),u_{1,i}(\vartheta)\in\mathbb R^3$ are evaluated at the observed covariate realizations, with the car component normalized to zero. Define $v_{0,i}(\vartheta,e_{0,i}):=u_{0,i}(\vartheta)+e_{0,i}$. For $j\in\mathcal J$, the exact CCPs are
\[
\begin{aligned}
\Phi_{ij}^{\mathrm{NI}}(\vartheta^{\mathrm{NI}})
&=
\frac{\exp\!\left(u_{0,ij}(\vartheta^{\mathrm{NI}})\right)}
{\sum_{\ell\in\mathcal J}\exp\!\left(u_{0,i\ell}(\vartheta^{\mathrm{NI}})\right)},
\\[4pt]
\Phi_{ij}^{\mathrm{RI}}(\vartheta)
&=
\int_{e_{0,i}}
\frac{
P_{ij}^{0*}\!\left(v_{0,i}(\vartheta,e_{0,i})\right)
\exp\!\left(\left[v_{0,ij}(\vartheta,e_{0,i})+u_{1,ij}(\vartheta)\right]/\lambda_0\right)
}{
\sum_{\ell\in\mathcal J}
P_{i\ell}^{0*}\!\left(v_{0,i}(\vartheta,e_{0,i})\right)
\exp\!\left(\left[v_{0,i\ell}(\vartheta,e_{0,i})+u_{1,i\ell}(\vartheta)\right]/\lambda_0\right)
}
\,dF_{\varepsilon_0}(e_{0,i}),
\\[4pt]
\Phi_{ij}^{\mathrm{PI}}(\vartheta)
&=
\frac{\exp\!\left(u_{0,ij}(\vartheta)+u_{1,ij}(\vartheta)\right)}
{\sum_{\ell\in\mathcal J}\exp\!\left(u_{0,i\ell}(\vartheta)+u_{1,i\ell}(\vartheta)\right)}.
\end{aligned}
\]

In the RI expression, $P_i^{0*}$ solves the inner problem under correct beliefs, with travel-time shocks distributed according to $F_{X_1,W_1}(\cdot;q_0)$ and information cost $\lambda_0$. NI estimates the restricted vector $\vartheta^{\mathrm{NI}}$, while RI and PI estimate the full vector $\vartheta$. We estimate RI by simulated maximum likelihood and the two MNL benchmarks by maximum likelihood. Appendix~\ref{app:modecanada-comparison-numerics} reports the numerical implementation and independent-draw checks.

Having estimated the three models, we first compare in-sample likelihood fit and payoff-coefficient recovery. We then examine recovery of baseline CCPs $\Phi_j$ and their responses to reduced train travel time. These comparisons use the same generated sample, with DGP CCPs and responses evaluated numerically using RI at the true coefficients.

\paragraph{In-sample likelihood fit.}
Average NLL measures fit to the realized choices, with lower values indicating better in-sample fit. Table~\ref{tab:modecanada-fit} shows that RI fits the generated choices substantially better than NI and similarly to PI. However, the small RI--PI likelihood difference is sensitive to simulation seeds and the number of draws (Appendix~\ref{app:modecanada-draw-sensitivity}).

\begin{table}[H]
\centering
\caption{In-sample likelihood fit}
\label{tab:modecanada-fit}
\footnotesize
\setlength{\tabcolsep}{6pt}
\renewcommand{\arraystretch}{1.08}
\begin{tabular*}{0.40\linewidth}{@{\extracolsep{\fill}}lr@{}}
\toprule
Model & Average NLL \\
\midrule
NI & $0.717904$ \\
RI & $0.602853$ \\
PI & $0.603527$ \\
\bottomrule
\end{tabular*}
\par\vspace{0.35em}
\begin{minipage}{0.75\linewidth}
\raggedright
\emph{Notes:} NLL is averaged over the $N=2{,}769$ generated choices. RI uses one evaluation with the estimation seeds, $R_0=128$, and $S=1{,}600$; NI and PI are analytic. 
\end{minipage}
\end{table}

\paragraph{Parameter recovery.}
Similar in-sample likelihood fit does not imply similar recovery of preferences. Under misspecification, likelihood estimates need not converge to the true structural parameters \citep{White1982}. In attention models, ignoring information-processing frictions can also distort inferred preferences \citep{TurloEtAl2025}. We therefore compare the reported RI and PI coefficient estimates directly with their DGP values (Appendix~\ref{app:modecanada-parameter-comparison}, Table~\ref{tab:modecanada-coefficients}).

For the travel-time-shock coefficients, the RI estimates $(-4.3613,4.4516)$ are close to the DGP values $(\gamma_0,\delta_0)=(-4.4374,4.4374)$. The PI estimates, $(-2.1565,2.0773)$, are roughly half the DGP magnitudes. Although RI is not closer for every coefficient, it recovers these uncertain travel-time effects much more closely in this sample.

\paragraph{Baseline CCP recovery.}
We next assess CCP recovery at the baseline travel times. Unlike NLL, which measures fit to the choice outcomes $Y_i$, this comparison evaluates fitted CCPs against numerically evaluated DGP CCPs. We summarize the differences using CCP RMSE across individuals and the three travel modes.\footnote{For model $m\in\{\mathrm{NI},\mathrm{RI},\mathrm{PI}\}$, define $\operatorname{RMSE}_{\mathrm{CCP}}^{m}=\left[\frac{1}{3N}\sum_{i=1}^{N}\sum_{j=0}^{2}(\widehat\Phi_{ij}^{m}-\widehat\Phi_{ij}^{\mathrm{DGP}})^2\right]^{1/2}$, where $\widehat\Phi_{ij}^{m}$ is the fitted CCP and $\widehat\Phi_{ij}^{\mathrm{DGP}}$ is the numerical DGP reference.}

RI's CCP RMSE is approximately $15\%$ lower than PI's, while NI's error is substantially larger (Table~\ref{tab:modecanada-baseline-ccps}). RI retains lower CCP RMSE than both benchmarks across all ten new seeds and under separate or joint doubling of $R_0$ and $S$ (Appendix~\ref{app:modecanada-probability-check}).

\begin{table}[H]
\centering
\caption{Recovery of baseline DGP CCPs}
\label{tab:modecanada-baseline-ccps}
\footnotesize
\setlength{\tabcolsep}{6pt}
\renewcommand{\arraystretch}{1.08}
\begin{tabular*}{0.40\linewidth}{@{\extracolsep{\fill}}lr@{}}
\toprule
Model & CCP RMSE \\
\midrule
NI & $0.14965$ \\
RI & $0.02307$ \\
PI & $0.02703$ \\
\bottomrule
\end{tabular*}
\par\vspace{0.35em}
\begin{minipage}{0.75\linewidth}
\raggedright
\emph{Notes:} CCP RMSE is computed across $N=2{,}769$ DMs and three modes at the baseline travel times. Fitted RI CCPs use the estimation seeds, while the DGP reference uses separate seeds. Both use $R_0=128$ and $S=1{,}600$.
\end{minipage}
\end{table}

\paragraph{Responses to reduced train travel time.}

Finally, we reduce train in-vehicle time by $40\%$ and recompute the CCPs without re-estimating coefficients. Other travel attributes, realized travel-time shocks, and the prior $\pi_1$ remain fixed. Information choices are re-optimized for both RI and the DGP. For each model, the individual response is $\widehat\Delta_{ij}^{m}:=\widehat\Phi_{ij}^{m,\mathrm{after}}-\widehat\Phi_{ij}^{m,\mathrm{before}}$. Response RMSE measures how closely these changes match the corresponding DGP changes.\footnote{For $m\in\{\mathrm{NI},\mathrm{RI},\mathrm{PI}\}$, define $\operatorname{RMSE}_{\mathrm{response}}^{m}=\left[\frac{1}{3N}\sum_{i=1}^{N}\sum_{j=0}^{2}(\widehat\Delta_{ij}^{m}-\widehat\Delta_{ij}^{\mathrm{DGP}})^2\right]^{1/2}$, where $\widehat\Delta_{ij}^{\mathrm{DGP}}$ is the CCP change computed from RI at the true coefficients.}

Table~\ref{tab:modecanada-substitution} reports average share changes, the sources of train's gain, and response RMSE. We also summarize departures from independence of irrelevant alternatives (IIA) using the mean absolute change in individual car--air log odds, labeled IIA deviation.\footnote{For each fitted model and the DGP, define $D_{\mathrm{IIA}}^{m}=\frac{1}{N}\sum_{i=1}^{N}\left|\ln\frac{\widehat\Phi_{i0}^{m,\mathrm{after}}}{\widehat\Phi_{i2}^{m,\mathrm{after}}}-\ln\frac{\widehat\Phi_{i0}^{m,\mathrm{before}}}{\widehat\Phi_{i2}^{m,\mathrm{before}}}\right|$, where modes $0$ and $2$ denote car and air, respectively.}

\begin{table}[H]
\centering
\caption{Responses to a 40\% reduction in train in-vehicle time}
\label{tab:modecanada-substitution}
\footnotesize
\setlength{\tabcolsep}{6pt}
\renewcommand{\arraystretch}{1.08}
\begin{tabular*}{0.76\linewidth}{@{\extracolsep{\fill}}l*{4}{S[table-format=-2.4,group-digits=false]}@{}}
\toprule
Measure & {DGP} & {NI} & {RI} & {PI} \\
\midrule
\multicolumn{5}{@{}l}{\emph{Choice-share change (percentage points)}} \\
Car & -3.784 & -4.892 & -3.789 & -4.368 \\
Train & 8.718 & 11.377 & 8.968 & 10.146 \\
Air & -4.934 & -6.485 & -5.179 & -5.778 \\
\addlinespace[3pt]
\multicolumn{5}{@{}l}{\emph{Sources of train's gain (\%)}} \\
From car & 43.408 & 43.000 & 42.249 & 43.051 \\
From air & 56.592 & 57.000 & 57.751 & 56.949 \\
\addlinespace[3pt]
\multicolumn{5}{@{}l}{\emph{Response accuracy and substitution pattern}} \\
Response RMSE & {---} & 0.0506 & 0.0073 & 0.0166 \\
IIA deviation & 0.0708 & 0.0000 & 0.0925 & 0.0000 \\
\bottomrule
\end{tabular*}
\par\vspace{0.35em}
\begin{minipage}{0.76\linewidth}
\raggedright
\emph{Notes:} Share changes are $100N^{-1}\sum_{i=1}^{N}\widehat\Delta_{ij}^{m}$ percentage points. RI and DGP calculations use separate seeds, with $R_0=128$ and $S=1{,}600$. Within each, the same draws are used before and after the travel-time reduction.
\end{minipage}
\end{table}

RI more closely reproduces the DGP increase in train's choice share than either MNL benchmark and has the lowest response RMSE, approximately $56\%$ lower than PI's. RI and PI imply similar divisions of train's gain between car and air, with PI slightly closer to the DGP on this measure.

Both MNL benchmarks impose IIA, so each DM's car-to-air choice odds remain unchanged when train improves. Their IIA deviation is therefore zero. Under RI and the DGP, these odds change, with baseline IIA deviations of $0.0925$ and $0.0708$, respectively.

Table~\ref{tab:modecanada-response-sensitivity} in Appendix~\ref{app:modecanada-response-sensitivity} shows little numerical variation in average share changes or the division of train's gain. RI retains lower response RMSE than both MNL benchmarks, while PI remains closer to the DGP than RI on the division of train's gain, across all ten paired evaluations at each of the four draw configurations. The magnitudes of the car--air log-odds changes, however, remain sensitive to integration draws.

Taken together, the results show that NI fits the generated choices substantially less well than RI and PI. The RI--PI comparison then shows that similar in-sample likelihood fit need not imply similar recovery of preferences, baseline CCPs, or responses to changes in travel time.

\section{Conclusion}
\label{section:conclusion}

This paper develops likelihood-based estimation for the parametric RI--ARUM. We establish consistency and asymptotic normality of the exact MLE and provide conditions for the simulated estimator to be first-order equivalent to it. In our computational comparisons, the KKT-guided Newton--BA algorithm reduces computation time relative to direct BA. The two methods produce closely aligned likelihood values and parameter estimates. Monte Carlo results show improving parameter recovery with sample size and confidence-interval coverage generally close to the nominal level, with modest undercoverage for some parameters. Paired draw increases leave typical standard errors and coverage largely unchanged.

The empirically calibrated ModeCanada experiment illustrates the value of modeling costly information acquisition beyond in-sample likelihood fit. RI and perfect-information logit fit the generated choices similarly, but RI more closely recovers preferences over uncertain travel time. RI also has lower RMSE for baseline choice probabilities and their travel-time responses than both logit benchmarks. Although the RI--PI likelihood ranking is sensitive to numerical integration, the probability and response rankings persist across the seeds and numbers of draws examined.

This numerical sensitivity motivates further work on more accurate integration at manageable computational cost. Such improvements would help clarify the close RI--PI likelihood comparison in the calibrated experiment.

Beyond improving numerical integration, a methodological extension is to estimate the prior under the correct-belief benchmark used in the Monte Carlo simulations and calibrated experiment. Under this benchmark, $\pi_1=F_{X_1,W_1}\otimes F_{\varepsilon_1}$. We could first estimate $F_{X_1,W_1}$ nonparametrically from the observed covariates $(X_1,W_1)$ and then estimate $\theta$ by maximizing the plug-in log likelihood $L_N(\theta,\widehat F_{X_1,W_1})$. The distributions of $\varepsilon_0$ and $\varepsilon_1$ would remain known. This extension requires estimation and inference that account for uncertainty in the estimated prior.

\clearpage
\bibliographystyle{apalike}
\bibliography{references}
\clearpage

\appendix
\begin{appendices}
\numberwithin{equation}{section}
\numberwithin{table}{section}
\numberwithin{figure}{section}

\section{Proofs for Section~\ref{sec:asymptotic-theory}}

\subsection{Proof of Lemma~\ref{lem:criterioncontinuity}}
\label{appendix:criterion}

For fixed $(x_i,w_i,e_{0,i})$, write $v_{0,i}(\theta):=v_{0,i}(x_{0,i},w_{0,i},e_{0,i};\theta)$. The reduced objective is
\[
Q^0\!\left(P^0;v_{0,i}(\theta)\right)
=
\lambda
\int_{\Omega_1}
\ln\!\left(
\sum_{j\in\mathcal J}
P_j^0
\exp\!\left(
\frac{
v_{0,ij}(\theta)
+
v_{1,ij}(\omega_{1,i};\theta)
}{\lambda}
\right)
\right)
\,d\pi_1(\omega_{1,i}).
\]
The integrand is jointly continuous in $(P^0,\theta)$. Its logarithm is well defined because the full probability vector sums to one and every exponential term is positive.

Since $P^0$ is a probability vector over $\mathcal J$, the weighted log-sum-exp inequality gives
\[
\left|
\lambda
\ln\!\left(
\sum_{j\in\mathcal J}
P_j^0
\exp\!\left(
\frac{
v_{0,ij}(\theta)
+
v_{1,ij}(\omega_{1,i};\theta)
}{\lambda}
\right)
\right)
\right|
\leq
\left\|
v_{0,i}(\theta)
+
v_{1,i}(\omega_{1,i};\theta)
\right\|_\infty.
\]
For the fixed $(x_i,w_i,e_{0,i})$, define $C_i:=\sup_{\theta\in\Theta}\|v_{0,i}(\theta)\|_{\infty}$. Compactness of $\Theta$ and continuity of the parameterized payoff vector imply that $C_i<\infty$. Hence, the integrand is uniformly bounded over $(P^0,\theta)$ by $C_i+\sup_{\theta\in\Theta}\|v_{1,i}(\omega_{1,i};\theta)\|_{\infty}$. This bound is integrable under $\pi_1$ by Assumption~\ref{ass:priorpayoffintegrability}. Dominated convergence therefore implies that $Q^0(P^0;v_{0,i}(\theta))$ is jointly continuous in $(P^0,\theta)$ on $\Delta(\mathcal J)\times\Theta$.

Now fix any $\theta\in\Theta$ and consider a sequence $\theta_n\to\theta$. By Assumptions~\ref{ass:affineindependence} and~\ref{ass:independence} and Fubini's theorem, for almost every $(x_i,w_i)$, the inner optimizer at $\theta$ is unique for $F_{\varepsilon_0}$-almost every $e_{0,i}$. For such realizations, joint continuity of the objective and compactness of the simplex imply, by Berge's maximum theorem,
\[
P^{0*}(v_{0,i}(\theta_n))\longrightarrow P^{0*}(v_{0,i}(\theta)).
\]
Where an optimizer is not unique, use a measurable maximizing selection. Uniqueness at the limit parameter $\theta$ suffices for this convergence.

For each fixed $e_{1,i}$, continuity of the payoff vectors and Equation~\eqref{eq:logitccp} give convergence of the model-implied CCPs along the same sequence. The denominator remains positive because the integrated CCPs sum to one and each exponential term is positive.

For each $j=1,\ldots,J$, the CCP $\Phi_j$ evaluated at $(x_i,w_i)$ is
\[
\Phi_j(x_i,w_i;\theta)
=
\int_{\mathcal E_0\times\mathcal E_1}
P_j^*\!\left(
v_{1,i}(x_{1,i},w_{1,i},e_{1,i};\theta)
\mid v_{0,i}(x_{0,i},w_{0,i},e_{0,i};\theta)
\right)
\,dF_{\varepsilon_0}(e_{0,i})
\,dF_{\varepsilon_1}(e_{1,i}).
\]
For almost every $(x_i,w_i)$, the integrand converges for $F_{\varepsilon_0}\otimes F_{\varepsilon_1}$-almost every pair of realizations $(e_{0,i},e_{1,i})$. Since it lies between zero and one, the constant one is an integrable bound. Dominated convergence therefore gives $\Phi_j(x_i,w_i;\theta_n)\to\Phi_j(x_i,w_i;\theta)$. The outside probability, recovered by adding up, also converges.

Under Assumptions~\ref{ass:independence} and~\ref{ass:epsdist}, $\Phi_j(x_i,w_i;\theta)>0$ for every $j\in\mathcal J$ \citep[see][]{kristensenzhu2026RI}. Continuity of $\ln(\cdot)$ on $(0,\infty)$ gives convergence of the log probabilities.

The indicators $\mathbf 1\{Y_i=j\}$ do not depend on the parameter, so $\ell_i(\theta_n)\to\ell_i(\theta)$ almost surely. Since the convergent sequence was arbitrary, $\ell_i(\cdot)$ is continuous at each fixed $\theta$ almost surely. This is the pointwise probability-one continuity condition used in \citet[Lemma~2.4]{newey1994large}.

\subsection{Mean-Zero Score and Information-Matrix Equality}
\label{appendix:likelihoodidentities}

All expectations below are taken under the true population distribution indexed by $\theta_0$. Recall that $s_i(\theta):=\nabla_\theta\ell_i(\theta)$. Under the maintained assumptions, $\Phi_j(X_i,W_i;\theta_0)>0$ almost surely for every $j\in\mathcal J$ \citep[see][]{kristensenzhu2026RI}. Assumption~\ref{ass:AN}(i) ensures that the required derivatives exist. Since $\Pr_{\theta_0}(Y_i=j\mid X_i,W_i)=\Phi_j(X_i,W_i;\theta_0)$, we have
\[
\begin{aligned}
s_i(\theta_0)
&=\sum_{j\in\mathcal J}\mathbf 1\{Y_i=j\}
\frac{\nabla_\theta\Phi_j(X_i,W_i;\theta_0)}
{\Phi_j(X_i,W_i;\theta_0)},\\
\mathbb E[s_i(\theta_0)\mid X_i,W_i]
&=\sum_{j\in\mathcal J}
\nabla_\theta\Phi_j(X_i,W_i;\theta_0)=0,
\end{aligned}
\]
where the last equality follows from $\sum_{j\in\mathcal J}\Phi_j(X_i,W_i;\theta)=1$ for all $\theta$. Assumption~\ref{ass:AN}(ii) ensures that the score is square-integrable and hence integrable, so taking unconditional expectations gives $\mathbb E[s_i(\theta_0)]=0$.

Differentiating the log-likelihood contribution twice and conditioning on $(X_i,W_i)$ gives
\[
\begin{aligned}
\mathbb E\!\left[
-\nabla_{\theta\theta^{\prime}}^2\ell_i(\theta_0)
\mid X_i,W_i
\right]
&=
\sum_{j\in\mathcal J}
\frac{
\nabla_\theta\Phi_j(X_i,W_i;\theta_0)
\left(
\nabla_\theta\Phi_j(X_i,W_i;\theta_0)
\right)^{\prime}
}{
\Phi_j(X_i,W_i;\theta_0)
}
-
\sum_{j\in\mathcal J}
\nabla_{\theta\theta^{\prime}}^2
\Phi_j(X_i,W_i;\theta_0)\\
&=
\sum_{j\in\mathcal J}
\frac{
\nabla_\theta\Phi_j(X_i,W_i;\theta_0)
\left(
\nabla_\theta\Phi_j(X_i,W_i;\theta_0)
\right)^{\prime}
}{
\Phi_j(X_i,W_i;\theta_0)
}\\
&=
\mathbb E\!\left[
s_i(\theta_0)s_i(\theta_0)^{\prime}
\mid X_i,W_i
\right].
\end{aligned}
\]
The second equality follows because differentiating the adding-up identity twice gives $\sum_{j\in\mathcal J}\nabla_{\theta\theta^{\prime}}^2\Phi_j(X_i,W_i;\theta_0)=0$. Assumption~\ref{ass:AN}(ii) ensures that the score second moment is finite, while Assumption~\ref{ass:AN}(iii) ensures that the log-likelihood Hessian is integrable. Taking unconditional expectations therefore yields
\[
\mathbb E\!\left[
-\nabla_{\theta\theta^{\prime}}^2\ell_i(\theta_0)
\right]
=
\mathbb E\!\left[
s_i(\theta_0)s_i(\theta_0)^{\prime}
\right]
=
\mathcal I_\theta(\theta_0).
\]

\subsection{Asymptotic equivalence of the exact and simulated estimators}
\label{appendix:simulated-mle-equivalence}

\begin{lemma}
\label{lem:simulated-mle-equivalence}
Suppose the conditions of Proposition~\ref{prop:AN} and the high-level simulation-approximation conditions in Equation~\eqref{eq:simulation-approximation-conditions} hold. Suppose also that $\widehat L_N(\theta)$ is twice continuously differentiable on $\mathcal N_0$ with probability approaching one. Then
\[
\widehat\theta_N^{\mathrm{sim}}\xrightarrow{p}\theta_0,
\qquad
\sqrt N
\left(
\widehat\theta_N^{\mathrm{sim}}-\widehat\theta_N
\right)
=o_p(1),
\]
and
\[
-\nabla_{\theta\theta^{\prime}}^2
\widehat L_N
\left(
\widehat\theta_N^{\mathrm{sim}}
\right)
\xrightarrow{p}
\mathcal I_\theta(\theta_0).
\]
Hence,
\[
\sqrt N
\left(
\widehat\theta_N^{\mathrm{sim}}-\theta_0
\right)
\xrightarrow{d}
\mathcal N\!\left(
0,
\mathcal I_\theta(\theta_0)^{-1}
\right).
\]
\end{lemma}

\begin{proof}
The consistency conditions for the exact estimator give $\sup_{\theta\in\Theta}|L_N(\theta)-L(\theta)|\xrightarrow{p}0$, where $L(\theta)$ is uniquely maximized at $\theta_0$. The first condition in Equation~\eqref{eq:simulation-approximation-conditions} therefore implies
\[
\sup_{\theta\in\Theta}
\left|
\widehat L_N(\theta)-L(\theta)
\right|
\leq
\sup_{\theta\in\Theta}
\left|
\widehat L_N(\theta)-L_N(\theta)
\right|
+
\sup_{\theta\in\Theta}
\left|
L_N(\theta)-L(\theta)
\right|
\xrightarrow{p}0.
\]
Theorem~2.1 of \citet{newey1994large} then gives $\widehat\theta_N^{\mathrm{sim}}\xrightarrow{p}\theta_0$.

Because both estimators converge to the interior point $\theta_0$, they lie in the convex neighborhood $\mathcal N_0$ and satisfy their respective first-order conditions with probability approaching one. Define
\[
\overline H_N
:=
\int_0^1
\nabla_{\theta\theta^{\prime}}^2
L_N\!\left(
\widehat\theta_N
+
t\left(
\widehat\theta_N^{\mathrm{sim}}-\widehat\theta_N
\right)
\right)
\,dt.
\]
The integral mean-value expansion of the exact score gives
\[
\sqrt N
\left(
\widehat\theta_N^{\mathrm{sim}}-\widehat\theta_N
\right)
=
-\overline H_N^{-1}\sqrt N
\left(
\nabla_\theta\widehat L_N
\left(
\widehat\theta_N^{\mathrm{sim}}
\right)
-
\nabla_\theta L_N
\left(
\widehat\theta_N^{\mathrm{sim}}
\right)
\right).
\]
Local uniform convergence of the exact sample Hessian and consistency of both estimators imply $\overline H_N\xrightarrow{p}-\mathcal I_\theta(\theta_0)$ and hence $\overline H_N^{-1}=O_p(1)$. The second condition in Equation~\eqref{eq:simulation-approximation-conditions} makes the score difference $o_p(N^{-1/2})$. Therefore,
\[
\sqrt N
\left(
\widehat\theta_N^{\mathrm{sim}}-\widehat\theta_N
\right)
=o_p(1).
\]
Proposition~\ref{prop:AN} and Slutsky's theorem yield the stated asymptotic distribution.

Finally, the third condition in Equation~\eqref{eq:simulation-approximation-conditions}, local uniform convergence of the exact sample Hessian, and consistency of $\widehat\theta_N^{\mathrm{sim}}$ give
\[
\nabla_{\theta\theta^{\prime}}^2
\widehat L_N(\widehat\theta_N^{\mathrm{sim}})
=
\nabla_{\theta\theta^{\prime}}^2
L_N(\widehat\theta_N^{\mathrm{sim}})
+o_p(1)
\xrightarrow{p}
\mathbb E\!\left[
\nabla_{\theta\theta^{\prime}}^2\ell_i(\theta_0)
\right]
=
-\mathcal I_\theta(\theta_0).
\]
This proves the Hessian result.
\end{proof}

\section{Additional details for Section~\ref{sec:computation}}

\subsection{Boundary-capable solvers for the inner problem}
\label{app:boundary-solvers}

This appendix describes the four boundary-capable solvers considered for the simulated inner problem. Fix an individual $i$, a likelihood draw $r_0$, and a candidate parameter value $\theta$, and write $v_0:=v_{0,i}^{(r_0)}(\theta)$. The numerical decision variables are the inside integrated CCPs $P^0(v_0):=(P_1^0(v_0),\ldots,P_J^0(v_0))^{\prime}$. They satisfy $P_j^0(v_0)\geq0$ for $j=1,\ldots,J$ and $\sum_{j=1}^JP_j^0(v_0)\leq1$, while the outside probability is recovered as $P_0^0(v_0)=1-\sum_{j=1}^JP_j^0(v_0)$. Thus, the \texttt{fmincon}-based methods solve
\[
\min_{P^0(v_0)\in\Delta(\mathcal J)}
-\widehat Q^0(P^0;v_0).
\]
Whenever the outside alternative appears in a sum over $\mathcal J$, its probability is understood in this residual sense. The derivatives supplied to \texttt{fmincon} are taken with respect to the $J$ inside coordinates, holding $v_0$ fixed.

\paragraph{SQP.}

Sequential quadratic programming solves a sequence of constrained quadratic approximations to the nonlinear objective. Let $P^{0,(k)}(v_0)$ denote the current inside probability vector. At iteration $k$, SQP computes a direction $d^{(k)}\in\mathbb R^J$ from a quadratic subproblem of the form
\[
d^{(k)}
\in
\arg\min_{d\in\mathbb R^J}
\left\{
-\nabla_{P^0}\widehat Q^0(P^{0,(k)};v_0)^{\prime}d
+
\frac{1}{2}d^{\prime}H^{(k)}d
\right\},
\]
subject to $P^{0,(k)}(v_0)+d\in\Delta(\mathcal J)$. Here, $H^{(k)}$ approximates the Hessian of the Lagrangian for the minimization problem. The next iterate is $P^{0,(k+1)}(v_0)=P^{0,(k)}(v_0)+\omega^{(k)}d^{(k)}$, where $\omega^{(k)}\in(0,1]$ is selected by the solver. Because the quadratic subproblem imposes the simplex constraints directly, SQP can return either full-support or boundary solutions. We implement SQP using \texttt{fmincon} with \texttt{Algorithm = 'sqp'}.

\paragraph{Active-set.}

The active-set method also solves the constrained problem directly. At each iteration, it maintains a working set of simplex constraints treated as binding. These may include the lower bounds $P_j^0(v_0)=0$ and the adding-up constraint $\sum_{j=1}^JP_j^0(v_0)=1$, which corresponds to a zero outside probability. Conditional on the current working set, the method solves a local equality-constrained quadratic subproblem and then updates the working set as constraints become binding or cease to bind. It can therefore return both interior and boundary solutions, although its computational performance depends on how efficiently the binding constraints are identified. We implement this method using \texttt{fmincon} with \texttt{Algorithm = 'active-set'}.

\paragraph{Interior-point.}

The interior-point method approximates the original problem on the closed simplex through a sequence of barrier problems defined on its interior. A representative barrier subproblem is
\[
\min_{P^0(v_0)\in\operatorname{int}\Delta(\mathcal J)}
\left\{
-\widehat Q^0(P^0;v_0)
-\mu_{\mathrm{bar}}
\sum_{j\in\mathcal J}\ln P_j^0(v_0)
\right\},
\]
where $\operatorname{int}\Delta(\mathcal J):=\{p\in\mathbb R^J:p_j>0,\ j=1,\ldots,J,\ \text{and}\ \sum_{j=1}^Jp_j<1\}$, $\mu_{\mathrm{bar}}>0$ is the barrier parameter, and $P_0^0(v_0):=1-\sum_{j=1}^JP_j^0(v_0)$ is the outside probability. The barrier term tends to positive infinity as any inside probability or the outside probability approaches zero. For each fixed $\mu_{\mathrm{bar}}>0$, this keeps the solution strictly inside the simplex. As $\mu_{\mathrm{bar}}$ decreases toward zero, the barrier weakens and the solutions may approach a boundary solution of the original constrained problem. Thus, the interior-point method approximates a boundary optimum using increasingly small positive probabilities rather than explicitly imposing zero-probability constraints. We implement this method using \texttt{fmincon} with \texttt{Algorithm = 'interior-point'}.

\paragraph{Blahut--Arimoto.}

The Blahut--Arimoto algorithm exploits the fixed-point structure of the rational-inattention problem. For a related application in rational inattention, see \citet{caplin2019rational}. The prior expectation is approximated using the individual-specific prior-state draws $\{\omega_{1,i}^{(s)}\}_{s=1}^S$. To simplify notation, write $v_{1,j}^{(s)}:=v_{1,ij}(\omega_{1,i}^{(s)};\theta)$ for each $j\in\mathcal J$ and $s=1,\ldots,S$. The outside payoffs are normalized to zero.

By Equation~\eqref{eq:logitccp}, at the simulated inner solution $\widehat P^{0*}(v_0)$, the simulated model-implied CCP for alternative $j\in\mathcal J$ and prior-state draw $s$ is
\[
\widehat P_j^*\left(v_1^{(s)}\mid v_0;\theta\right)
=
\frac{
\widehat P_j^{0*}(v_0)
\exp\left(
\left(
v_{0j}+v_{1,j}^{(s)}
\right)/\lambda
\right)
}{
\sum_{\ell\in\mathcal J}
\widehat P_\ell^{0*}(v_0)
\exp\left(
\left(
v_{0\ell}+v_{1,\ell}^{(s)}
\right)/\lambda
\right)
}.
\]
Recall from Section~\ref{subsec:inner-problem-solver} that, for a candidate integrated CCP vector $P^0(v_0)$, the corresponding simulated KKT ratio is
\[
\widehat m_j(P^0;v_0)
:=
\frac{1}{S}
\sum_{s=1}^{S}
\frac{
\exp\left(
\left(
v_{0j}+v_{1,j}^{(s)}
\right)/\lambda
\right)
}{
\sum_{\ell\in\mathcal J}
P_\ell^0(v_0)
\exp\left(
\left(
v_{0\ell}+v_{1,\ell}^{(s)}
\right)/\lambda
\right)
},
\qquad j\in\mathcal J.
\]
The simulated Bayes-consistency condition can then be written as
\[
\widehat P_j^{0*}(v_0)
=
\frac{1}{S}
\sum_{s=1}^{S}
\widehat P_j^*\left(v_1^{(s)}\mid v_0;\theta\right)
=
\widehat P_j^{0*}(v_0)
\widehat m_j\left(\widehat P^{0*};v_0\right),
\qquad j\in\mathcal J.
\]
Together with $\widehat P^{0*}(v_0)\in\Delta(\mathcal J)$, the simulated KKT conditions are
\begin{equation}
\widehat m_j\left(\widehat P^{0*};v_0\right)\leq1
\qquad\text{and}\qquad
\widehat P_j^{0*}(v_0)
\left(
\widehat m_j\left(\widehat P^{0*};v_0\right)-1
\right)
=0,
\qquad j\in\mathcal J.
\label{eq:simulated-inner-kkt}
\end{equation}
Thus, an alternative with $\widehat P_j^{0*}(v_0)>0$ satisfies $\widehat m_j(\widehat P^{0*};v_0)=1$, whereas an alternative with $\widehat P_j^{0*}(v_0)=0$ may satisfy the inequality strictly.

Motivated by this fixed-point relation, BA applies the multiplicative update
\[
P_j^{0,(k+1)}(v_0)
=
P_j^{0,(k)}(v_0)
\widehat m_j\left(P^{0,(k)};v_0\right),
\qquad
j\in\mathcal J.
\]
By construction, the BA update has two useful properties: it preserves simplex feasibility and is defined directly by the KKT ratios. First,
\[
\sum_{j\in\mathcal J}P_j^{0,(k+1)}(v_0)
=
\sum_{j\in\mathcal J}
P_j^{0,(k)}(v_0)
\widehat m_j\!\left(P^{0,(k)};v_0\right)
=
1.
\]
Because both factors in the update are nonnegative, every iterate remains in the probability simplex. Second, each probability is multiplied by its current KKT ratio. The iteration terminates when the candidate passes the simplex-feasibility check and the KKT check based on Equation~\eqref{eq:simulated-inner-kkt}. Because the update is multiplicative, probabilities whose optimal values are zero approach zero only gradually, which can increase the number of iterations near the boundary and motivates the solver comparison below.

\subsection{Inner-solver performance and near-boundary solutions}
\label{app:inner-solver-performance}

We compare SQP, active-set, interior-point, and Blahut--Arimoto (BA) as direct boundary-capable solvers for the simulated inner problem. The comparison uses the true parameter vector $\theta_0$, the covariate and unobserved payoff-component distributions, and the simulation draw sizes specified in Table~\ref{tab:mc-design}. We use $B=10$ replications, fix $N=200$, and consider $\lambda\in\{1,1.5,2\}$. Although each inner problem conditions on a fixed likelihood draw, we repeat the comparison over all $R_0$ draws for every individual. This gives $B\times N\times R_0=60{,}000$ common test problems for each value of $\lambda$. The $R_1$ draws are not used because they are needed only for likelihood evaluation. For notational simplicity, throughout this subsection we write $v_0:=v_{0,i}^{(r_0)}(\theta_0)$ for the payoff vector held fixed in each inner problem.

To compare the four solvers on a common numerical basis, we apply the simplex-feasibility and KKT checks defined in Section~\ref{subsec:inner-problem-solver} to every method. At each accepted iterate, the simplex-feasibility check is applied first: the probability vector must be finite and nonnegative, and its adding-up error must not exceed $\tau_{\mathrm{simplex}}$. The KKT check is applied only to iterates passing the simplex-feasibility check and requires each KKT ratio to be at most $1+\tau_{\mathrm{KKT}}$ and each absolute complementarity residual to be at most $\tau_{\mathrm{KKT}}$. Before final acceptance, the returned candidate is normalized once, after which both checks are repeated. For this comparison, we set $\tau_{\mathrm{simplex}}=10^{-10}$ and $\tau_{\mathrm{KKT}}=10^{-6}$. All reported solver outputs pass the simplex-feasibility check, while the KKT pass rates are reported below.

BA applies the common checks at each iteration and stops at the first iterate satisfying both requirements. For each \texttt{fmincon}-based method, an output function applies the same checks after every accepted iteration. All four methods begin from the uniform probability vector. If a \texttt{fmincon} attempt does not produce a candidate satisfying both final checks, the method is rerun from additional strictly interior starting values. The reported solver time includes all attempted solves.

Table~\ref{tab:boundary-solver-comparison} shows that the four methods produce numerically similar solutions. Their KKT pass rates are at least $99.995$ percent, and their mean objective gaps are below $4.4\times10^{-7}$. In terms of computation time, active-set is consistently slower than SQP, and interior-point is the slowest method. BA has the lowest mean and median solve times at every value of $\lambda$. Its computation time increases with $\lambda$, however, causing its speed advantage over SQP to narrow. Thus, BA is the fastest direct boundary-capable solver over the designs considered, although the size of its advantage depends on $\lambda$.

\begin{table}[!htbp]
\centering
\caption{Performance of direct boundary-capable inner solvers}
\label{tab:boundary-solver-comparison}
\scriptsize
\setlength{\tabcolsep}{6pt}
\renewcommand{\arraystretch}{1.08}
\begin{tabular}{@{}llrrrr@{}}
\toprule
$\lambda$ & Method & Mean time (ms) & Median time (ms) & KKT pass (\%) & Mean Obj. gap \\
\midrule
1.0 & SQP & 1.7284 & 1.6776 & 100.000 & $6.652\times10^{-13}$ \\
& Active-set & 2.2755 & 2.1156 & 99.995 & $1.323\times10^{-12}$ \\
& Interior-point & 3.1965 & 3.0859 & 100.000 & $3.832\times10^{-8}$ \\
& Blahut--Arimoto & 0.4769 & 0.3541 & 100.000 & $1.103\times10^{-7}$ \\
\addlinespace
1.5 & SQP & 1.5507 & 1.5462 & 100.000 & $7.708\times10^{-11}$ \\
& Active-set & 2.0173 & 1.9700 & 100.000 & $1.204\times10^{-11}$ \\
& Interior-point & 3.1055 & 3.1014 & 100.000 & $7.824\times10^{-8}$ \\
& Blahut--Arimoto & 0.8297 & 0.5795 & 100.000 & $2.787\times10^{-7}$ \\
\addlinespace
2.0 & SQP & 1.4575 & 1.4502 & 100.000 & $1.535\times10^{-9}$ \\
& Active-set & 1.8766 & 1.8477 & 100.000 & $3.565\times10^{-11}$ \\
& Interior-point & 3.0534 & 3.0418 & 100.000 & $1.190\times10^{-7}$ \\
& Blahut--Arimoto & 1.2278 & 0.8130 & 100.000 & $4.370\times10^{-7}$ \\
\bottomrule
\end{tabular}
\par\vspace{0.7em}
\begin{minipage}{0.92\textwidth}
\footnotesize
\emph{Notes:} Times are milliseconds per inner solve and include same-method restarts. KKT pass is the percentage of outputs satisfying the common tolerance $\tau_{\mathrm{KKT}}=10^{-6}$. For each problem, the objective gap is the highest objective among valid solver outputs minus the method's objective. Mean gaps use only the method's valid outputs.
\end{minipage}
\end{table}

To examine this sensitivity, we distinguish the theoretical consideration set from its numerical classification. The theoretical consideration set at $v_0$ is
\[
\mathcal C(v_0)
:=
\left\{
j\in\mathcal J:
P_j^{0*}(v_0)>0
\right\}.
\]
This set cannot be recovered from finite BA iterates by checking for exact zeros. Because BA starts from full support and its multiplicative factors are strictly positive, every finite iterate remains strictly positive in exact arithmetic, even when the limiting solution assigns zero probability to an alternative. We therefore define the thresholded numerical consideration set as
\[
\widehat{\mathcal C}_{\tau_{\mathrm{supp}}}(v_0)
:=
\left\{
j\in\mathcal J:
\widehat P_j^{0*}(v_0)>\tau_{\mathrm{supp}}
\right\}.
\]
For descriptive reporting, we classify a computed solution as numerically near-boundary when $|\widehat{\mathcal C}_{\tau_{\mathrm{supp}}}(v_0)|<|\mathcal J|$, using the cutoff $\tau_{\mathrm{supp}}=10^{-6}$.

Table~\ref{tab:ba-lambda-performance} shows that the numerically near-boundary share increases with $\lambda$, rising from $5.31$ percent at $\lambda=1$ to $29.60$ percent at $\lambda=2$. Over the same range, the mean number of BA iterations rises from $49.46$ to $140.44$, while its 95th percentile rises from $122$ to $428$. The mean BA solve time also increases from $0.4769$ to $1.2278$ milliseconds.

\begin{table}[!htbp]
\centering
\caption{BA performance and numerically near-boundary solutions}
\label{tab:ba-lambda-performance}
\scriptsize
\setlength{\tabcolsep}{6pt}
\renewcommand{\arraystretch}{1.08}
\begin{tabular}{@{}lrrrrr@{}}
\toprule
$\lambda$ & Near-boundary (\%) & Mean iter. & P95 iter. & Mean time (ms) & BA/SQP \\
\midrule
1.0 & 5.31 & 49.46 & 122 & 0.4769 & 0.276 \\
1.5 & 16.06 & 92.77 & 265 & 0.8297 & 0.535 \\
2.0 & 29.60 & 140.44 & 428 & 1.2278 & 0.842 \\
\bottomrule
\end{tabular}
\par\vspace{0.7em}
\begin{minipage}{0.92\textwidth}
\footnotesize
\emph{Notes:} Near-boundary is the percentage of BA solutions for which at least one probability does not exceed $\tau_{\mathrm{supp}}=10^{-6}$. Mean iter.\ and P95 iter.\ report the mean and 95th percentile of the number of completed BA multiplicative updates. Mean time is the mean elapsed wall-clock time per inner solve. BA/SQP is the ratio of mean BA time to mean SQP time, so a value below one indicates that BA is faster.
\end{minipage}
\end{table}

This joint pattern is consistent with the multiplicative BA update. When an optimal probability is zero, BA approaches zero gradually rather than reaching it after finitely many iterations. As numerically near-boundary solutions become more frequent, BA tends to require more iterations and computation time. This narrows BA's advantage over SQP, with the BA/SQP mean-time ratio rising from $0.276$ at $\lambda=1$ to $0.842$ at $\lambda=2$.

\begin{samepage}
Although BA's computation time is sensitive to $\lambda$, the BA/SQP ratio remains below one throughout the values considered. BA therefore remains the fastest direct solver under the benchmark and robustness designs examined here. In settings with a substantially higher prevalence of near-boundary solutions, SQP could become competitive or faster, so the preferred direct solver need not be universal. For the present design, we select BA as the boundary-capable solver for the inner problem.

\end{samepage}

\afterpage{\FloatBarrier}

\subsection{KKT-guided Newton--BA algorithm}
\label{app:newton-ba}

The preceding comparison identifies BA as the fastest direct boundary-capable solver over the designs considered. We next examine whether its computation time can be reduced by applying Newton--Raphson when the iterate remains in the interior and transitioning to BA when the full-support Newton system is no longer appropriate. We call this one-way procedure the KKT-guided Newton--BA algorithm, abbreviated Newton--BA. Once the transition to BA occurs, the algorithm does not return to Newton--Raphson. We first introduce the Newton stage and the switching rule, then compare Newton--BA with direct BA in isolated inner problems, and finally validate both methods within the full simulated-likelihood problem.

Fix an inner problem with $v_0:=v_{0,i}^{(r_0)}(\theta)$. For notational simplicity, write $v_{1,j}^{(s)}:=v_{1,ij}(\omega_{1,i}^{(s)};\theta)$ and define $r_j^{(s)}:=\exp((v_{0j}+v_{1,j}^{(s)})/\lambda)$ for $j\in\mathcal J$ and $s=1,\ldots,S$. Under the outside-option normalization, $r_0^{(s)}=1$. We also suppress the fixed argument $v_0$ and write $P^0:=P^0(v_0)$, with the outside probability recovered as $P_0^0=1-\sum_{j=1}^J P_j^0$.

At an interior $P^0$, the gradient of the simulated reduced objective with respect to the inside probabilities has components
\[
G_j(P^0)
=
\frac{\lambda}{S}
\sum_{s=1}^{S}
\frac{
r_j^{(s)}-r_0^{(s)}
}{
\sum_{\ell\in\mathcal J}P_\ell^0r_\ell^{(s)}
},
\qquad
j=1,\ldots,J,
\]
and the Hessian has elements
\[
H_{jk}(P^0)
=
-\frac{\lambda}{S}
\sum_{s=1}^{S}
\frac{
\bigl(r_j^{(s)}-r_0^{(s)}\bigr)
\bigl(r_k^{(s)}-r_0^{(s)}\bigr)
}{
\left(\sum_{\ell\in\mathcal J}P_\ell^0r_\ell^{(s)}\right)^2
},
\qquad
j,k=1,\ldots,J.
\]
Newton--Raphson uses the direction $(\Delta P_1^0,\ldots,\Delta P_J^0)^{\prime}=-H(P^0)^{-1}G(P^0)$. A backtracking line search shortens the step when necessary so that every component of $P^0$, including the residually determined outside probability, remains strictly positive and the simulated objective increases.

Newton--Raphson is well suited to an interior solution because the interior first-order conditions are the equality system $G(P^0)=0$. It is not, however, a boundary-capable method. As a probability approaches zero, the line search must prevent the Newton step from leaving the simplex, the Hessian may become ill-conditioned, and, more fundamentally, the full-support equality system no longer characterizes a boundary optimum. At a boundary solution, alternatives assigned zero probability are instead governed by the KKT inequalities.

\subsubsection{KKT-guided switching rule}

The remaining question is when Newton--Raphson should switch to BA. Recall that the KKT conditions in Equation~\eqref{eq:simulated-inner-kkt} require $\widehat m_j(\widehat P^{0*};v_0)\leq1$ for every alternative, with equality whenever $\widehat P_j^{0*}>0$. If the Newton candidate has not passed the KKT check and at least one probability is at or below $\tau_{\mathrm{switch}}$, the KKT ratios guide the choice of the BA starting vector. Each ratio is the multiplicative factor for the corresponding probability in the next BA update. BA starts directly from the candidate if $\widehat m_j(P^0;v_0)\leq1+\tau_{\mathrm{KKT}}$ for every alternative with $P_j^0\leq\tau_{\mathrm{switch}}$. Otherwise, the Newton step is shortened before the switch. The algorithm proceeds as follows:

\begin{enumerate}
\item The algorithm starts from the uniform probability vector and evaluates the numerical KKT criterion defined in Appendix~\ref{app:inner-solver-performance} after each Newton iteration. If the criterion is satisfied, the current probability vector is retained as the inner solution.

\item If the KKT criterion is not satisfied and every probability exceeds $\tau_{\mathrm{switch}}$, the algorithm continues with Newton--Raphson. If the Newton direction cannot be computed reliably, the line search fails, or the maximum number of Newton iterations is reached, the algorithm transitions to BA from the last feasible interior iterate.

\item Let $P^{0,\mathrm{old}}$ denote the probability vector before an accepted Newton step and let $P^{0,\mathrm{new}}$ denote the resulting candidate. If $\min_{j\in\mathcal J}P_j^{0,\mathrm{new}}\leq\tau_{\mathrm{switch}}$, the algorithm first evaluates the numerical KKT criterion at $P^{0,\mathrm{new}}$. If the criterion is satisfied, the candidate is retained as the inner solution. Otherwise, define
\[
\mathcal R(P^{0,\mathrm{new}})
:=
\left\{
j\in\mathcal J:
P_j^{0,\mathrm{new}}\leq\tau_{\mathrm{switch}}
\ \text{and}\
\widehat m_j(P^{0,\mathrm{new}};v_0)>1+\tau_{\mathrm{KKT}}
\right\}.
\]

\item If $\mathcal R(P^{0,\mathrm{new}})$ is empty, BA starts directly from $P^{0,\mathrm{new}}$. If it is nonempty, the algorithm shortens the accepted Newton step by setting
\[
\rho^*
=
\min_{j\in\mathcal R(P^{0,\mathrm{new}})}
\frac{
P_j^{0,\mathrm{old}}-\tau_{\mathrm{switch}}
}{
P_j^{0,\mathrm{old}}-P_j^{0,\mathrm{new}}
}
\]
and defining
\[
P^{0,\mathrm{switch}}
=
P^{0,\mathrm{old}}
+
\rho^*
\left(
P^{0,\mathrm{new}}-P^{0,\mathrm{old}}
\right).
\]
BA then starts from $P^{0,\mathrm{switch}}$ and continues until the numerical KKT criterion is satisfied, with no return to Newton--Raphson.
\end{enumerate}

The set $\mathcal R(P^{0,\mathrm{new}})$ contains the small probabilities that the KKT ratios indicate should increase. The factor $0<\rho^*\leq1$ retains the largest feasible fraction of the accepted Newton displacement while preventing any such probability from falling below $\tau_{\mathrm{switch}}$. The resulting choice of $\rho^*$ places at least one relevant probability at $\tau_{\mathrm{switch}}$ and keeps the others at or above it. Because $P^{0,\mathrm{switch}}$ is a convex combination of $P^{0,\mathrm{old}}$ and $P^{0,\mathrm{new}}$, it remains in the probability simplex. The adjustment shortens the Newton step without truncating or renormalizing probabilities and provides a more efficient starting value for the multiplicative BA iterations. We set $\tau_{\mathrm{switch}}=10^{-4}$, $\tau_{\mathrm{KKT}}=10^{-8}$, and $\tau_{\mathrm{simplex}}=10^{-10}$.

For this example, we display the full probability vectors in alternative order $(0,1,2)$. Suppose an accepted Newton step moves from $P^{0,\mathrm{old}}=(0.200000,\allowbreak 0.799800,\allowbreak 0.000200)'$ to $P^{0,\mathrm{new}}=(0.200090,\allowbreak 0.799909,\allowbreak 0.000001)'$. The probability of alternative $2$ has fallen from $2\times10^{-4}$ to $10^{-6}$, well below $\tau_{\mathrm{switch}}=\allowbreak 10^{-4}$. Suppose its KKT ratio is $\widehat m_2(P^{0,\mathrm{new}};v_0)=\allowbreak 1.05$. Because the BA update is $P_2^{0,(k+1)}=\allowbreak P_2^{0,(k)}\widehat m_2(P^{0,(k)};v_0)$, the next BA update would increase this probability. Starting from $10^{-6}$ would require at least $\lceil\ln(10^{-4}/10^{-6})/\allowbreak\ln(1.05)\rceil=\allowbreak 95$ BA iterations merely to return the probability to $10^{-4}$ if the KKT ratio remained equal to $1.05$. This illustrates why a small probability with a KKT ratio above one can be an inefficient BA starting value. The shortening rule gives $\rho^*\approx\allowbreak 0.503$, so the algorithm retains approximately $50.3$ percent of the accepted Newton displacement and initializes BA at $P^{0,\mathrm{switch}}\approx\allowbreak(0.200045,\allowbreak 0.799855,\allowbreak 0.000100)'$. If instead $\widehat m_2(P^{0,\mathrm{new}};v_0)\leq\allowbreak 1+\tau_{\mathrm{KKT}}$, the switching rule starts BA directly from the Newton candidate.

\subsubsection{Inner-problem comparison}

We compare BA and Newton--BA using the same $60{,}000$ inner problems for each value of $\lambda$ considered in Appendix~\ref{app:inner-solver-performance}. Both methods start from the uniform probability vector. We set $\tau_{\mathrm{KKT}}=10^{-8}$ and $\tau_{\mathrm{simplex}}=10^{-10}$ to match the numerical tolerances used later in the full simulated-likelihood implementation. The $R_1$ draws are not used because they do not enter the inner problem. The reported Newton--BA computation time includes all Newton iterations, any KKT-guided shortening of the Newton step, and all BA iterations following a transition.

Table~\ref{tab:newton-ba-comparison} shows that BA and Newton--BA return numerically similar solutions. The mean largest componentwise probability difference does not exceed $1.11\times10^{-7}$, and the mean absolute objective difference does not exceed $1.47\times10^{-9}$. Relative to BA, Newton--BA reduces mean computation time by approximately $35.3$ percent at $\lambda=1$, $49.8$ percent at $\lambda=1.5$, and $48.3$ percent at $\lambda=2$.

\begin{table}[!htbp]
\centering
\caption{Performance of BA and KKT-guided Newton--BA}
\label{tab:newton-ba-comparison}
\scriptsize
\setlength{\tabcolsep}{4pt}
\renewcommand{\arraystretch}{1.08}
\begin{tabular}{@{}lrrrrrr@{}}
\toprule
& \multicolumn{2}{c}{Mean time (ms)}
& Transition to BA
& Newton--BA faster
& Mean $\max_j|\Delta P_j^0|$
& Mean $|\Delta\widehat Q^0|$ \\
\cmidrule(lr){2-3}
$\lambda$
& BA
& Newton--BA
& (\%)
& (\%)
&
& \\
\midrule
1.0 & 0.8395 & 0.5434 & 21.50 & 86.15 & $3.328\times10^{-8}$ & $3.378\times10^{-10}$ \\
1.5 & 1.6173 & 0.8119 & 44.86 & 74.60 & $6.875\times10^{-8}$ & $9.850\times10^{-10}$ \\
2.0 & 2.5195 & 1.3014 & 62.28 & 61.95 & $1.109\times10^{-7}$ & $1.467\times10^{-9}$ \\
\bottomrule
\end{tabular}
\par\vspace{0.7em}
\begin{minipage}{0.92\textwidth}
\footnotesize
\emph{Notes:} Times are milliseconds per inner solve. Transition to BA is the share of Newton--BA cases requiring BA, while Newton--BA faster is its share of paired timing wins. The final columns give the mean maximum absolute probability difference and mean absolute objective difference across paired solves. Both methods pass the simplex-feasibility and KKT checks in every case. Elapsed times may vary with machine load; all comparisons use the same problems evaluated within a common run.
\end{minipage}
\end{table}

\begin{samepage}
As $\lambda$ increases, the share of Newton--BA problems transitioning to BA rises from $21.50$ percent to $62.28$ percent, while the share for which Newton--BA is faster falls from $86.15$ percent to $61.95$ percent. Thus, the hybrid method is more likely to require BA as near-boundary solutions become more prevalent, but remains faster on average at every $\lambda$ tested. These results do not establish its advantage in more extreme near-boundary designs.

\end{samepage}

These results favour Newton--BA for solving the inner problem. A lower inner-solve time does not necessarily imply a proportional reduction in total estimation time because the outer optimizer may follow different evaluation paths under the two methods. We therefore compare BA and Newton--BA within the complete simulated-likelihood problem next.

\FloatBarrier

\subsubsection{Full simulated-likelihood comparison}

We use $B_{\mathrm{MLE}}=3$ paired replications with $N=100$, $R_0=30$, $R_1=400$, $S=800$, and $\lambda\in\{1,1.5,2\}$. Within each replication and value of $\lambda$, the two methods use the same simulated dataset, estimation draws, initial parameter value, and outer-optimization settings. The base seed is $20260619$. The DGP and estimation draws are generated from separate random-number streams, and the prior-state draws are individual-specific. Both methods use $\tau_{\mathrm{KKT}}=10^{-8}$ and $\tau_{\mathrm{simplex}}=10^{-10}$, with no fallback from either method to the other.

The outer quasi-Newton optimizer uses central finite differences with step parameter $10^{-2}$. The likelihood and prior-state draws are generated once within each replication and held fixed across evaluations of $\theta$. An untimed likelihood evaluation is performed before timing each method, and the order of the two methods is rotated across paired comparisons. Reported times cover the outer optimization and exclude the untimed evaluation and post-estimation checks. After each run, we reevaluate the final estimate and compute a fresh central-difference gradient for the strict-validity and stationarity checks. These checks require a finite estimate, a clean inner-solver path, a reproducible objective value, and a maximum absolute gradient no greater than $10^{-3}$. This exercise evaluates the computational performance of the two inner algorithms rather than the sampling performance of the estimator.

\begin{table}[!htbp]
\centering
\caption{Full simulated-likelihood comparison of BA and Newton--BA}
\label{tab:newton-ba-full-mle}
\begin{minipage}{0.86\textwidth}
\scriptsize
\setlength{\tabcolsep}{4pt}
\renewcommand{\arraystretch}{1.08}
\begin{tabular*}{\linewidth}{@{\extracolsep{\fill}}l
S[table-format=3.3]
S[table-format=3.3]
S[table-format=2.1]rr@{}}
\toprule
& \multicolumn{2}{c}{Mean time (s)}
& \multicolumn{1}{c}{Time reduction}
& \multicolumn{1}{c}{Mean parameter}
& \multicolumn{1}{c@{}}{Mean same-$\theta$} \\
\cmidrule(lr){2-3}
$\lambda$
& \multicolumn{1}{c}{BA}
& \multicolumn{1}{c}{Newton--BA}
& \multicolumn{1}{c}{(\%)}
& \multicolumn{1}{c}{difference}
& \multicolumn{1}{c@{}}{NLL gap} \\
\midrule
1.0 & 71.997 & 42.866 & 40.5 & $7.763\times10^{-6}$ & $1.643\times10^{-8}$ \\
1.5 & 293.016 & 149.471 & 49.0 & $5.271\times10^{-5}$ & $3.190\times10^{-8}$ \\
2.0 & 350.190 & 181.640 & 48.1 & $6.576\times10^{-4}$ & $3.821\times10^{-8}$ \\
\bottomrule
\end{tabular*}
\par\vspace{0.4em}
\footnotesize
\emph{Notes:} Means are across three paired replications per $\lambda$. Times cover outer optimization, and reductions are relative to BA. NLL is the average negative log likelihood. For each pair, parameter difference is the maximum absolute coefficient difference; the same-$\theta$ NLL gap averages the absolute differences between algorithms at the two fitted coefficient vectors.
\end{minipage}
\end{table}

Table~\ref{tab:newton-ba-full-mle} reports mean estimation times across three paired replications per value of $\lambda$. Newton--BA reduces mean estimation time by $40.5$ percent at $\lambda=1$, $49.0$ percent at $\lambda=1.5$, and $48.1$ percent at $\lambda=2$. It is also faster in each of the nine underlying paired runs. The mean numbers of outer likelihood evaluations are the same for both methods: $102$, $219$, and $168$ at $\lambda=1$, $1.5$, and $2$, respectively. The time reductions therefore primarily reflect faster likelihood evaluations rather than fewer outer evaluations. Along the Newton--BA estimation paths, the mean share of inner problems transitioning to BA rises from $18.7$ percent at $\lambda=1$ to $53.5$ percent at $\lambda=2$, but the computational advantage remains substantial.

All 18 method runs pass the strict-validity check: the parameter estimates and NLL values are finite, \texttt{fminunc} returns a positive exit flag, fresh likelihood evaluations reproduce the final objective values, and every inner solution passes the KKT and simplex-feasibility checks without solver failure. Both methods terminate with exit flag 1 in seven paired comparisons and exit flag 5 in the remaining two, both at $\lambda=1.5$. Because exit flag 5 does not establish stationarity, we also apply the post-estimation stationarity check: we recompute the central-difference gradient at each final estimate and require its largest absolute component to be no greater than $\tau_{\mathrm{stat}}=10^{-3}$. The largest such component across all runs is $2.94\times10^{-4}$, so all runs pass the stationarity check and are retained.

The two methods also produce closely aligned likelihoods and parameter estimates. For each paired run, we average the absolute same-$\theta$ NLL gaps evaluated at the two fitted coefficient vectors. Across the three values of $\lambda$, the largest mean of this measure in Table~\ref{tab:newton-ba-full-mle} is $3.821\times10^{-8}$, while its maximum across individual paired runs is $6.805\times10^{-8}$. The mean largest componentwise parameter difference does not exceed $6.58\times10^{-4}$, while the largest difference in any paired comparison is $1.88\times10^{-3}$. These results indicate that the computational gain does not come at the cost of materially different likelihood values or parameter estimates.

The inner-problem and full simulated-likelihood comparisons therefore both favor Newton--BA. We use Newton--BA as the inner solver in our implementation, retaining the one-way transition to BA whenever the interior Newton system is no longer appropriate.

\FloatBarrier

\section{Additional Results for Section~\ref{sec:simulation-results}}
\label{app:simulation-results}

This appendix provides the covariance-estimator comparison discussed in Section~\ref{subsec:mc-inference} and the draw-sensitivity checks summarized in Section~\ref{subsec:mc-draw-sensitivity}.

\subsection{Sandwich-inference robustness check}
\label{app:mc-sandwich}

We compare information-based and sandwich inference as a robustness check on the choice of covariance estimator. Let $\widehat s_i(\widehat\theta_N^{\mathrm{sim}}):=\nabla_\theta\widehat\ell_i(\widehat\theta_N^{\mathrm{sim}})$ and define
$\widehat{\mathcal B}_{\theta,N}:=N^{-1}\sum_{i=1}^N
\widehat s_i(\widehat\theta_N^{\mathrm{sim}})\widehat s_i(\widehat\theta_N^{\mathrm{sim}})^{\prime}$.
The sandwich covariance estimator is
\[
\widehat{\operatorname{Var}}^{\mathrm{sand}}(\widehat\theta_N^{\mathrm{sim}})
=
\frac{1}{N}
\widehat{\mathcal I}_{\theta,N}^{-1}
\widehat{\mathcal B}_{\theta,N}
\widehat{\mathcal I}_{\theta,N}^{-1}.
\]

We construct nominal $95$ percent marginal Wald confidence intervals using the corresponding sandwich standard errors. Table~\ref{tab:mc-sandwich} compares their empirical coverage with that of the information-based intervals reported in the main text.

Coverage is generally closer to the nominal level for the information-based intervals. The difference is most pronounced for $\delta$, favoring information-based inference in this design.

\begin{table}[!htbp]
\centering
\caption{Coverage of information-based and sandwich confidence intervals}
\label{tab:mc-sandwich}
\footnotesize
\setlength{\tabcolsep}{5pt}
\renewcommand{\arraystretch}{1.05}
\begin{tabular}{@{}lrrrrrr@{}}
\toprule
& \multicolumn{2}{c}{$N=50$}
& \multicolumn{2}{c}{$N=100$}
& \multicolumn{2}{c}{$N=200$} \\
\cmidrule(lr){2-3}
\cmidrule(lr){4-5}
\cmidrule(lr){6-7}
Parameter
& Info. & Sand.
& Info. & Sand.
& Info. & Sand. \\
\midrule
$\alpha$ & 95.35 & 92.93 & 94.89 & 93.15 & 94.99 & 94.48 \\
$\beta$  & 96.46 & 92.53 & 94.38 & 92.94 & 92.95 & 92.13 \\
$\gamma$ & 94.44 & 93.94 & 94.99 & 94.48 & 95.40 & 95.20 \\
$\delta$ & 92.53 & 83.74 & 92.23 & 88.96 & 93.05 & 90.91 \\
\bottomrule
\end{tabular}
\par\smallskip
\begin{minipage}{0.88\textwidth}
\footnotesize
\emph{Notes:} Entries are empirical coverage rates, in percent, for nominal $95$ percent marginal Wald confidence intervals. Both methods use the same valid-inference subsets as the main analysis.
\end{minipage}
\end{table}

\afterpage{\FloatBarrier}

\subsection{Inference sensitivity to simulation draws}
\label{app:mc-draw-diagnostics}

We examine the sensitivity of information-based standard errors and coverage to the number of simulation draws using paired re-estimation at $N=200$. We separately double prior-state draws from $S=800$ to $S=1{,}600$, $\varepsilon_0$ likelihood draws from $R_0=30$ to $R_0=60$, and $\varepsilon_1$ likelihood draws from $R_1=400$ to $R_1=800$.

All three comparisons use the same $50$ selected Monte Carlo replications. Within each pair, we keep the dataset and original draws, add independent draws only for the component being examined, and leave the other draws and solver settings unchanged. Re-estimation starts from the saved baseline estimate, after which we recompute information-based inference.

\begin{table}[!htbp]
\centering
\caption{Paired finite-draw diagnostics at $N=200$}
\label{tab:mc-draw-diagnostics}
\footnotesize
\setlength{\tabcolsep}{8pt}
\renewcommand{\arraystretch}{1.08}
\begin{tabular}{@{}llrrrr@{}}
\toprule
Draw increase
& Statistic
& $\alpha$
& $\beta$
& $\gamma$
& $\delta$ \\
\midrule
$S:800\rightarrow1{,}600$
& Mean SE ratio
& 1.010 & 1.013 & 0.997 & 1.008 \\
& Median SE ratio
& 1.006 & 1.009 & 0.997 & 1.009 \\
\addlinespace
$R_0:30\rightarrow60$
& Mean SE ratio
& 1.009 & 1.042 & 1.009 & 1.083 \\
& Median SE ratio
& 1.000 & 1.000 & 0.999 & 1.006 \\
\addlinespace
$R_1:400\rightarrow800$
& Mean SE ratio
& 1.001 & 1.004 & 1.003 & 1.020 \\
& Median SE ratio
& 0.998 & 1.000 & 0.999 & 1.001 \\
\bottomrule
\end{tabular}
\par\smallskip
\begin{minipage}{0.90\textwidth}
\footnotesize
\emph{Notes:} Entries summarize within-replication ratios of augmented to baseline information-based standard errors. From the same 50 selected replications, each diagnostic retains pairs with valid estimates and inference at both numbers of draws: 48 for $S$, 49 for $R_0$, and 49 for $R_1$.
\end{minipage}
\end{table}

Table~\ref{tab:mc-draw-diagnostics} reports the mean and median of the paired SE ratios. A ratio of one indicates no change. The coverage and timing comparisons below use the same valid-pair subset as the table for each diagnostic.

Table~\ref{tab:mc-draw-diagnostics} shows little change in typical standard errors across all three draw increases. Doubling $S$ gives median SE ratios from $0.997$ to $1.009$. Increasing $R_0$ likewise gives median ratios close to one. Although its mean ratios for $\beta$ and $\delta$ exceed one, their corresponding medians are $1.000$ and $1.006$, with the mean increases concentrated in a small number of replications. Doubling $R_1$ gives median SE ratios from $0.998$ to $1.001$ and mean ratios from $1.001$ to $1.020$.

The corresponding paired coverage calculations, not displayed in Table~\ref{tab:mc-draw-diagnostics}, show unchanged coverage for all four parameters when either $R_0$ or $R_1$ is doubled. Doubling $S$ leaves coverage unchanged for $\alpha$, $\beta$, and $\delta$; coverage for $\gamma$ rises from $97.92$ to $100$ percent, reflecting one additional covered interval among 48 valid pairs. Coverage for $\delta$ remains $87.50$ percent in the $S$ comparison and $85.71$ percent in the $R_0$ and $R_1$ comparisons. These rates describe the selected paired subsets, not the full Monte Carlo sample in Table~\ref{tab:mc-inference-performance}.

Over the same paired subsets, mean outer-optimization time increases by approximately $17$ percent for $S$ and $56$ percent for $R_0$, but decreases by $20$ percent for $R_1$. Mean numerical-inference time increases by approximately $58$, $125$, and $9$ percent in the $S$, $R_0$, and $R_1$ comparisons, respectively. These timings are descriptive, not controlled estimates of the cost of increasing draws: baseline and augmented timings come from separate runs, and augmented searches start from the saved baseline estimates.

Overall, typical standard errors and paired coverage show limited sensitivity to the tested draw increases in these selected replications.

\afterpage{\FloatBarrier}

\subsection{Likelihood sensitivity at the true parameters}
\label{app:mc-r0-likelihood-truth}

We assess the sensitivity of the simulated criterion $\widehat L_N(\theta_0)$ to each draw size. We generate $B=100$ independent datasets of $N=200$ DMs under the design in Section~\ref{sec:monte-carlo-design}, with $\theta_0=(0.5,0.5,0.5,0.5)^{\prime}$ and $\lambda=1$. Starting from $(R_0,R_1,S)=(30,400,800)$, we double $R_0$ to $60$, $R_1$ to $800$, or $S$ to $1{,}600$, changing one draw size at a time. All evaluations use the true parameters, without re-estimation.

Within each replication, all four evaluations use the same choices and covariates. Each augmented calculation retains the original draws, adds independent draws for the selected dimension, and leaves the other dimensions unchanged. Doubling $R_1$ reuses the baseline inner solutions; doubling $S$ re-solves the inner problems using all $1{,}600$ prior draws. Likelihood probabilities are averaged before taking logarithms. Replications use new DGP and likelihood draws, with separate seeds for choice generation and likelihood evaluation. Choice generation uses BA, while likelihood evaluation uses KKT-guided Newton--BA. All inner solutions pass the KKT and simplex-feasibility checks at tolerances $10^{-8}$ and $10^{-10}$, respectively. All $B=100$ replications enter the reported statistics.

For each draw-size comparison, define $\Delta_b$ as the augmented criterion minus the baseline criterion in replication $b$. Table~\ref{tab:mc-r0-likelihood-truth} reports its mean and root mean square (RMS), so that variation across paired differences is not hidden by averaging.

\begin{table}[!htbp]
\centering
\caption{Simulated log likelihood at the true parameters}
\label{tab:mc-r0-likelihood-truth}
\footnotesize
\setlength{\tabcolsep}{5pt}
\renewcommand{\arraystretch}{1.08}
\begin{minipage}{0.86\linewidth}
\centering
\begin{tabular*}{\linewidth}{@{\extracolsep{\fill}}lrrrrr@{}}
\toprule
& \multicolumn{2}{c}{Log likelihood}
& \multicolumn{3}{c}{Paired change $\Delta_b$} \\
\cmidrule(lr){2-3}\cmidrule(l){4-6}
Draws & Mean & SD & Mean & MC SE & RMS \\
\midrule
Baseline     & $-0.838335$ & $0.035820$ & --- & --- & --- \\
Double $R_0$ & $-0.837859$ & $0.036021$ & $0.000476$ & $0.000290$ & $0.002924$ \\
Double $R_1$ & $-0.837896$ & $0.035513$ & $0.000438$ & $0.000189$ & $0.001934$ \\
Double $S$   & $-0.837148$ & $0.035606$ & $0.001187$ & $0.000333$ & $0.003517$ \\
\bottomrule
\end{tabular*}
\par\vspace{0.35em}
\raggedright
\emph{Notes:} The baseline uses $(R_0,R_1,S)=(30,400,800)$. Each criterion is averaged over $N=200$ DMs. SD is its standard deviation across all $B=100$ replications. For the paired changes, MC SE is $\mathrm{SD}(\Delta_b)/\sqrt{B}$ and RMS is $(B^{-1}\sum_{b=1}^{B}\Delta_b^2)^{1/2}$. MC SE describes uncertainty in the mean change, not a parameter standard error.
\end{minipage}
\end{table}

All three doublings produce small mean increases in the simulated criterion. Doubling $S$ gives the largest mean and RMS paired changes. The across-dataset SDs change little, but they include sampling variation and do not measure simulation error alone.

To isolate variability from the $\varepsilon_0$ draws, we also evaluate the simulated criterion separately on the original and additional blocks of $30$ draws. Denote these evaluations by $\widehat L_{N,b}^{[1]}(\theta_0)$ and $\widehat L_{N,b}^{[2]}(\theta_0)$. Their scaled root mean squared difference is
\[
\left\{
\frac{1}{2B}\sum_{b=1}^{B}
\left[
\widehat L_{N,b}^{[1]}(\theta_0)
-
\widehat L_{N,b}^{[2]}(\theta_0)
\right]^2
\right\}^{1/2}
=0.004078.
\]
Conditional on each dataset and its $R_1$ and $S$ draws, the two blocks are independent and identically distributed. The factor $1/2$ accounts for both blocks contributing to the variance of the difference. This measure summarizes variability from the $R_0$ draws at $R_0=30$, rather than overall variation across datasets.

Overall, the simulated criterion at $\theta_0$ shows limited sensitivity to the tested increases in $R_0$, $R_1$, and $S$. This check concerns log-likelihood values, while Appendix~\ref{app:mc-draw-diagnostics} examines standard errors and coverage after re-estimation.

\section{Calibrated Experiment: Supplementary Details}
\label{app:modecanada-experiment}
\label{app:modecanada-comparison-numerics}

This appendix documents the calibration, estimation, and numerical checks for Section~\ref{sec:modecanada-experiment}. After the calibration and implementation details, the comparisons follow the main text: likelihood fit, coefficient recovery, baseline CCP recovery, and travel-time responses.

\subsection{Preliminary MNL and DGP coefficients}
\label{app:modecanada-preliminary}

Let $y_i^{\mathrm{orig}}$ denote the choice recorded for DM $i$ in the original ModeCanada data. With car as the reference alternative and the type-I extreme-value scale normalized to one, the preliminary MNL specifies
\[
\widetilde u_{0,i0}=0,
\qquad
\widetilde u_{0,ij}
=
\widetilde c_j
+\widetilde\alpha^{\prime}x_{0,ij}
+\widetilde\beta^{\prime}w_{0,i},
\qquad j=1,2.
\]
All vectors are treated as columns: $\widetilde c\in\mathbb R^2$, while $x_{0,ij},w_{0,i},\widetilde\alpha,\widetilde\beta\in\mathbb R^4$. The fitted probabilities are
\[
\widetilde p_{i0}
=
\frac{1}{1+\sum_{\ell=1}^{2}\exp(\widetilde u_{0,i\ell})},
\qquad
\widetilde p_{ij}
=
\frac{\exp(\widetilde u_{0,ij})}{1+\sum_{\ell=1}^{2}\exp(\widetilde u_{0,i\ell})},
\quad j=1,2.
\]

Table~\ref{tab:modecanada-preliminary-mnl} reports the preliminary estimates. The regressors are constructed and scaled as described in Section~\ref{sec:modecanada-construction}.

\begin{table}[!htbp]
\centering
\caption{Preliminary MNL estimates}
\label{tab:modecanada-preliminary-mnl}
\footnotesize
\setlength{\tabcolsep}{8pt}
\renewcommand{\arraystretch}{1.05}
\begin{tabular}{@{}llr@{}}
\toprule
Parameter & Regressor & Estimate \\
\midrule
$\widetilde c_1$                         & Train indicator                             & $ 1.5518$ \\
$\widetilde c_2$                         & Air indicator                               & $ 0.2884$ \\
$\widetilde\alpha_{\mathrm{cost}}$       & Car-relative cost difference                & $-0.7027$ \\
$\widetilde\alpha_{\mathrm{ivt}}$        & Car-relative in-vehicle-time difference     & $-0.8771$ \\
$\widetilde\alpha_{\mathrm{ovt}}$        & Car-relative out-of-vehicle-time difference & $-4.4374$ \\
$\widetilde\alpha_{\mathrm{freq}}$       & Car-relative service-frequency difference   & $ 0.7316$ \\
$\widetilde\beta_{\mathrm{inc}}$         & Standardized log income                     & $ 0.0104$ \\
$\widetilde\beta_{\mathrm{dist}}$        & Standardized log distance                   & $ 0.8779$ \\
$\widetilde\beta_{\mathrm{urb1}}$        & Urban category $1$ indicator                & $ 1.0901$ \\
$\widetilde\beta_{\mathrm{urb2}}$        & Urban category $2$ indicator                & $ 1.2209$ \\
\bottomrule
\end{tabular}
\end{table}

The preliminary negative log likelihood is $1833.3818$. We set the true $V_0$ covariate coefficients in the experimental DGP to $\alpha_0=\widetilde\alpha$ and $\beta_0=\widetilde\beta$. To place the generated travel-time shocks on the same payoff scale as the observed travel-time variables, we also set
\[
\gamma_0
=
\widetilde\alpha_{\mathrm{ovt}}
=
-4.4374,
\qquad
\delta_0
=
-\widetilde\alpha_{\mathrm{ovt}}
=
4.4374.
\]
For a realization $d_i$ of the travel-time shocks, the resulting $V_1$ payoff index is
\[
u_{1,ij}(\vartheta_0)
=
4.4374\,\frac{d_{i0}-d_{ij}}{100},
\qquad j=1,2.
\]
The parameter $q$ governs the distribution from which $d_i$ is drawn, while $\gamma_0$ and $\delta_0$ determine how these shocks affect payoffs. Using the out-of-vehicle-time coefficient to scale total-time shocks is a design choice, not a separate estimate of their effect from ModeCanada. The equal and opposite values imply that a common travel-time shock to car and inside mode $j$ leaves their relative payoff unchanged. Because $\alpha_0,\beta_0\in\mathbb R^4$ and $\gamma_0,\delta_0\in\mathbb R$, the parameter vector $\theta_0=(\alpha_0^{\prime},\beta_0^{\prime},\gamma_0,\delta_0)^{\prime}$ belongs to $\mathbb R^{10}$.

We do not retain the preliminary constants $\widetilde c$. Once costly information acquisition and the generated travel-time shocks are introduced, these constants no longer reproduce the observed choice shares. For each candidate $(q,\lambda)$, we therefore adjust $c=(c_1,c_2)^{\prime}$ to match the observed choice shares under the corresponding rational-inattention specification.

\subsection{Calibration and validation}
\label{app:modecanada-calibration}

\paragraph{Candidate search.}

The observed choice shares are the sample frequencies calculated from
$\{y_i^{\mathrm{orig}}\}_{i=1}^{N}$. They are $0.457566$, $0.167208$, and $0.375226$ for car, train, and air, respectively, and serve as the calibration targets. The model-implied aggregate choice share of an alternative is obtained by averaging its computed integrated CCP over DMs and draws from $F_{\varepsilon_0}$.

Holding $\theta_0$ fixed at the values reported in Appendix~\ref{app:modecanada-preliminary}, we consider
\[
q\in\{0.10,0.15,0.20,0.25,0.30\},
\qquad
\lambda\in\{0.25,0.50,0.75,1.00,1.50,2.00,3.00\}.
\]
For each candidate $(q,\lambda)$, we adjust $c=(c_1,c_2)^{\prime}$ so that the model-implied aggregate choice shares match the observed choice shares. Thus, the shares alone do not identify $q$ and $\lambda$. Every candidate within a given search stage uses the same simulation draws, generated with fixed seeds.

The initial screen evaluates all $35$ candidates using $N=600$ DMs, $R_0=4$ draws from $F_{\varepsilon_0}$, and $S=512$ prior draws. Based on share matching and the numerical diagnostics, ten candidates are retained for refinement using all $N=2{,}769$ DMs, $R_0=8$, and $S=1{,}024$. Among the refined candidates passing the simplex-feasibility and KKT checks, we select a design that exhibits nontrivial information acquisition and substantial proportions of both numerically near-boundary inner solutions and solutions containing all three alternatives. This procedure selects
\[
q_0=0.20,
\qquad
\lambda_0=0.75.
\]

\paragraph{Numerical validation.}

After selecting $(q_0,\lambda_0)$, we generate draws with new seeds, using $R_0=128$ and $S=16{,}384$, to determine the final train and air constants:
\[
c_0=(1.026416,\,0.336371)^{\prime}.
\]
Thus, $c_0\in\mathbb R^2$, and the complete experimental payoff parameter
$\vartheta_0=(c_0^{\prime},\theta_0^{\prime})^{\prime}$
belongs to $\mathbb R^{12}$. We then hold $q_0$, $\lambda_0$, and $c_0$ fixed and validate the calibrated DGP using different seeds, with $R_0=128$ and $S=32{,}768$. The first $16{,}384$ prior draws from this evaluation are retained only for the draw-size stability check.

At $S=32{,}768$, the model-implied aggregate choice shares for car, train, and air are $0.457372$, $0.167850$, and $0.374779$, respectively. Their maximum absolute difference from the observed choice shares is $0.000641$, below the validation tolerance of $0.002$.

Following Appendix~\ref{app:inner-solver-performance}, we include alternative $j$ in the numerical consideration set when
$\widehat P_j^{0*}(v_0)>10^{-6}$. Under this rule, $10.11$, $40.03$, and $49.86$ percent of the computed inner solutions contain one, two, and three alternatives, respectively. Thus, $50.14$ percent are classified as numerically near-boundary, while the remaining $49.86$ percent include all three alternatives.

For each computed inner solution with positive entropy, we divide the Shannon mutual information
$\mathcal I(\widehat P^*;v_0)$
by the entropy
$-\sum_{j\in\mathcal J}\widehat P_j^{0*}(v_0)
\ln\widehat P_j^{0*}(v_0)$.
This ratio measures the fraction of ex ante choice uncertainty resolved through information acquisition. Its mean is $0.4881$, indicating an intermediate degree of information acquisition between the no- and perfect-information cases. Every computed inner solution also passes the simplex-feasibility and KKT checks with tolerances $10^{-10}$ and $10^{-8}$, respectively.

A draw-size check compares nested Sobol nets with $S=16{,}384$ and $S=32{,}768$. Doubling $S$ changes the model-implied aggregate choice shares by at most $1.56\times10^{-6}$, the numerical consideration-set proportions by at most $0.00222$, and the mean information ratio by $0.00101$. These changes are below the respective tolerances of $0.002$, $0.02$, and $0.02$, confirming that $c_0$ remains appropriate and supporting the use of $S=32{,}768$ in the final design.

\paragraph{Frozen outcome.}

After completing the numerical validation, we fix $(\vartheta_0,q_0,\lambda_0)$ and generate the experimental choices $Y_i$ using $S=32{,}768$ prior draws for each DM. Let $y_i$ denote the realized choice. All pseudorandom and scrambled Sobol draws use fixed, nonoverlapping seeds. The sample $\{y_i\}_{i=1}^{N}$ has realized car, train, and air shares of $0.465150$, $0.170820$, and $0.364030$, respectively. We draw this sample once and freeze it for the subsequent evaluation of the RI--ARUM and the two MNL benchmarks.

\subsection{Numerical implementation of RI estimation}
\label{app:modecanada-ri-implementation}

We minimize the average simulated NLL, $-N^{-1}\sum_{i=1}^{N}\ln\widehat\Phi_{iy_i}^{\mathrm{RI}}(\vartheta)$. Each NLL evaluation requires solving the inner problems and integrating over $\varepsilon_{0,i}$. A preliminary RI estimate with $R_0=64$ and $S=800$ provides the starting coefficients. The final search uses $R_0=128$ scrambled Sobol draws over $\varepsilon_{0,i}$ and $S=1{,}600$ pseudorandom Gaussian prior-state draws per DM. Fixed seeds keep both types of simulation draws unchanged within each search.

Preliminary \texttt{fminunc} runs for this 12-parameter problem did not pass the stationarity check. We therefore use Pattern Search rather than the \texttt{fminunc} routine used in the Monte Carlo exercise. Pattern Search relies on objective comparisons and does not require numerical gradients.

We check the final estimate using central-difference gradients computed at two step sizes. The maximum absolute gradient is $6.57\times10^{-4}$, below the tolerance of $10^{-3}$. The two gradient calculations differ by at most $5.10\times10^{-5}$, below the tolerance of $5\times10^{-4}$. At the reported estimate, all inner solutions pass the KKT and simplex-feasibility checks at tolerances $10^{-8}$ and $10^{-10}$, respectively. These checks support approximate stationarity of the fixed-draw NLL, not global optimality.

\subsection{Likelihood sensitivity to seeds and draw sizes}
\label{app:modecanada-draw-sensitivity}

We keep the estimated coefficients and generated sample fixed. At $R_0=128$ and $S=1{,}600$, we reevaluate RI using ten new seeds, none used for estimation. For each seed, we also double $R_0$, $S$, or both, retaining the original draws. Comparisons at fixed draw sizes measure seed sensitivity, while comparisons using the same seed measure the effect of increasing draws. NI and PI have analytic likelihoods.

\begin{table}[!htbp]
\centering
\caption{RI likelihood sensitivity to simulation seeds and draw sizes}
\label{tab:modecanada-integration-sensitivity}
\vspace{0.3em}
\footnotesize
\setlength{\tabcolsep}{7pt}
\renewcommand{\arraystretch}{1.05}
\begin{minipage}{0.84\linewidth}
\centering
\begin{tabular*}{\linewidth}{@{\extracolsep{\fill}}rrrrrr@{}}
\toprule
$R_0$ & $S$ & Mean NLL & Median NLL & SD & RI below PI \\
\midrule
$128$ & $1{,}600$ & $0.605639$ & $0.604577$ & $0.003244$ & $1/10$ \\
$256$ & $1{,}600$ & $0.604688$ & $0.604571$ & $0.002266$ & $3/10$ \\
\addlinespace[3pt]
$128$ & $3{,}200$ & $0.606115$ & $0.604427$ & $0.004798$ & $2/10$ \\
$256$ & $3{,}200$ & $0.604594$ & $0.604230$ & $0.002318$ & $4/10$ \\
\bottomrule
\end{tabular*}
\par\vspace{0.35em}
\raggedright
\emph{Notes:} Each row summarizes ten NLL values, one per seed, each averaged over the same $N=2{,}769$ choices. SD measures variation across seeds, not sampling uncertainty. The final column counts evaluations with RI NLL below PI's analytic NLL of $0.603527$.
\end{minipage}
\end{table}

Mean and median RI NLL exceed PI's NLL in every row of Table~\ref{tab:modecanada-integration-sensitivity}, unlike the estimation-seed comparison in Table~\ref{tab:modecanada-fit}. Doubling $R_0$ reduces variation across seeds. Doubling $S$ does not move mean NLL in a consistent direction and reverses the RI--PI ranking for one seed at each $R_0$. The likelihood ranking therefore remains sensitive to numerical integration. These checks do not establish accuracy relative to the exact likelihood or provide a formal model-comparison test.

\subsection{Coefficient comparisons}
\label{app:modecanada-parameter-comparison}

Table~\ref{tab:modecanada-coefficients} reports the fitted RI and PI payoff coefficients and their absolute deviations from the DGP values in the generated sample.

\begin{table}[!htbp]
\centering
\caption{DGP and fitted payoff coefficients}
\label{tab:modecanada-coefficients}
\vspace{0.3em}
\footnotesize
\setlength{\tabcolsep}{5pt}
\renewcommand{\arraystretch}{1.05}
\begin{minipage}{0.96\linewidth}
\centering
\begin{tabular*}{\linewidth}{@{\extracolsep{\fill}}l*{3}{S[table-format=-1.4]}*{2}{S[table-format=1.4]}@{}}
\toprule
Parameter & {DGP} & {RI} & {PI}
& \multicolumn{2}{c@{}}{Absolute deviation} \\
\cmidrule(l){5-6}
& & & & {RI} & {PI} \\
\midrule
$c_1$ (train)
& 1.0264 & 0.7083 & 0.9872 & 0.3181 & 0.0392 \\
$c_2$ (air)
& 0.3364 & -0.2081 & -0.3228 & 0.5445 & 0.6592 \\
\addlinespace[3pt]
$\alpha_1$ (cost)
& -0.7027 & -0.5240 & -0.3284 & 0.1787 & 0.3743 \\
$\alpha_2$ (in-vehicle time)
& -0.8771 & -0.8814 & -0.9086 & 0.0043 & 0.0316 \\
$\alpha_3$ (out-of-vehicle time)
& -4.4374 & -4.2002 & -4.2349 & 0.2372 & 0.2025 \\
$\alpha_4$ (frequency)
& 0.7316 & 0.6854 & 0.6629 & 0.0462 & 0.0687 \\
\addlinespace[3pt]
$\beta_1$ (log income)
& 0.0104 & 0.0752 & 0.0742 & 0.0648 & 0.0638 \\
$\beta_2$ (log distance)
& 0.8779 & 1.0238 & 0.9989 & 0.1459 & 0.1209 \\
$\beta_3$ (urban 1)
& 1.0901 & 1.2319 & 1.1955 & 0.1418 & 0.1054 \\
$\beta_4$ (urban 2)
& 1.2209 & 1.4950 & 1.4449 & 0.2741 & 0.2241 \\
\addlinespace[3pt]
$\gamma$
& -4.4374 & -4.3613 & -2.1565 & 0.0761 & 2.2809 \\
$\delta$
& 4.4374 & 4.4516 & 2.0773 & 0.0142 & 2.3601 \\
\bottomrule
\end{tabular*}
\par\vspace{0.35em}
\raggedright
\emph{Notes:} Absolute deviations are from the DGP coefficients. Both models use the same covariate scaling and unit taste-shock scale. RI holds $(q_0,\lambda_0)$ fixed. NI does not estimate $\gamma$ or $\delta$ and is omitted.
\end{minipage}
\end{table}

In this sample, RI recovers $\gamma$ and $\delta$ much more closely than PI, but neither model is closer for every coefficient. Because RI re-estimation is computationally intensive, we do not assess coefficient sensitivity to simulation seeds or the number of draws here. This could be examined in future work.

\afterpage{\FloatBarrier}

\subsection{Sensitivity of baseline CCP recovery}
\label{app:modecanada-probability-check}

We assess the sensitivity of baseline CCP recovery in Table~\ref{tab:modecanada-baseline-ccps}, keeping the generated sample, baseline travel times, and all fitted and DGP coefficients fixed. Fitted RI uses the ten new seeds from Appendix~\ref{app:modecanada-draw-sensitivity}, while the DGP reference uses ten separate seeds. For each paired evaluation, we start at $R_0=128$ and $S=1{,}600$, then double $R_0$, $S$, or both. Seeds remain matched across draw sizes, and the original draws are retained.

We calculate CCP RMSE separately for each paired evaluation before summarizing the ten values in Table~\ref{tab:modecanada-probability-sensitivity}. All three models are compared with the same DGP reference within each paired evaluation. NI and PI CCPs are analytic, but their RMSE varies with the simulated reference.

\begin{table}[!htbp]
\centering
\caption{Baseline CCP recovery across seeds and draw sizes}
\label{tab:modecanada-probability-sensitivity}
\footnotesize
\setlength{\tabcolsep}{4pt}
\renewcommand{\arraystretch}{1.12}
\begin{minipage}{0.90\linewidth}
\centering
\begin{tabular*}{\linewidth}{@{\extracolsep{\fill}}rrccc@{}}
\toprule
 & & \multicolumn{3}{c}{CCP RMSE: mean (SD)} \\
\cmidrule(lr){3-5}
$R_0$ & $S$ & NI & RI & PI \\
\midrule
$128$ & $1{,}600$ & $0.149439\;(0.000093)$ & $0.022997\;(0.000180)$ & $0.026942\;(0.000170)$ \\
$256$ & $1{,}600$ & $0.149404\;(0.000069)$ & $0.022456\;(0.000160)$ & $0.026716\;(0.000194)$ \\
$128$ & $3{,}200$ & $0.149512\;(0.000105)$ & $0.022426\;(0.000198)$ & $0.026642\;(0.000137)$ \\
$256$ & $3{,}200$ & $0.149476\;(0.000094)$ & $0.021874\;(0.000177)$ & $0.026414\;(0.000154)$ \\
\bottomrule
\end{tabular*}
\par\vspace{0.35em}
\raggedright
\emph{Notes:} Entries are means of ten seed-specific CCP RMSEs, with SDs in parentheses. Each RMSE is across the same $N=2{,}769$ DMs and three modes at the baseline travel times. SDs measure numerical variation at fixed data and coefficients, not sampling uncertainty.
\end{minipage}
\end{table}

Mean RI and PI CCP RMSE fall as either $R_0$ or $S$ increases. RI has lower baseline CCP RMSE than both MNL benchmarks for all ten seeds at every reported pair of $R_0$ and $S$.

\FloatBarrier

\subsection{Travel-time response sensitivity}
\label{app:modecanada-response-sensitivity}

We assess the numerical sensitivity of the responses in Table~\ref{tab:modecanada-substitution}, holding the generated sample, all fitted and DGP coefficients, and the prior $\pi_1$ fixed. Starting from $(R_0,S)=(128,1{,}600)$, we double $R_0$, $S$, or both. Each configuration uses ten paired evaluations: fitted RI uses the ten new seeds from Appendix~\ref{app:modecanada-draw-sensitivity}, while the DGP reference uses ten separate seeds.

Seeds remain matched across draw sizes, and the original draws are retained. For both RI and the DGP, the before and after calculations share draws, with information choices re-optimized after the travel-time reduction. No $R_1$ integration is needed because $\varepsilon_{1,i}=\mathbf 0_3$.

Table~\ref{tab:modecanada-response-sensitivity} follows the main response table: Panel A reports average share changes, Panel B the sources of train's gain, and Panel C response RMSE and IIA deviation. Each measure is computed separately for each seed before taking its ten-seed mean and SD. For response RMSE, all three fitted models use the same DGP reference in each paired evaluation.

\begingroup
\footnotesize
\setlength{\tabcolsep}{3pt}
\renewcommand{\arraystretch}{1.05}
\setlength{\LTleft}{\fill}
\setlength{\LTright}{\fill}
\setlength{\LTcapwidth}{\textwidth}
\newcommand{\responsesenscell}[2]{\shortstack[t]{$#1$\\{\scriptsize $(#2)$}}}
\begin{longtable}{@{}rrl*{4}{>{\centering\arraybackslash}p{0.15\linewidth}}@{}}
\caption{Travel-time response sensitivity across seeds and draw sizes}
\label{tab:modecanada-response-sensitivity}\\
\toprule
& & & \multicolumn{4}{c}{Ten-seed mean (SD)} \\
\cmidrule(lr){4-7}
$R_0$ & $S$ & Measure & DGP & NI & RI & PI \\
\midrule
\endfirsthead

\caption[]{Travel-time response sensitivity across seeds and draw sizes (continued)}\\
\toprule
& & & \multicolumn{4}{c}{Ten-seed mean (SD)} \\
\cmidrule(lr){4-7}
$R_0$ & $S$ & Measure & DGP & NI & RI & PI \\
\midrule
\endhead

\midrule
\multicolumn{7}{r@{}}{\emph{Continued on next page}}\\
\endfoot

\bottomrule
\endlastfoot

\multicolumn{7}{@{}l}{\emph{Panel A. Choice-share changes (percentage points)}}\\*
\addlinespace[3pt]
$128$ & $1{,}600$ & Car & \responsesenscell{-3.7834}{0.0060} & \responsesenscell{-4.8919}{0} & \responsesenscell{-3.7943}{0.0038} & \responsesenscell{-4.3681}{0} \\*
 &  & Train & \responsesenscell{8.7199}{0.0053} & \responsesenscell{11.3766}{0} & \responsesenscell{8.9641}{0.0063} & \responsesenscell{10.1464}{0} \\*
 &  & Air & \responsesenscell{-4.9365}{0.0039} & \responsesenscell{-6.4847}{0} & \responsesenscell{-5.1697}{0.0048} & \responsesenscell{-5.7783}{0} \\
\addlinespace[4pt]
$256$ & $1{,}600$ & Car & \responsesenscell{-3.7827}{0.0027} & \responsesenscell{-4.8919}{0} & \responsesenscell{-3.7925}{0.0031} & \responsesenscell{-4.3681}{0} \\*
 &  & Train & \responsesenscell{8.7194}{0.0041} & \responsesenscell{11.3766}{0} & \responsesenscell{8.9628}{0.0043} & \responsesenscell{10.1464}{0} \\*
 &  & Air & \responsesenscell{-4.9367}{0.0033} & \responsesenscell{-6.4847}{0} & \responsesenscell{-5.1703}{0.0033} & \responsesenscell{-5.7783}{0} \\
\addlinespace[4pt]
$128$ & $3{,}200$ & Car & \responsesenscell{-3.7833}{0.0055} & \responsesenscell{-4.8919}{0} & \responsesenscell{-3.7948}{0.0041} & \responsesenscell{-4.3681}{0} \\*
 &  & Train & \responsesenscell{8.7200}{0.0057} & \responsesenscell{11.3766}{0} & \responsesenscell{8.9643}{0.0065} & \responsesenscell{10.1464}{0} \\*
 &  & Air & \responsesenscell{-4.9367}{0.0029} & \responsesenscell{-6.4847}{0} & \responsesenscell{-5.1695}{0.0052} & \responsesenscell{-5.7783}{0} \\
\addlinespace[4pt]
$256$ & $3{,}200$ & Car & \responsesenscell{-3.7824}{0.0028} & \responsesenscell{-4.8919}{0} & \responsesenscell{-3.7932}{0.0024} & \responsesenscell{-4.3681}{0} \\*
 &  & Train & \responsesenscell{8.7192}{0.0046} & \responsesenscell{11.3766}{0} & \responsesenscell{8.9632}{0.0032} & \responsesenscell{10.1464}{0} \\*
 &  & Air & \responsesenscell{-4.9368}{0.0032} & \responsesenscell{-6.4847}{0} & \responsesenscell{-5.1700}{0.0033} & \responsesenscell{-5.7783}{0} \\
\midrule
\multicolumn{7}{@{}l}{\emph{Panel B. Sources of train's gain (\%)}}\\*
\addlinespace[3pt]
$128$ & $1{,}600$ & From car & \responsesenscell{43.3882}{0.0514} & \responsesenscell{42.9996}{0} & \responsesenscell{42.3282}{0.0325} & \responsesenscell{43.0511}{0} \\
$256$ & $1{,}600$ & From car & \responsesenscell{43.3824}{0.0246} & \responsesenscell{42.9996}{0} & \responsesenscell{42.3140}{0.0266} & \responsesenscell{43.0511}{0} \\
$128$ & $3{,}200$ & From car & \responsesenscell{43.3866}{0.0415} & \responsesenscell{42.9996}{0} & \responsesenscell{42.3321}{0.0362} & \responsesenscell{43.0511}{0} \\
$256$ & $3{,}200$ & From car & \responsesenscell{43.3799}{0.0222} & \responsesenscell{42.9996}{0} & \responsesenscell{42.3197}{0.0261} & \responsesenscell{43.0511}{0} \\
\midrule
\multicolumn{7}{@{}l}{\emph{Panel C. Response accuracy and substitution pattern}}\\*
\addlinespace[3pt]
$128$ & $1{,}600$ & Response RMSE & --- & \responsesenscell{0.050563}{0.000051} & \responsesenscell{0.007226}{0.000106} & \responsesenscell{0.016557}{0.000050} \\*
 &  & IIA deviation & \responsesenscell{0.078775}{0.011327} & \responsesenscell{0}{0} & \responsesenscell{0.088267}{0.011566} & \responsesenscell{0}{0} \\
\addlinespace[4pt]
$256$ & $1{,}600$ & Response RMSE & --- & \responsesenscell{0.050513}{0.000038} & \responsesenscell{0.006346}{0.000068} & \responsesenscell{0.016383}{0.000047} \\*
 &  & IIA deviation & \responsesenscell{0.043577}{0.005472} & \responsesenscell{0}{0} & \responsesenscell{0.049076}{0.005057} & \responsesenscell{0}{0} \\
\addlinespace[4pt]
$128$ & $3{,}200$ & Response RMSE & --- & \responsesenscell{0.050570}{0.000038} & \responsesenscell{0.007064}{0.000085} & \responsesenscell{0.016513}{0.000043} \\*
 &  & IIA deviation & \responsesenscell{0.075540}{0.008468} & \responsesenscell{0}{0} & \responsesenscell{0.087996}{0.010798} & \responsesenscell{0}{0} \\
\addlinespace[4pt]
$256$ & $3{,}200$ & Response RMSE & --- & \responsesenscell{0.050519}{0.000032} & \responsesenscell{0.006170}{0.000051} & \responsesenscell{0.016339}{0.000043} \\*
 &  & IIA deviation & \responsesenscell{0.042704}{0.004937} & \responsesenscell{0}{0} & \responsesenscell{0.046783}{0.005021} & \responsesenscell{0}{0} \\
\end{longtable}
\noindent \emph{Notes:} SDs reflect numerical variation across seeds, not sampling uncertainty. In Panel B, the contribution from air is $100\%$ minus the contribution from car. NI and PI share changes, gain contributions, and IIA deviations are analytic and unchanged across evaluations; their response RMSE varies with the simulated DGP reference. The dash marks the DGP reference, for which response RMSE is not reported. All inner solutions pass the KKT and simplex-feasibility checks at tolerances $10^{-8}$ and $10^{-10}$, respectively.\par
\endgroup

Panels A and B show little numerical variation in average share changes or the division of train's gain. Across all $40$ paired evaluations, RI's train-share increase is closer to the DGP than either benchmark's, while PI is closer than RI in the division of train's gain between car and air.

Panel C shows that RI has lower response RMSE than both MNL benchmarks in every paired evaluation. The magnitude of IIA deviation remains sensitive to integration draws: doubling $R_0$ substantially reduces it for both RI and the DGP.

\end{appendices}

\end{document}